\documentclass[acmsmall]{acmart}

\startPage{1}

\setcopyright{rightsretained}
\acmPrice{}
\acmDOI{10.1145/3622856}
\acmYear{2023}
\copyrightyear{2023}
\acmSubmissionID{oopslab23main-p463-p}
\acmJournal{PACMPL}
\acmVolume{7}
\acmNumber{OOPSLA2}
\acmArticle{280}
\acmMonth{10}
\received{2023-04-14}
\received[accepted]{2023-08-27}

\usepackage{balance}

\usepackage{amsmath,amsfonts,bm}

\def\eqref#1{equation~\ref{#1}}

\def\1{\bm{1}}

\DeclareMathAlphabet{\mathsfit}{\encodingdefault}{\sfdefault}{m}{sl}
\SetMathAlphabet{\mathsfit}{bold}{\encodingdefault}{\sfdefault}{bx}{n}

\DeclareMathOperator*{\argmin}{arg\,min}

\usepackage[setmargin=false,hide=true]{marginalia}

\usepackage{microtype}
\usepackage{graphicx}

\usepackage[labelsep=period,labelfont=it]{caption}
\usepackage{subcaption}
\usepackage{booktabs}
\usepackage{xr}
\usepackage{hyperref}
\usepackage{listings}
\usepackage{amsmath}
\usepackage{amsthm}
\usepackage{stmaryrd}
\usepackage{mleftright}
\usepackage{cleveref}
\usepackage{mathpartir}
\usepackage{thmtools}
\usepackage{thm-restate}
\usepackage{scalerel}
\usepackage{stackengine,wasysym}
\usepackage{xurl}
\usepackage[most]{tcolorbox}
\usepackage[T1]{fontenc}
\usepackage{xcolor}
\usepackage{soul}
\usepackage{colortbl}
\usepackage[normalem]{ulem}
\usepackage{makecell}
\usepackage{wrapfig}
\usepackage{multirow}
\usepackage{capt-of}
\usepackage[%
        all=normal,
        paragraphs=tight,
        wordspacing=tight,
]{savetrees}

\usepackage[group-separator={,}]{siunitx}

\DeclareGraphicsExtensions{.jpg,.pdf}

   \makeatletter
   \let\old@lstKV@SwitchCases\lstKV@SwitchCases
   \def\lstKV@SwitchCases#1#2#3{}
   \makeatother
   \usepackage{lstlinebgrd}
   \makeatletter
   \let\lstKV@SwitchCases\old@lstKV@SwitchCases

   \lst@Key{numbers}{none}{%
       \def\lst@PlaceNumber{\lst@linebgrd}%
       \lstKV@SwitchCases{#1}%
       {none:\\%
        left:\def\lst@PlaceNumber{\llap{\normalfont
                   \lst@numberstyle{\thelstnumber}\kern\lst@numbersep}\lst@linebgrd}\\%
        right:\def\lst@PlaceNumber{\rlap{\normalfont
                   \kern\linewidth \kern\lst@numbersep
                   \lst@numberstyle{\thelstnumber}}\lst@linebgrd}%
       }{\PackageError{Listings}{Numbers #1 unknown}\@ehc}}

  \newcount\bt@rangea
  \newcount\bt@rangeb

  \newcommand\btIfInRange[2]{%
      \global\let\bt@inrange\@secondoftwo%
      \edef\bt@rangelist{#2}%
      \foreach \range in \bt@rangelist {%
          \afterassignment\bt@getrangeb%
          \bt@rangea=0\range\relax%
          \pgfmathtruncatemacro\result{ ( #1 >= \bt@rangea) && (#1 <= \bt@rangeb) }%
          \ifnum\result=1\relax%
              \breakforeach%
              \global\let\bt@inrange\@firstoftwo%
          \fi%
      }%
      \bt@inrange%
  }
  \newcommand\bt@getrangeb{%
      \@ifnextchar\relax%
          {\bt@rangeb=\bt@rangea}%
          {\@getrangeb}%
  }
  \def\@getrangeb-#1\relax{%
      \ifx\relax#1\relax%
          \bt@rangeb=100000%
      \else%
          \bt@rangeb=#1\relax%
      \fi%
  }

 \newcommand{\btLstHL}[1]{%
    \btIfInRange{\value{lstnumber}}{#1}%
      {\color{morange!40}}%
      {\def\lst@linebgrdcmd##1##2##3{}}%
  }%
\makeatother

\DeclareMathAlphabet{\mathbbold}{U}{bbold}{m}{n}
\newcommand*{\boldone}{\mathbbold{1}}

\makeatletter
\AtBeginDocument{%
  \let\c@figure\c@lstlisting
  
  \let\ftype@lstlisting\ftype@figure %
}
\makeatother

\newcommand\reallywidetilde[1]{\ThisStyle{%
  \setbox0=\hbox{$\SavedStyle#1$}%
  \stackengine{-.1\LMpt}{$\SavedStyle#1$}{%
    \stretchto{\scaleto{\SavedStyle\mkern.2mu\AC}{.5150\wd0}}{.6\ht0}%
  }{O}{c}{F}{T}{S}%
}}

\newcommand\mhl[1]{\bgroup\markoverwith
        {\textcolor{#1!40}{\rule[-.5ex]{2pt}{2.5ex}}}\ULon}
\newcommand\mhlr[1]{\bgroup\markoverwith
        {\textcolor{#1!30}{\rule[-.5ex]{2pt}{2.5ex}}}\ULon}

\newcommand\moptimal{complexity-guided}
\newcommand\mOptimal{Complexity-guided}

\newcommand{\mtitle}{\lang{}: Complexity-Guided Data Sampling for Training Neural Surrogates of Programs}

\newcommand{\minput}{x}
\newcommand{\mInput}{\mathcal{X}}
\newcommand{\moutput}{y}
\newcommand{\mOutput}{\mathcal{Y}}

\newcommand{\msurrogate}{\hat{f}}
\newcommand{\mfunction}{f}
\newcommand{\mFunction}{\mInput{} \to \mOutput{}}
\newcommand{\mlearner}{\mathit{tr}}
\newcommand{\mdistribution}{D}
\newcommand{\mDistribution}{\mathcal{D}}
\newcommand{\mnat}{n}

\newcommand{\mNat}{\mathbb{N}}
\newcommand{\mReal}{\mathbb{R}}
\newcommand{\mloss}{\ell{}}

\newcommand{\mstratum}{s}

\newcommand{\mnstrata}{{c}}

\newcommand{\merror}{\epsilon}
\newcommand{\mprob}{\delta}
\newcommand{\mcomplexityrelation}{\zeta}
\newcommand{\mcomplexity}[1]{\mcomplexityrelation\mleft(#1\mright)}

\newcommand\bnf{\,\mathrm{::=}\,}
\newcommand\bor{\;|\;}
\newcommand\epislon\epsilon

\newcommand{\mdiststratum}[1]{\int_{x \in \mstratum{}_{#1}} \mdistribution{}\mleft(x\mright) {\rm d}x}

\DeclareMathOperator*{\expectation}{\mathbb{E}}
\DeclareMathOperator*{\probability}{P}
\DeclareMathOperator*{\indicator}{\boldone}

\newcommand*{\suchthat}{\;\ifnum\currentgrouptype=16 \middle\fi|\;}
\newcommand{\lang}{\textsc{Turaco}}

\declaretheorem[name=Lemma,numberwithin=section]{lemma}
\declaretheorem[name=Note,numberwithin=section]{note}

\lstdefinelanguage{imp}{
    language=C,
    morekeywords={choose, infer, observe, input}
}

\newcommand\tkeyword[1]{\texttt{\color{blue}#1}}

\input{listings-glsl.prf}
\input{listings-learnguage.prf}
\definecolor{morange}{RGB}{255,100,0}
\definecolor{mblue}{RGB}{86,180,233}
\definecolor{mgreen}{RGB}{0,158,115}
\definecolor{myellow}{RGB}{255,255,0}
\definecolor{mred}{RGB}{204,51,139}

\definecolor{mop}{RGB}{7, 0, 255}
\definecolor{mconst}{RGB}{113,113,113}
\definecolor{mreg}{RGB}{25,23,124}

\newcommand{\cstepsto}{\;\tilde{\Downarrow}\;}

\NewTCBListing{framedlisting}{O{}}
 {%
  boxrule=0.4pt,
  colback=white,
  colframe=black,
  sharp corners,
  listing only,
  listing options={#1},
  top=0pt,
  bottom=0pt,
  left=0pt,
  right=0pt,
 }

\makeatletter
\let\orig@lstnumber=\thelstnumber

\newcommand\lstresetnumber{\global\let\thelstnumber=\orig@lstnumber}
\makeatother

\makeatletter

\newcommand*\wthelper[2]{%
        \hbox{\dimen@\accentfontxheight#1%
                \accentfontxheight#11.3\dimen@
                $\m@th#1\widetilde{#2}$%
                \accentfontxheight#1\dimen@
        }%
}

\newcommand*\accentfontxheight[1]{%
        \fontdimen5\ifx#1\displaystyle
                \textfont
        \else\ifx#1\textstyle
                \textfont
        \else\ifx#1\scriptstyle
                \scriptfont
        \else
                \scriptscriptfont
        \fi\fi\fi3
}
\makeatother

\makeatletter
\providecommand\phantomcaption{\caption@refstepcounter\@captype}
\makeatother

\begin{document}
\title{\mtitle}

\author{Alex Renda}
\affiliation{
  \institution{MIT CSAIL}            %
  \country{USA}                    %
}
\email{renda@csail.mit.edu}          %

\author{Yi Ding}
\affiliation{
  \institution{Purdue University}           %
  \country{USA}                   %
}
\email{yiding@purdue.edu}         %

\author{Michael Carbin}
\affiliation{
  \institution{MIT CSAIL}           %
  \country{USA}                   %
}
\email{mcarbin@csail.mit.edu}         %

\begin{abstract}
Programmers and researchers are increasingly developing \emph{surrogates} of programs, models of a subset of the observable behavior of a given program, to solve a variety of software development challenges.
Programmers train surrogates from measurements of the behavior of a program on a dataset of input examples.
A key challenge of surrogate construction is determining what training data to use to train a surrogate of a given program.

We present a methodology for sampling datasets to train neural-network-based surrogates of programs.
We first characterize the proportion of data to sample from each region of a program's input space (corresponding to different execution paths of the program)
based on the complexity of learning a surrogate of the corresponding execution path.
We next provide a program analysis to determine the complexity of different paths in a program.
We evaluate these results on a range of real-world programs, demonstrating that complexity-guided sampling results in empirical improvements in accuracy.

\end{abstract}

\ccsdesc[500]{Software and its engineering~Automatic programming}
\ccsdesc[300]{Computing methodologies~Machine learning}
\ccsdesc[300]{Software and its engineering~Software evolution}
\keywords{programming languages, surrogate models, neural networks}  %

\maketitle

\section{Introduction}
Surrogate programming -- the act of replacing a program with a \emph{surrogate} model of its behavior -- has been increasingly leveraged to solve a variety of software development challenges~\citep{renda_programming_2021}.
For example, \citet{esmaeilzadeh_neural_2012} train neural networks to mimic existing numerical programs (a fast Fourier transform, a triangle intersection detection algorithm, a JPEG encoder, etc.) and execute these neural networks on specialized hardware, achieving an average speedup of $2.3\times$ across benchmarks at the cost of $10\%$ accuracy in these numerical programs.
\citet{kustowski_transfer_2019} train a neural network to mimic computer simulations of inertial confinement fusion, then fine-tune the neural network on a small amount of data from real experiments to improve the accuracy of the simulated predictions.
\citet{renda_difftune_2020} train a neural network to mimic a CPU simulator, then exploit the fact that this neural \NA{surrogate} is differentiable to optimize the CPU simulator's simulation parameters with gradient descent.

\subsection{Surrogate Programming}

Beyond these examples, surrogate programming has developed into a diverse set of techniques and applications across many areas of computing and science. \citet{renda_programming_2021} \NA{organize} the uses of neural network based surrogates of programs into three categories.
\vspace*{-0.12em}
\paragraph{Surrogate compilation.}
In \emph{surrogate compilation}, programmers develop a surrogate that replicates the behavior of a program to deploy to end-users in place of that program.
For example, in addition to \citet{esmaeilzadeh_neural_2012} (who use a surrogate to speed up numerical programs), \citet{mendis_thesis_2020} use a surrogate to speed up compiler autovectorization by replacing an integer linear program (ILP) solver with a surrogate.
Key benefits of surrogate compilation include the ability to execute the surrogate on different hardware and the ability to bound or to accelerate the execution time of the surrogate.

\vspace*{-0.12em}
\paragraph{Surrogate adaptation.}
In \emph{surrogate adaptation}, programmers first develop a surrogate of a program and then further train that surrogate on data from a different task.
For example, in addition to \citet{kustowski_transfer_2019} (who use a surrogate to improve the accuracy of inertial confinement fusion simulations),
\citet{tercan_transfer_2018} use this technique to improve the accuracy of computer simulations of plastic injection molding.
Key benefits of surrogate adaptation include that it makes it possible to alter the semantics of the program to perform a different task of interest and that it may be more data-efficient or result in higher accuracy than training a model from scratch for the task.

\vspace*{-0.12em}
\paragraph{Surrogate optimization.}
In \emph{surrogate optimization} programmers develop a surrogate of a program,
use gradient descent against the surrogate to identify program inputs that optimize a downstream objective,
then use the inputs for executing the original program.
In addition to \citet{renda_difftune_2020} (who use this technique to optimize CPU simulation parameters), \citet{tseng_hyperparameter_2019} use this technique to find parameters for camera pipelines that lead to the most photorealistic images.
\citet{she_neuzz_2019} use this technique to find inputs that trigger branches that may cause bugs in the original program.
The key benefit of surrogate optimization is that it can identify optimal inputs faster than optimizing directly against the program, due to the faster execution speed of the surrogate and that the surrogate is differentiable even when the original program is not (allowing the use of gradient descent).

\vspace*{-0.3em}
\subsection{Dataset Generation}

In each of these \NA{scenarios}, training a surrogate of a program requires measuring the behavior of the program on a dataset of input examples.
There are three common approaches to collecting this dataset.
The first is to use data instrumented from running the original program on a workload of interest~\citep{renda_difftune_2020,esmaeilzadeh_neural_2012}.
In the absence of an available workload, another is to uniformly sample (or sample using another manually defined distribution) from the input space of the program~\citep{tseng_hyperparameter_2019,kustowski_transfer_2019}.
The third is to use \emph{active learning}~\citep{settles_active_2009}, a class of online methods that iteratively query labels for the data points that are most useful (however defined) for further training the surrogate~\citep{ipek_exploring_2006,she_neuzz_2019,pestourie_active_2020}.

Each of these approaches face challenges on programs with different behaviors in different regions of the input space.
For example, \citet[Section IV.A]{renda_difftune_2020} identify a scenario in which an instrumented dataset does not exercise a set of control flow paths in the program enough times for the surrogate to learn the program's behavior along those paths, resulting in a surrogate that generates highly inaccurate predictions for inputs in the regions of the input space corresponding to those paths.

\vspace*{-0.3em}
\subsection{Our Approach: Complexity-Guided Sampling}

Rather than treating the program as a black box, our approach uses the source code and semantics of the program under study to guide dataset generation for training a surrogate of the program.
The core concept is to allocate samples based on both the the \emph{complexity} of learning the program's behavior on a given path and the frequency of that path in the input data distribution.
\paragraph{\mOptimal{} sampling.}
Our objective is to find how many samples to allocate to each region of the input space to minimize the expected error of the resulting surrogate.
To reason about the error of a surrogate, we use neural network sample complexity bounds for learning analytic functions~\citep{arora_fine-grained_2019,agarwala_monolithic_2021}.
These bounds give an upper bound on how many samples are required to learn a surrogate of an analytic function to a given error as a function of a \emph{complexity measure} of that function.
Our approach calculates a complexity measure for the function induced by each control flow path in the program and combines that with the frequency of each path according to an input data distribution.
The output of the approach is the proportion of samples to allocate to each region of the input space, minimizing an upper bound on the stratified~surrogate's~error.

\vspace*{-0.15em}
\paragraph{Stratified functions.}
Our core modeling assumption is to
represent the program as a \emph{stratified function}, a piecewise\footnote{We choose the term \emph{stratified} by analogy with the technique of stratified sampling.} function across different regions (\emph{strata}) of the input space.
We use \emph{stratified surrogates} to model such functions.
To construct a stratified surrogate, we train independent surrogates of each component of the stratified function.
During evaluation, a stratified surrogate uses the original program to check which stratum an input is in, then applies the corresponding surrogate.

\vspace*{-0.15em}
\paragraph{Complexity analysis.}
We present a programming language, \lang{}, in which programs denote stratified functions with well-defined complexity measures (specifically, stratified analytic functions).
We provide a static program analysis for \lang{} programs that automatically %
calculates an upper bound on the complexity of each component of the stratified~function~that~the~program~denotes.

\vspace*{-0.15em}
\paragraph{Evaluation.}
To demonstrate that \moptimal{} sampling using our complexity analysis improves surrogate error on downstream tasks, we evaluate our approach on a range of programs, finding that across this selection of programs \moptimal{} sampling improves error relative to baseline distributions by around $5\%$.
We demonstrate that a $5\%$ improvement in error of a surrogate can result in a $28\%$ improvement in execution speed in an application with a maximum error threshold.
We then analyze the classes of problems for which \moptimal{} sampling excels, finding potential improvements in error of up to $30\%$, and the classes of problems for which \moptimal{} sampling using \lang{}'s complexity analysis sampling fails, finding deteriorations in error of up to $500\%$.
\vspace*{-0.15em}
\paragraph{Renderer demonstration.}
We further present a case study of learning a surrogate of a renderer in a video game engine.
We show that our \moptimal{} sampling approach results in between $17\%$ and $44\%$ lower error than training using
baseline distributions that do not take into account path complexity.
These error improvements correspond with perceptual improvements in~the~generated~renders.

\vspace*{-0.15em}
\paragraph{Contributions.}
We present the following contributions:
\vspace*{-0.15em}
\begin{itemize}
\item An approach to allocating samples among strata to train stratified neural network surrogates of stratified analytic functions that minimizes an upper bound on the surrogate's error.
\item A programming language, \lang{}, in which all programs are stratified analytic functions, and a program analysis to bound the complexity of learning surrogates of those programs.
\item Empirical evaluations on real-world programs demonstrating that \moptimal{} sampling using \lang{}'s complexity analysis results in empirical improvements in error, and that these improvements in error result in improvements in downstream applications.
\item Further empirical evaluations of the classes of problems where \moptimal{} sampling using \lang{}'s complexity analysis succeeds and fails.
\end{itemize}
We lay the groundwork for analyzing \moptimal{} sampling approaches for training surrogates of programs.
Our results hold out the promise of surrogate training approaches that intelligently use the program's semantics to guide the design and training of surrogates of programs.

\section{Example}
\label{sec:example}
\Cref{fig:program-case-study} presents an example distilled from our evaluation (\Cref{sec:renderer}) that we use to demonstrate how \moptimal{} sampling results in a more accurate surrogate than \emph{frequency-based sampling}, sampling according to the frequency of paths alone.

\begin{figure}
  \begin{center}
    \begin{subfigure}[b]{0.51\textwidth}
      \input{code/case_learnguage_code.tex}

      \caption{Graphics program calculating the luminance of a pixel as a function of ambient light~and~material~properties.}
      \label{fig:program-case-study}
    \end{subfigure}
    \hfill
    \begin{subfigure}[b]{0.47\textwidth}
      \begin{center}
        \includegraphics[width=\textwidth]{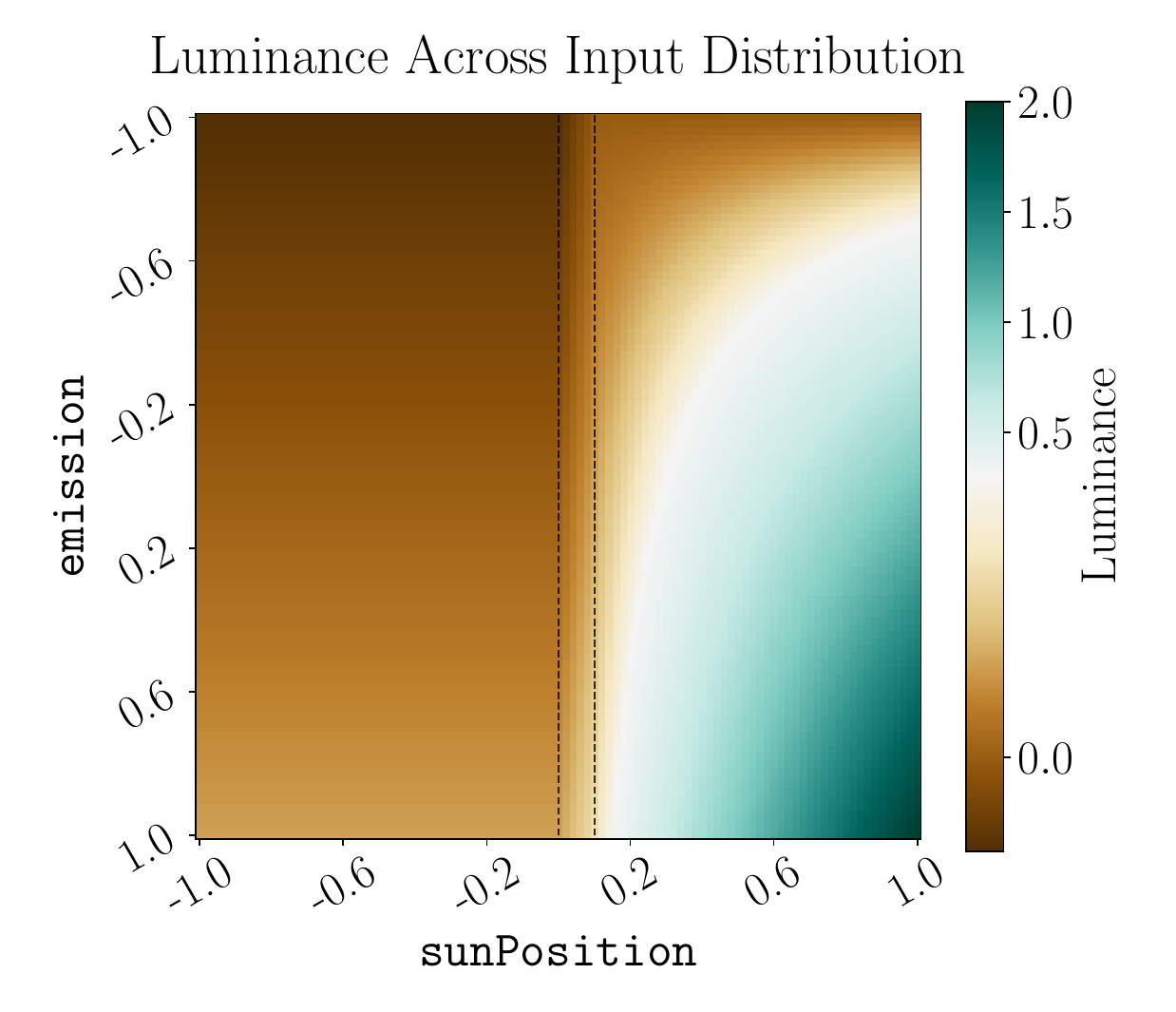}
      \end{center}
      \caption{Output of the program on inputs in $[-1, 1]$, with dashes separating~the~three~paths.}
      \label{fig:case-heatmap}
    \end{subfigure}
    \\

    \begin{center}
    \begin{subfigure}[b]{0.32\textwidth}
      \resizebox{\columnwidth}{!}{%
        \input{code/case_learnguage_code_l.tex}%
      }
      \caption{Nighttime (\texttt{ll}) path.}
      \label{fig:case-l-path}
    \end{subfigure}
    \hfill
    \begin{subfigure}[b]{0.32\textwidth}
      \resizebox{\columnwidth}{!}{%
        \input{code/case_learnguage_code_m.tex}%
        }
      \caption{Twilight (\texttt{rl}) path.}
      \label{fig:case-m-path}
    \end{subfigure}
    \hfill
    \begin{subfigure}[b]{0.32\textwidth}
      \resizebox{\columnwidth}{!}{%
        \input{code/case_learnguage_code_r.tex}%
        }
      \caption{Daytime (\texttt{rr}) path.}
      \label{fig:case-r-path}
    \end{subfigure}
    \end{center}
  \end{center}
  \caption{Example program, outputs, and traces.}
\end{figure}

\paragraph{Program under study.}
\Cref{fig:program-case-study} presents a graphics program that calculates the luminance (i.e., brightness) at a point in a scene as a function of
\texttt{sunPosition}, the height of the sun in the sky (i.e., the time of day), and \texttt{emission}, which describes how reflective the material is at that point.

The program first checks whether it is nighttime (\Cref{line:case-cond-1}), and sets the ambient lighting variable to zero accordingly.
The program next checks whether the sun position is below a threshold indicating direct sunlight (\Cref{line:case-cond-2}) and sets the emission variable accordingly.
The output is then the sum of the ambient light and the light emitted by the material.
\Cref{fig:case-heatmap} presents the output of this program over the valid input range of \texttt{sunPosition} and \texttt{emission} (i.e., between $-1$ and $1$ for both variables).

The path conditions (\Cref{line:case-cond-1,line:case-cond-2}) partition the program into three traces:
nighttime, when \texttt{sunPosition} is less than 0 (\Cref{fig:case-l-path});
twilight, when \texttt{sunPosition} is between $0$ and $0.1$ (\Cref{fig:case-m-path});
and daytime, when \texttt{sunPosition} is greater than $0.1$ (\Cref{fig:case-r-path}).
These paths are separated by dashed black lines in \Cref{fig:case-heatmap}.

\paragraph{Complexity.}
Training \NA{a surrogate} of this program poses a particular challenge because these traces have not only different behavior but also different relative complexities:
when \texttt{sunPosition} is less than $0.1$ the function is linear, but when \texttt{sunPosition} is above $0.1$ the function is quadratic.
This notion of complexity is quantified by the \emph{sample complexity} of each trace: traces that are more complex require more samples to learn to a given error than traces that are less complex.
\Cref{fig:case-path-errors} presents the error as a function of the training dataset size of surrogates of each trace trained in isolation, showing that
indeed the quadratic daytime path the has the highest error, followed by twilight then nighttime.

\paragraph{\mOptimal{} sampling.}
Our objective is to find the number of data points to sample from each path to minimize the expectation of error of a surrogate of the overall program, given a data distribution and a data budget.
To accomplish this, our approach leverages the
complexity of each path
and the
frequency of each path in the data distribution,
prioritizing sampling paths that are more complex (requiring more samples to learn) and that are more frequent (and thus more important to learn).

First our approach determines the sample complexity of each trace along each path, the number of samples required to learn a surrogate of the trace (in isolation) to a given error.
Our approach extends the sample complexity results of \citet{agarwala_monolithic_2021}, who give an upper bound on the number of samples required to learn a neural network approximation of a given analytic function.
Using this bound (as implemented by our \lang{} analysis described in \Cref{sec:language}), our approach determines that the twilight path takes $1.4\times$ as many samples to train a surrogate to a given error as the nighttime path,
and the daytime path requires $3.7\times$ as many samples.

Then given a distribution with the frequency of each path, our approach determines the \moptimal{} sampling rates for each path.
In this example we assume that the data has a uniform distribution over inputs between $-1$~and~$1$, resulting in path frequencies for the nighttime path (\texttt{sunPosition~<~0}) of $50\%$, the twilight path (\texttt{0~<~sunPosition~<~0.1}) of $10\%$, and the daytime path (\texttt{0.1~<~sunPosition})~of~$40\%$.
With this, our approach determines that the nighttime path should be sampled at $36.9\%$ of the data budget (undersampling relative to its frequency because it is simple to learn), the twilight path at $14.0\%$, and the daytime path at $49.1\%$ (oversampling relative to its frequency because it is complex to learn).

\paragraph{Stratified surrogates.}
The class of surrogate model for which we derive the above approach is that of a \emph{stratified neural surrogate} -- a set of disjoint neural networks which are applied based on which path the inputs induce in the program.
Concretely, this means that we train one surrogate per path, and pick which to apply for each input at evaluation time.
For this example program, picking which surrogate to apply just requires comparing \texttt{sunPosition} against constant threshold values. %

\begin{figure*}
  \begin{center}
    \begin{subfigure}[t]{0.49\textwidth}
      \includegraphics[width=\textwidth]{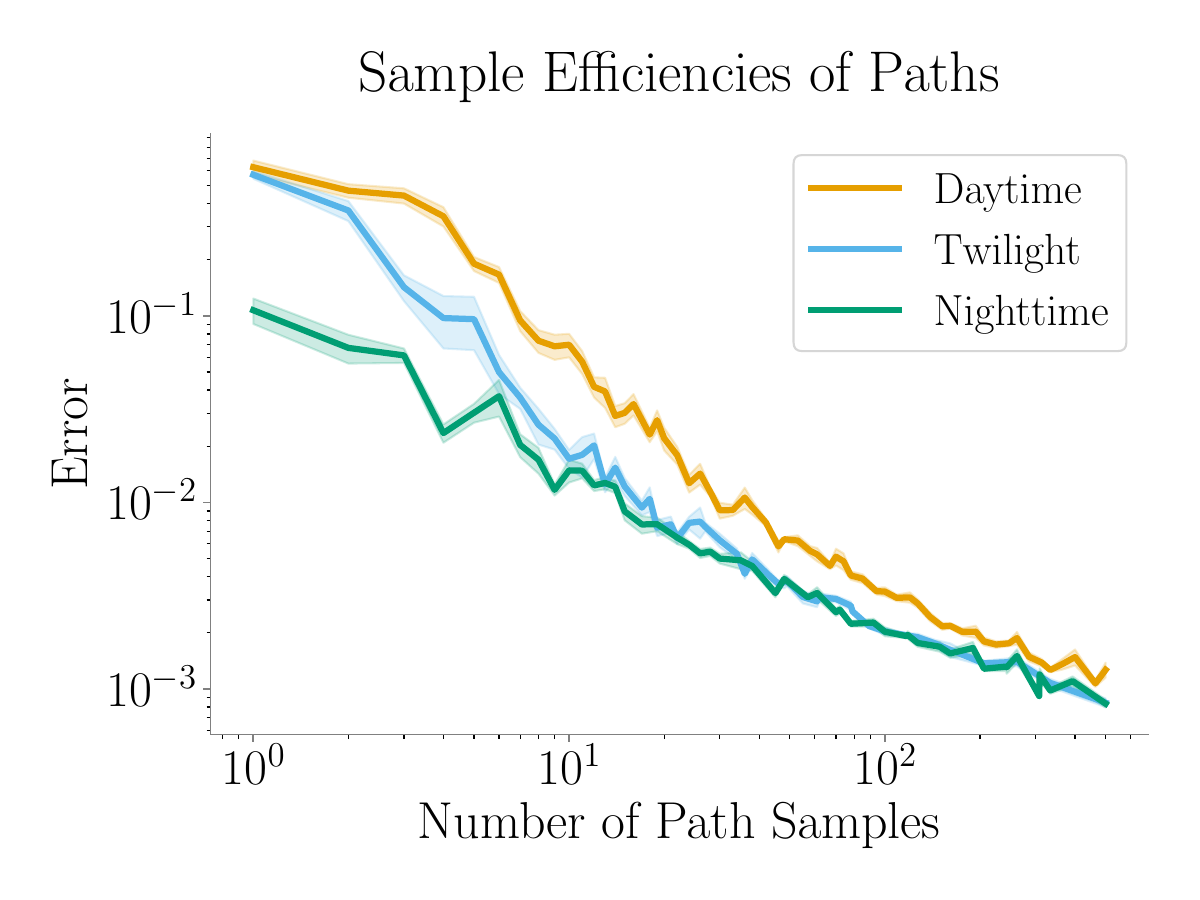}
      \caption{Per-path surrogate errors (log-log plot).}
      \label{fig:case-path-errors}
    \end{subfigure}
    \begin{subfigure}[t]{0.49\textwidth}
      \includegraphics[width=\textwidth]{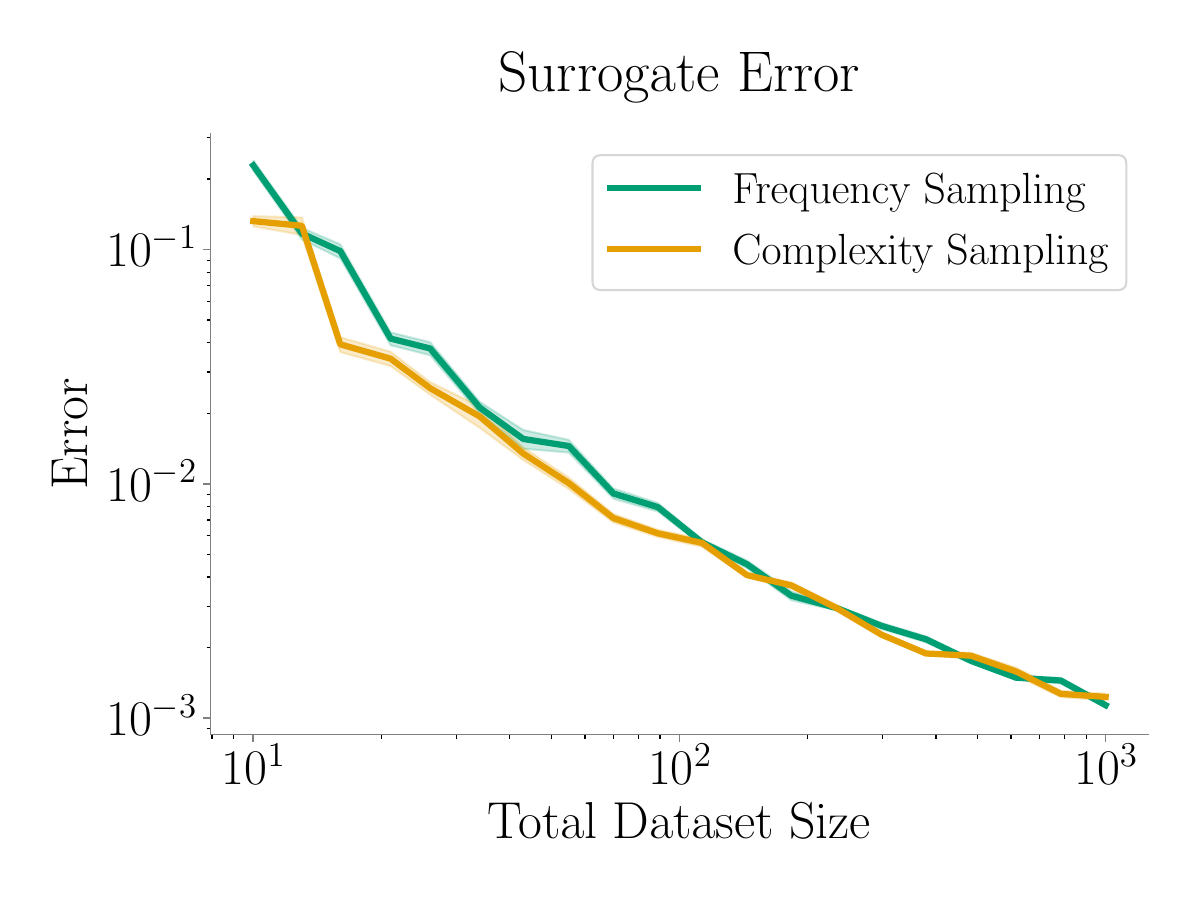}
      \caption{Stratified surrogate errors (log-log plot). Our approach decreases the error by $15\%$.}
      \label{fig:case-all-errors}
    \end{subfigure}
  \end{center}
  \caption{Per-path surrogate errors (left) and combined errors (right)~for~the~example.}

  \label{fig:case-results}
\end{figure*}

\paragraph{Results.}
\Cref{fig:case-all-errors} presents the error as a function of the training dataset size of stratified surrogates of the entire program
for a baseline of sampling according to path frequency alone and for \moptimal{} sampling.
\Cref{fig:case-all-errors} shows that the \moptimal{} sampling approach results in lower error than sampling according to path frequency alone.
For datasets of total size below $70$ samples, the surrogate trained with \moptimal{} sampling has a geometric mean decrease in error of $27.5\%$.
For datasets of total size above $70$ samples, the surrogate trained with \moptimal{} sampling has a geometric mean decrease in error of $5.5\%$.
Across the entire range of dataset sizes evaluated in this plot, the surrogate trained with \moptimal{} sampling has a geometric mean decrease in error of $15\%$.
In sum, our approach results in a surrogate that produces a more accurate luminance calculation, and therefore a better final output from the graphics program, than a surrogate trained using frequency-based path sampling.

\section{Complexity-Guided Sampling}
\label{sec:formalism}

In this section we present the stratified surrogate sample allocation problem and derive our solution, complexity-guided stratified surrogate dataset selection.

\subsection{The Stratified Surrogate Sample Allocation Problem}
Our goal is to learn a \emph{stratified surrogate}, $\msurrogate{}$, of a \emph{stratified function}, $\mfunction{}$, constrained by a \emph{sample budget},~$n$, that defines the number of data samples to be used by the learning algorithm.

We approach this problem through the definition of a stratified function as a piecewise function; we term each piece a \emph{stratum}.
We then define a stratified surrogate as a stratified function itself, with each stratum a surrogate of a corresponding stratum of the stratified function.
Learning a stratified surrogate therefore requires learning a surrogate for each stratum.

We assume that we have a technique for learning a surrogate of a function, $\mfunction{}$, given a sample budget.
Our precise goal is thus to partition the overall sample budget, $n$, into per-stratum sample budgets for each stratum of the stratified function, with the objective of minimizing the overall error of the stratified surrogate.

\subsubsection{Stratified Functions and Surrogates}

We define a \emph{stratified function} $\mfunction$ as follows: \[
  \mfunction{}\mleft(\minput{}\mright) \triangleq \begin{cases}
    \mfunction{}_1\mleft(\minput{}\mright) & \text{if}\; \minput{} \in \mstratum{}_1 \\
    \vdots \\
    \mfunction{}_\mnstrata{}\mleft(\minput{}\mright) & \text{if}\; \minput{} \in \mstratum{}_\mnstrata{} \\
  \end{cases}
\]
where $f$ and each $f_i$ is a function from inputs $\minput{}: \mInput{}$ to outputs $\moutput{}: \mOutput{}$, $\mnstrata{}$ is the number of strata, $\{\mstratum{}_i\}_{i=1}^c$ are strata, and where $\cup_i \mstratum{}_i = \mInput{}$ and $\forall i \neq j.\; \mstratum{}_i \cap \mstratum{}_j = \varnothing$.
We define a \emph{stratified surrogate} $\msurrogate{}$ as a stratified function with components $\msurrogate{}_i$.
\subsubsection{The Stratified Surrogate Sample Allocation Problem}

To restate, our goal is to learn a stratified surrogate $\msurrogate{}$ of a stratified function $\mfunction{}$. %

Formally, we define a \emph{learning algorithm}, a function that learns a surrogate of a given input function, as
a \NA{random function} $\mlearner{}: \mleft(\mFunction{}\mright) \times \mDistribution{} \times \mNat{} \times \mleft(\mOutput{} \times \mOutput{} \to \mReal{}_{\geq 0}\mright) \to \mleft(\mFunction{}\mright)$ that takes a function $\mfunction{}: \mFunction{}$ from inputs $\minput{}: \mInput{}$ to outputs $\moutput{}: \mOutput{}$, a distribution $\mdistribution{}: \mDistribution{}$ over inputs $\minput{}$, a number of training examples $\mnat{}: \mNat{}$, and a loss function $\mloss{}: \mOutput{} \times \mOutput{} \to \mReal{}_{\geq 0}$ which measures the cost of an incorrect prediction, and returns a function (representing the output surrogate) $\msurrogate{}: \mFunction{}$.

We also define notation for data distributions $\mdistribution{}$.
Let $\mdistribution{}\mleft(\minput{}\mright)$ be the probability that $\minput{}$ is sampled from $\mdistribution{}$, and $\mdiststratum{i}$ be the probability mass of all data points within $\mstratum{}_i$ over $\mdistribution{}$ (reducing to a summation for discrete distributions).
Let $\mdistribution{}\mleft(\minput{} \middle| \mstratum{}_i\mright)$, the distribution of $x$ within stratum~$s_i$, be defined as:

\begin{equation*}
  \mdistribution{}\mleft(\minput{} \middle| \mstratum{}_i\mright) \triangleq
  \begin{cases}
    \frac{\mdistribution\mleft(\minput{}\mright)}{\int_{x' \in \mstratum{}_{i}} \mdistribution{}\mleft(x'\mright) {\rm d}x'} & x \in s_i \\
    0 & \text{otherwise}
\end{cases}
\end{equation*}

We next define a \emph{stratified learning algorithm}.
A stratified learning algorithm learns a stratified surrogate of a stratified function by learning each component surrogate independently (given their respective dataset budgets).
We use the following notation to denote the operation of a stratified learning algorithm, where $\vec{n}$ is a vector of sample budgets for each stratum: \[
  \msurrogate{} \sim \mlearner{}\mleft(\mfunction{}, \mdistribution{}, \vec{n}, \mloss{}\mright) \triangleq \mleft\{\msurrogate{}_i \sim \mlearner{}\mleft(\mfunction{}_i, \mdistribution{}\mleft(\minput \middle| \mstratum_i \mright), \vec{\mnat{}}_i, \mloss{}\mright) \mright\}
\]

We formalize stratified surrogate sample allocation with the following optimization problem:
\begin{equation}
  \label{eq:distributionsampling-expectation}
  \begin{split}
    \argmin_{\vec{n}} %
    \expectation{}_{\msurrogate{} \sim \mlearner{}\mleft(\mfunction{}, \mdistribution{}, \vec{\mnat{}}, \mloss{}\mright)}
      \mleft[
      \expectation_{\minput{} \sim \mdistribution{}} \mleft[
      \mloss\mleft(\msurrogate\mleft(\minput\mright), \mfunction\mleft(\minput\mright)\mright)
      \mright]
      \mright]
    \text{~~such that~~}
      \sum_i \vec{\mnat{}}_i \leq n
  \end{split}
  \raisetag{24pt}
\end{equation}
The objective of this problem is to find a vector of per-stratum sample budgets $\vec{n}$ that
in the expectation over the outcomes of the stratified surrogate learning algorithm (the outer expectation)
minimize the expected loss over the data distribution (the inner expectation),
subject to a constraint that the total number of samples used is no more than $n$.

\subsection{Complexity-Guided Stratified Surrogate Dataset Selection}
\label{sec:optimal-sampling}
In this section, our goal is to solve \Cref{eq:distributionsampling-expectation}.
To solve this optimization problem, we need to model the relationship between the sample budget afforded to the learning algorithm for each stratum and the error of the resulting surrogate.
We leverage the PAC learning framework for neural networks to derive a conservative probabilistic upper bound on the error of the surrogate.
We then solve this optimization problem with our derived upper bound in place of our original objective.

\subsubsection{PAC Learning}
To reason about the error of a surrogate, we use the \emph{probably approximately correct} (PAC) learning framework~\citep{valiant_theory_1984}.
The PAC learning framework bounds the number of examples needed to learn a surrogate as a function of the allowable error threshold for the surrogate.

\Cref{eq:learnability} defines a given function $\mfunction{}$ as \emph{probably approximately correctly learnable}~\citep{valiant_theory_1984} (abbreviated as learnable) for a given learning algorithm $\mlearner{}$ and loss function $\mloss{}$ if for all distributions~$D$, with probability $1 - \mprob{}$ the learning algorithm returns a surrogate $\msurrogate{}$ that approximately matches the original function $\mfunction{}$ over the distribution $\mdistribution{}$ (i.e., the expectation of the error is bounded by $\epsilon$):
 \begin{equation}
   \label{eq:learnability}
   \begin{split}
     \forall \mdistribution{}, \merror \in (0, 1), \mprob \in (0, 1).\, \exists n.
     \probability_{\msurrogate{} \sim \mlearner{}\mleft(\mfunction{}, \mdistribution{}, \mnat{}, \mloss{}\mright)}\mleft(\expectation_{\minput \sim \mdistribution} \mleft[\mloss{}\mleft(\msurrogate{}\mleft(\minput{}\mright), \mfunction{}\mleft(\minput{}\mright)\mright)\mright] \leq \merror{} \mright) \geq 1 - \mprob
   \end{split}
 \end{equation}

 \subsubsection{Neural Network Sample Complexity Measures}
 \label{sec:nn-sample-compl}
It is an open problem to determine the exact relationship between the number of samples $n$ and the target error threshold $\epsilon$ in the PAC bound (\Cref{eq:learnability}) for neural networks on arbitrary target functions $f$.
Rather than use the exact relationship, we use an upper bound on $\epsilon$ as a function of $n$, and minimize the induced upper bound.

\citet{arora_fine-grained_2019} and \citet{agarwala_monolithic_2021} present such an upper bound for learning analytic functions $f$ with neural networks.
\citeauthor{agarwala_monolithic_2021} define a \emph{sample complexity measure} $\mcomplexity{\mfunction} \in \mathbb{R}_{\geq 0}$, where higher values denote functions that require more samples $n$ to learn $\mfunction{}$ to a given error $\epsilon$.
With this sample complexity measure $\mcomplexity{f}$,
\Cref{eq:learnability} holds
for all analytic $f$, $n$, $D$, $\epsilon$, and $\delta$
with:
\begin{equation}
  \label{eq:sample-complexity}
 \exists K.\, \epsilon \leq K \sqrt{\frac{\mcomplexity{\mfunction} + \log \mleft(\mprob^{-1}\mright)}{\mnat}}
\end{equation}
where $K$ is an unknown constant.
\citeauthor{agarwala_monolithic_2021} define $\mcomplexity{f}$
using the \emph{tilde}~$\tilde{f}$ of $f$, defined as follows for univariate functions:
\begin{equation}
  \label{eq:tilde-definition-univar}
  f(x) = \sum_{k=0}^\infty a_k x^k
  \hspace*{4em}
  \tilde{f}(x) \triangleq \sum_{k=0}^\infty |a_k| x^k
\end{equation}
The tilde measures the magnitude of each coefficient of $f$'s analytic representation; this is a measure of the influence of hard-to-model higher-order terms.
We work with the following generalization of the tilde to multivariate analytic functions, where $\vec{x}\Vert1$ denotes concatenating a $1$ to $\vec{x}$:
\begin{equation}
  \label{eq:tilde-definition-multivar}
  f\mleft(\vec{x}\mright) = \sum_{k=0}^\infty \sum_{v \in V_k} a_v \prod_{i=1}^k \mleft(\beta_{v, i} \cdot \vec{x}\Vert1\mright)
  \hspace*{3.8em}
  \tilde{f}\mleft(x\mright) = \sum_{k=0}^\infty \mleft(\sum_{v \in V_k} \mleft|a_v\mright| \prod_{i=1}^k \mleft\Vert \beta_{v, i} \mright\Vert_2\mright)x^k
\end{equation}
\citeauthor{agarwala_monolithic_2021} present the multivariate generalization; we contribute the novel generalization to $\vec{x}\Vert1$, which allows us to handle functions that are not analytic around 0 such as $\log$.

With the definition of the tilde, we now present \citeauthor{agarwala_monolithic_2021}'s core theorem, which says that the tilde induces a sample complexity measure for analytic functions:
\begin{restatable}{theorem}{complexity}
  \label{thm:complexity-definition}
For a sufficiently wide (see \citet[Theorem~5.1]{arora_fine-grained_2019}) 2-layer neural network trained with gradient descent for sufficient steps (ibid.),
if $f$ is analytic,
$\vec{x}$ is on the $d$-dimensional unit sphere, %
and $\mloss$ is 1-Lipschitz,
then $f\mleft(\vec{x}\mright)$
is learnable in the sense of \Cref{eq:learnability,eq:sample-complexity} with:
\[
  \mcomplexity{f} = \tilde{f}'\mleft(1\mright)^2
\]
\end{restatable} \noindent
We present the proof of this theorem in \Cref{app:complexity-algebra}.
The proof is a novel extension to inputs $\vec{x}\Vert1$ of \citet{agarwala_monolithic_2021}'s proof.

\subsubsection{Complexity-Guided Stratified Surrogate Dataset Selection}
Coming back to the stratified surrogate sample allocation problem (\Cref{eq:distributionsampling-expectation}), our goal is to find per-stratum sample budgets $\vec{n}$ that minimize the expectation of error of the stratified surrogate.
To help solve this optimization problem, we can refactor \Cref{eq:distributionsampling-expectation} to separate out each stratum as follows:
\begin{equation}
  \label{eq:distributionsampling-expectation-split}
  \begin{split}
    \argmin_{\vec{n}}
    \expectation_{s_i \sim \mdiststratum{i}} \mleft[
    \expectation_{\msurrogate{}_i \sim \mlearner{}\mleft(\mfunction{}_i, \mdistribution{}\mleft(\minput \middle| \mstratum_i \mright), \vec{\mnat{}}_i, \mloss{}\mright)} \mleft[
    \expectation_{x \sim \mdistribution{}\mleft(x \middle| s_i\mright)} \mleft[
    \mloss\mleft(\msurrogate_i\mleft(\minput\mright), \mfunction_i\mleft(\minput\mright)\mright)
    \mright]
    \mright]
    \mright]
    \text{~~such that~~}
    \sum_i \vec{\mnat{}}_i \leq n
  \end{split}
\end{equation}
This refactoring exploits that a stratified learning algorithm learns the surrogate for each stratum independently:\footnote{Specifically, we decompose the innermost expectation in \Cref{eq:distributionsampling-expectation} over strata using the law of total expectation, move the expectation over strata to the outside using that expectation is linear, then rewrite the expectation over the stratified learning algorithm to be the expectation over the single stratum under consideration, using that the stratified learning algorithm learns the surrogate for each stratum independently.}
we can decompose the expected loss of the learning algorithm into the expectation over strata (the outermost expectation in \Cref{eq:distributionsampling-expectation-split})
of the expectation over the outcomes of the surrogate learning algorithm on that stratum (the middle expectation)
of the expected loss of the expected loss over the data distribution on that stratum (the inner expectation).

\subsubsection{Predicted Error of a Surrogate}
Instead of optimizing \Cref{eq:distributionsampling-expectation-split} directly, our approach is to optimize the conservative probabilistic upper bound $\epsilon_i$ given by the PAC framework for each surrogate.

We define the \emph{predicted error} $\hat{\epsilon}_{f_i,n_i,\delta_i}$ of a stratified surrogate component to be the upper bound (with probability $1-\delta_i$) of the error of the surrogate $\hat{f}_i$ against the function $f_i$.
Concretely, the predicted error is the error for a given $n_i$ and $\delta_i$ assuming that \Cref{eq:sample-complexity} is tight with $K=1$ (the value of $K$ cancels out in our analysis, so this choice is just for notational convenience):
\begin{equation}
  \label{eq:predicted-error}
  \hat{\epsilon}_{f_i,n_i,\delta_i} \triangleq \sqrt{\frac{\zeta\mleft(f_i\mright) + \log\mleft(\delta_i^{-1}\mright)}{n_i}}
\end{equation}

We then replace the expectation of error in \Cref{eq:distributionsampling-expectation-split} in each stratum with the predicted error of that stratum, resulting in the objective that our approach optimizes:
\begin{equation}
  \label{eq:distributionsampling-predicted}
  \begin{split}
    \argmin_{\vec{n}}
    \expectation_{s_i \sim \mdiststratum{i}} \mleft[
    \hat{\epsilon}_{f_i,\vec{n}_i,\delta_i}
    \mright]
    \text{~~such that~~}
    \sum_i \vec{\mnat{}}_i \leq n
  \end{split}
\end{equation}
The objective of this problem is to find a vector of per-stratum sample budgets $\vec{n}$ that
in the expectation over strata (the outer expectation)
minimize the predicted error of the surrogate for that stratum,
subject to a constraint that the total number of samples used is no more than $n$.

Finally we can solve this optimization problem.
For a given stratified function $f$, sample budget $n$, and per-stratum failure~probabilities~$\delta_i$:
\begin{restatable}{theorem}{distributionsampling}
  \label{th:distributionsampling}
  \Cref{eq:distributionsampling-predicted} is minimized at:
	\begin{gather}
\vec{n}_i = n\frac{
	\mleft(\mleft(\mdiststratum{i}\mright) \sqrt{\mcomplexity{\mfunction_i} + \log\mleft(\delta_i^{-1}\mright)}\mright)^{\frac{2}{3}}
}{
  \sum_{j=1}^{\mnstrata} \mleft(\mleft(\mdiststratum{j}\mright) \sqrt{\mcomplexity{\mfunction_j} + \log\mleft(\delta_j^{-1}\mright)}\mright)^{\frac{2}{3}}
      }
	\end{gather}
\end{restatable} \noindent
\Cref{th:distributionsampling} defines how much data our complexity-guided sampling approach samples from each stratum.
Specifically, data is sampled from each stratum proportionally to: \[
  \mleft(\mleft(\mdiststratum{i}\mright) \sqrt{\mcomplexity{\mfunction_i} + \log\mleft(\delta_i^{-1}\mright)}\mright)^{\frac{2}{3}}
\]
This term incorporates the frequency of that stratum ($\mdiststratum{i}$), the complexity of that stratum ($\mcomplexity{\mfunction_i}$), and a term from the failure probability $\delta_i$.
We present the proof of \Cref{th:distributionsampling} in \Cref{app:formalism}.

For convenience, throughout the rest of this paper we assume that all $\delta_i$ are set to be equal ($\forall i, j.\, \delta_i=\delta_j$).
Because each surrogate training is independent, this induces an overall PAC failure probability $\delta=1 - \prod_i \mleft(1 - \delta_i\mright)$.

\subsubsection{Tightness of the Predicted Error Optimization}
We note that the optimal solution to \Cref{eq:distributionsampling-predicted} is not necessarily the optimal solution to \Cref{eq:distributionsampling-expectation}.

First, optimizing the predicted error is not the same as optimizing the expectation of error: specifically, there is a gap between the optimal solution to \Cref{eq:distributionsampling-predicted} and the optimal solution to \Cref{eq:distributionsampling-expectation}.
We note that assuming that the per-example loss is bounded by some value $L$, the expectation of error found by the optimal $\vec{n}$ for \Cref{eq:distributionsampling-predicted} is bounded by:
\begin{equation}
  \expectation_{s_i \sim \mdiststratum{i}} \mleft[\mleft(1 - \delta_i\mright)\sqrt{\frac{\mcomplexity{\mfunction_i} + \log \mleft(\mprob_i^{-1}\mright)}{\vec{\mnat}_i}} + \delta_i L \mright]
\end{equation}

Second, the bound on the predicted error itself may be loose.
We note that while the predicted error itself may be a loose bound on the error, our approach does not require exact values from these bounds, but instead compares the predicted error of each different component of the stratified function to minimize the overall predicted error.

\section{\lang{}: Programs as Stratified Functions}
\label{sec:language}

In this section we present \lang{}, a programming language in which programs denote learnable stratified functions.
We provide a program analysis for \lang{} programs that calculates an upper bound on the complexity of each component of the stratified function that the program denotes.

\subsection{Syntax and Standard Interpretation}
\label{sec:syntax}
\begin{wrapfigure}[11]{r}{0.4\textwidth}
  \centering
  \vspace*{-2em}
  \begin{minipage}{0.4\textwidth}
  \begin{align*}
    p \bnf& \texttt{{\color{blue}fun} (}x, \dots, x\texttt{) \{}s\,\texttt{;}\, \texttt{{\color{blue}return} }x \} \\
    s \bnf&
            \tkeyword{skip}
            \bor s \,\texttt{;}\,  s
            \bor x \;\texttt{=}\; e
            \\&
            \bor \texttt{{\color{blue}if} (}e \texttt{ > 0) \{} s \texttt{\} {\color{blue}else} \{} s \texttt{\} }
    \\
    e \bnf& x \bor v \bor e + e \bor e * e \bor {-e} \bor \tkeyword{sin}\texttt{(}e\texttt{)} \\
          &\bor \tkeyword{exp}\texttt{(}e\texttt{)} \bor \tkeyword{log}\texttt{\{}v\texttt{\}(}e\texttt{)}
    \\
    x \bnf& \text{set of variable names} \\
    v \bnf& \text{set of floating-point values}
  \end{align*}
  \end{minipage}

  \caption{Syntax of \lang{}.}
  \label{fig:lang-syntax}
  \end{wrapfigure}
\Cref{fig:lang-syntax} presents the syntax of \lang{}, a loop-free language similar to IMP~\citep{winksel_formal_1993}.
A \lang{} program~$p$ takes a list of inputs~$x$, executes a top-level statement~$s$, and returns a single variable~$x$.
Statements~$s$ are skips, sequences, assignments, or if statements.
Expressions~$e$ are variables~$x$, floating-point values~$v$, binary operations, or unary operations.

\lang{} supports analytic functions (e.g., $\sin$, $\exp$), including those which are analytic only on a subset of the reals (e.g., $\log$).
We restrict the supported operations to those required to implement the evaluation in \Cref{sec:evaluation}.
\makeatletter
\let\markeverypar\everypar
\newtoks\everypar
\everypar\markeverypar
\markeverypar{\the\everypar\looseness=-1\relax}
\makeatother

\begin{figure}[h]
  \begin{mathpar}
    \centering
    \inferrule{ }{\langle \sigma , v\rangle \Downarrow v}

    \inferrule{ }{\langle \sigma , x\rangle \Downarrow \sigma\mleft(x\mright)}

    \inferrule{
      \langle \sigma , e_1\rangle \Downarrow v_1 \\
      \langle \sigma , e_2\rangle \Downarrow v_2 \\
    }{\langle \sigma , e_1 \texttt{+}  e_2\rangle \Downarrow v_1 + v_2}

    \inferrule{
      \langle \sigma , e_1\rangle \Downarrow v_1 \\
      \langle \sigma , e_2\rangle \Downarrow v_2 \\
    }{\langle \sigma , e_1 \texttt{*} e_2\rangle \Downarrow v_1 \cdot v_2}

    \inferrule{
      \langle \sigma , e\rangle \Downarrow v \\
    }{\langle \sigma , {-e} \rangle \Downarrow -v}

    \inferrule{
      \langle \sigma , e\rangle \Downarrow v \\
    }{\langle \sigma , \texttt{{\color{blue}sin}(}e\texttt{)}\rangle \Downarrow \sin\mleft(v\mright)}

    \inferrule{
      \langle \sigma , e\rangle \Downarrow v \\
    }{\langle \sigma , \texttt{{\color{blue}exp}(}e\texttt{)}\rangle \Downarrow \exp\mleft(v\mright)}

    \inferrule{
      \langle \sigma , e\rangle \Downarrow v \\
      \mleft|b - v\mright| < b
    }{\langle \sigma , \texttt{{\color{blue}log}\{}b\texttt{\}(}e\texttt{)}\rangle \Downarrow \log\mleft(v\mright)}
  \end{mathpar}
  \caption{Big-step evaluation relation for expressions in \lang{}.}
  \label{fig:evaluation-expression-relational-app}

\end{figure}

\paragraph{Standard execution semantics.}
\Cref{fig:evaluation-expression-relational-app} presents the big-step evaluation relation for expressions in \lang{}.
The expression relation $\langle \sigma, e \rangle \Downarrow v$ says that under variable store $\sigma$ (assigning values to all variables in $e$), the expression $e$ evaluates to value $v$.
These semantics are standard to IMP-like languages with the exception of that for $\texttt{{\color{blue}log}\{}b\texttt{\}(}e\texttt{)}$:
note that the expression $\texttt{{\color{blue}log}\{}b\texttt{\}(}e\texttt{)}$ takes an additional parameter $b$ and requires $\mleft|b - v\mright| < b$.
We discuss this requirement in \Cref{sec:analysis} and \Cref{app:log-rule}.

The big-step evaluation relation for statements and full \lang{} programs are standard and are presented in \Cref{app:language-semantics}.

\subsection{Complexity Analysis}
\label{sec:analysis}
We now present a program analysis that gives an upper bound on the complexity of \emph{traces} of \lang{} programs,
sequences of statements without if statements.
The analysis uses two core concepts:
a \emph{complexity interpretation} of expressions to calculate an upper bound on the tilde of expressions (\Cref{sec:nn-sample-compl}), %
and a standard \emph{dual-number execution}~\citep{wengert_simple_1964,griewank_derivatives_2008} of the complexity interpretation to calculate the derivative of the upper bound on the tilde, which as we show below is also an upper bound on the derivative of the tilde.
The result of the dual-number execution allows us to upper bound the complexity of a trace of a \lang{} program.

\subsubsection{Program Analysis}
First we walk through the rules of the program analysis, presented as a big-step evaluation relation.

\begin{figure}[h!]
  \input{tilde-expressions-figure}
\end{figure}

\Cref{fig:tilde-expression-relational-app} presents the relation used to calculate the tilde for expressions in \lang{}.
The relation $\mleft \langle \tilde{\sigma}, e \mright\rangle \cstepsto \mleft(\tilde{v}, \tilde{v}'\mright)$
says that under the variable complexity mapping $\tilde{\sigma}$ (mapping variables to tuples with their respective tilde and tilde derivative),
the expression $e$ has $\tilde{e} \leq \tilde{v}$ and $\tilde{e}' \leq \tilde{v}'$.

Broadly, we define the rules in \Cref{fig:tilde-expression-relational-app} using the definition of the tilde, the fact that the tilde is compositional (as we prove in \Cref{sec:tilde-compositional}) and the definition of a dual-number execution.
For instance, the tilde of a constant $v$ is the absolute value $\mleft|v\mright|$ of that constant with a derivative of 0, and the tilde of $e_1 \texttt{+} e_2$ is the sum of the tilde of each expression with a derivative of the sum of their derivatives.

A slightly more complex rule is that of $\texttt{{\color{blue}sin}(}e\texttt{)}$, which
computes $\sin\mleft(x\mright) = \sum_{n=0}^\infty \frac{(-1)^n}{(2n+1)!}x^{2n+1}$.
Thus, $\widetilde{\sin}\mleft(x\mright) = \sum_{n=0}^\infty \mleft|\frac{(-1)^n}{(2n+1)!}\mright|x^{2n+1} = \sum_{n=0}^\infty \frac{x^{2n+1}}{(2n+1)!} = \sinh\mleft(x\mright)$.
The derivative is then $\sinh'\mleft(x\mright) = x'\cosh\mleft(x\mright)$.
Because we know $\mleft \langle \tilde{\sigma}, e \mright\rangle \cstepsto \mleft(\tilde{v}, \tilde{v}'\mright)$ and because the tilde is compositional (\Cref{lem:tilde-algebra}), we can plug in $\tilde{v}$ and $\tilde{v}'$ to get $\mleft(\sinh\mleft(\tilde{v}\mright), \tilde{v}'\cosh\mleft(\tilde{v}\mright)\mright)$.

The most complex rule is the rule for $\texttt{{\color{blue}log}\{}b\texttt{\}(}e\texttt{)}$.
To handle that $\log(x)$ is not analytic around $0$, $\texttt{{\color{blue}log}\{}b\texttt{\}(}e\texttt{)}$ expands $\log\mleft(x\mright)$ around $x=b$.
The value of $b$ is a nuisance parameter that must be set to allow the expansion around $x=b$ to converge for all inputs (inducing the $\mleft|b - v\mright| < b$ requirement in \Cref{fig:evaluation-expression-relational-app}) while minimizing the overall program complexity.
Note that the condition $b > \tilde{v}\sqrt{b^2+1}$ can always be satisfied by applying the identity $\log\mleft(x\mright) = \log\mleft(\frac{x}{c}\mright) + \log\mleft(c\mright)$.

\begin{figure}[h]
  \input{tilde-statements-figure}
\end{figure}

\Cref{fig:tilde-statement-relational-app} presents the relation for calculating the tilde of all variables computed by traces (branch-free statements) in \lang{}.
The relation $\mleft \langle \tilde{\sigma}, t \mright \rangle \cstepsto \tilde{\sigma}'$
says that under the variable complexity mapping $\tilde{\sigma}$,
executing the trace $t$ computes variables with tildes and tilde derivatives upper-bounded by those of $\tilde{\sigma}'$.

\begin{figure}[h]
  \input{tilde-traces-figure}
\end{figure}

\Cref{fig:complexity-trace-relational-app} presents the complexity relation for traces in \lang{}.
The relation $\zeta_{\cstepsto}\mleft(t, x\mright) \leq z$ says that the trace $t$ has complexity upper bounded by $z$ for computing variable $x$ (under the assumptions in \citet{agarwala_monolithic_2021}).

\begin{figure}[h]
  \input{traces-collection-figure}
\end{figure}

\Cref{fig:collect-statement-relational-app} presents the trace collection relation for \lang{} statements.
The relation $\mleft\langle\tau, s\mright\rangle \leadsto \tau'$ says that under the trace mapping $\tau$ (mapping paths that reach this statement to the trace of statements executed thus far), executing the statement $s$ can result in
possible paths and corresponding traces $\tau'$.

\begin{figure}[h]
  \input{program-collection-figure}
\end{figure}

\Cref{fig:collect-program-relational-app} presents the trace collection relation for \lang{} programs.
The trace collection relation $\mleft\langle\texttt{{\color{blue}fun} (}x_0, x_1 \dots, x_n\texttt{) \{}s \,\texttt{;}\, \texttt{{\color{blue}return} } x \texttt{\}}\mright\rangle \leadsto \tau$ says that executing the program can result in possible paths and corresponding traces $\tau$.

\subsubsection{Tilde Calculus}
\label{sec:tilde-compositional}
This section presents the core lemma stating that the upper bound on the tilde is compositional, and that the derivative of the upper bound is an upper bound on the derivative.
The bounds on the tilde are from \citet{agarwala_monolithic_2021}; we extend these bounds to also bound the derivative of the tilde.
\begin{restatable}{lemma}{tildealgebra}
  \label{lem:tilde-algebra}
  The tilde and its derivative have upper bounds that are compositional with respect to $f$:
  \begin{align*}
    f\mleft(\vec{x}\mright) = g\mleft(\vec{x}\mright) + h\mleft(\vec{x}\mright) &\Rightarrow \forall x \geq 0.\, \tilde{f}\mleft(x\mright) \leq \tilde{g}\mleft(x\mright) + \tilde{h}\mleft(x\mright) \land \tilde{f}'\mleft(x\mright) \leq \tilde{g}'\mleft(x\mright) + \tilde{h}'\mleft(x\mright) \\
    f\mleft(\vec{x}\mright) = g\mleft(\vec{x}\mright) \cdot h\mleft(\vec{x}\mright) &\Rightarrow \forall x \geq 0.\, \tilde{f}\mleft(x\mright) \leq \tilde{g}\mleft(x\mright) \cdot \tilde{h}\mleft(x\mright) \land \tilde{f}'\mleft(x\mright) \leq \tilde{g}'\mleft(x\mright) \tilde{h}\mleft(x\mright) + \tilde{g}\mleft(x\mright) \tilde{h}'\mleft(x\mright) \\
    f\mleft(\vec{x}\mright) = g\mleft(h\mleft(\vec{x}\mright)\mright) &\Rightarrow \forall x \geq 0.\, \tilde{f}\mleft(x\mright) \leq \tilde{g}\mleft(\tilde{h}\mleft(x\mright)\mright) \land \tilde{f}'\mleft(x\mright) \leq \tilde{g}'\mleft(\tilde{h}\mleft(x\mright)\mright) \cdot \tilde{h}'\mleft(x\mright) \\
    &\hspace*{8em} \text{(when $\tilde{h}\mleft(x\mright)$ is in the radius of convergence of $g$)}
  \end{align*}
\end{restatable} \noindent
The proof of this lemma is presented in \Cref{app:complexity-algebra}.

\subsection{Soundness}
This section proves that the \lang{} complexity analysis is sound: that it computes an upper bound on the true complexity of learning a trace.
We prove this by induction on expressions and traces.
Our approach is based on the observation that at a given program point the value of each variable was computed by some function $f_x$ applied to the input.
We use the notation $\mleft\{f_x\mright\}$ as shorthand for $\mleft\{f_x \,\middle|\, x \in \sigma\mright\}$, a set of functions indexed by $x \in \sigma$.
Our inductive hypothesis requires that $\tilde{\sigma}$ contain the tilde and tilde derivative of each of these functions (evaluated at $1$, as in \Cref{thm:complexity-definition}).
We use the notation $\tilde{\sigma} \vdash \mleft\{f_x\mright\}$ to denote the predicate that each $f_x$ have tilde and tilde derivative bounded by $\tilde{\sigma}$:
\begin{gather*}
  \tilde{\sigma} \vdash \mleft\{f_x\mright\} \Leftrightarrow \forall x \in \tilde{\sigma}.\, \mleft(\tilde{\sigma}\mleft(x\mright) = \mleft(\tilde{v}, \tilde{v}'\mright) \Rightarrow \mleft(
  0 \leq \tilde{f_x}\mleft(1\mright) \leq \tilde{v} \land 0 \leq \tilde{f_x}'\mleft(1\mright) \leq \tilde{v}'\mright)\mright)
\end{gather*}

We also note that the standard execution semantics big-step relation $\Downarrow$ both for expressions and for traces is a function.
We use $\mleft\llbracket \cdot \mright\rrbracket$ as notation to refer to that function for expressions and $\mleft\llbracket \cdot \mright\rrbracket_x$ to refer to that function for traces followed by taking the value of the variable $x$:
\begin{gather*}
  \mleft\llbracket e \mright\rrbracket\mleft(\sigma\mright) = v \Leftrightarrow \mleft\langle\sigma, e\mright\rangle \Downarrow v \\
  \mleft\llbracket t \mright\rrbracket_x\mleft(\sigma\mright) = v \Leftrightarrow \mleft\langle\sigma, t\mright\rangle \Downarrow \sigma' \land \sigma'\mleft(x\mright) = v
\end{gather*}
We use the notation $\circ$ to denote function composition.
The functions denoted by expressions and traces have multiple inputs;
in this context, composition with a set of functions $\mleft\{f_x\mright\}$ is defined as follows:
\begin{gather*}
  \mleft(\mleft\llbracket \cdot \mright\rrbracket \circ \mleft\{f_x\mright\}\mright)\mleft(\sigma\mright) \triangleq \mleft\llbracket \cdot \mright\rrbracket \mleft(\mleft\{x \mapsto f_x\mleft(\sigma\mright)\mright\}\mright)
\end{gather*}

Now we state the core theorems of correctness for the \lang{} analysis:
\begin{restatable}{lemma}{tildeexpression}
  \label{lem:tilde-expression}
  The tilde big-step relation for expressions upper bounds the tilde and tilde derivative:
  \begin{equation*}
    \mleft( \langle \tilde{\sigma}, e \rangle \;\tilde\Downarrow\; \mleft(\tilde{v}, \tilde{v}'\mright) \land \tilde{\sigma} \vdash \mleft\{f_x\mright\} \mright) \Rightarrow \mleft( \widetilde{\mleft\llbracket e \mright\rrbracket \circ \mleft\{f_x\mright\}}\mleft(1\mright) \leq \tilde{v} \land \widetilde{\mleft\llbracket e \mright\rrbracket \circ \mleft\{f_x\mright\}}'\mleft(1\mright) \leq \tilde{v}' \mright)
  \end{equation*}
\end{restatable}

\begin{restatable}{lemma}{tildetrace}
  \label{lem:tilde-trace}
  The tilde big-step relation for traces upper bounds on the tilde and~tilde~derivative:
  \begin{equation*}
    \mleft(\mleft\langle\tilde{\sigma}, t\mright\rangle \cstepsto \tilde{\sigma}' \land \tilde{\sigma} \vdash \mleft\{f_x\mright\} \mright) \Rightarrow \tilde{\sigma}' \vdash \mleft\{ \mleft\llbracket t \mright\rrbracket_y \circ \mleft\{f_x\mright\} \mright\}
  \end{equation*}
\end{restatable}

\begin{restatable}{theorem}{analysisupperbound}
  \label{thm:analysis-upper-bound}
  The complexity relation computes an upper bound on the true complexity:
  \begin{equation*}
    \zeta_{\cstepsto}\mleft(t, x\mright) \leq z \Rightarrow \zeta\mleft(\mleft\llbracket t\mright\rrbracket_x\mright) \leq z
  \end{equation*}
\end{restatable}
The proofs of \Cref{lem:tilde-expression,lem:tilde-trace} and \Cref{thm:analysis-upper-bound} are presented in \Cref{app:language}.

\subsection{Precision}
\label{sec:precision}
We note that the analysis is a sound but imprecise approximation of complexity, in that the upper bound it computes is not tight.
For example, consider the expression $x \texttt{+} \texttt{(-}x\texttt{)}$:
under the \lang{} analysis,
$\langle\mleft\{x \mapsto \mleft(1, 1\mright)\mright\}, x \texttt{+} \texttt{(-}x\texttt{)} \rangle \cstepsto \mleft(2, 2\mright)$ even though $\mleft\llbracket x \texttt{+} \texttt{(-}x\texttt{)}\mright\rrbracket\mleft(\sigma\mright) = 0$.

\subsection{Extensions}
Our implementation extends \lang{} to support vector-valued variables, applying all operations elementwise.
Following \citet{agarwala_monolithic_2021}, we define the complexity of learning a vector-valued function to be the sum of the complexity of learning each output component.
Our implementation also supports other syntactic sugar including a minus operation and division by constants.

\paragraph{Loops.}
Our implementation of \lang{} also supports fixed- and bounded-length loops, though they are not required for any case study in the paper (thus we do not present them in this paper).
However, unbounded loops pose a challenge because our approach trains a distinct surrogate per path, which is not possible with unbounded loops.
This restriction to statically bounded length loops is a common feature of analyses that reason about numerical approximation, including reliability analyses~\citep{carbin_rely_2013,misailovic_chisel_2014,boston_probability_2015} and floating-point error analyses~\citep{darulova_sound_2014,magron_certified_2017,solovyev_rigorous_2018}.
Reasoning about loops with dynamic, input-dependent bounds requires separate techniques (e.g., \citet{boston_probability_2015}).

\section{Evaluation}
\label{sec:evaluation}

In this section we evaluate our \moptimal{} sampling approach using \lang{}'s complexity analysis to determine sampling rates for a range of benchmark programs.
We demonstrate that \moptimal{} sampling consistently results in more accurate surrogates than those trained using baseline distributions (the frequency distribution of paths and the uniform distribution of paths).
We also demonstrate that such an improvement in surrogate error can result in an improvement in execution speed in an application with a maximum error threshold.

In \Cref{sec:evaluation-table} we first evaluate across both a set of real-world programs, showing expected error improvements in practice, and also a set of synthetic programs, showing cases where the complexity-guided sampling approach shines and cases where it fails.
Then in \Cref{sec:renderer} we dive into a case study on a specific large-scale program, a demonstration 3D renderer~\citep{renderer}, \NA{such as forms the core} of a graphics rendering pipeline for a movie or 3D game engine~\citep{renderman,real_time_3d_rendering}.

{
\let\section\subsection
\let\subsection\subsubsection
\let\subsubsection\paragraph
\section{Evaluation Across Programs}
\label{sec:evaluation-table}

In this section we evaluate our \moptimal{} sampling approach using \lang{}'s complexity analysis to determine sampling rates for a range of benchmark programs.
We evaluate both a set of real-world programs, showing expected error improvements in practice, and also a set of synthetic programs, showing cases where the complexity-guided sampling approach shines and cases where it fails.

\begin{table*}
  \caption{Average change in error across all budgets from using complexity-guided sampling compared to baselines on each benchmark (higher values means \moptimal{} sampling has lower error).}
  \label{tab:benchmark-results}
  \resizebox{\textwidth}{!}{%
  \begin{tabular}{ccc|cc|cc}
  \toprule
  \multicolumn{3}{c|}{\textbf{Benchmark}} & \multicolumn{4}{c}{\textbf{Baseline}} \\
  \multirow{2}{1.6cm}{\centering\bf Program}
  & \multirow{2}{0.8cm}{\centering\bf LoC}
  & \multirow{2}{0.8cm}{\centering\bf Paths}
  & \multirow{2}{2cm}{\centering\bf Frequency \\ (Predicted)}
  & \multirow{2}{2cm}{\centering\bf Frequency \\ (Empirical)}
  & \multirow{2}{2cm}{\centering\bf Uniform \\ (Predicted)}
  & \multirow{2}{2cm}{\centering\bf Uniform \\  (Empirical)} \\&&&&&&\\ \midrule
Luminance & $14$ & $3$ & $2.58\%$ & $15.01\%$ & $6.97\%$ & $15.17\%$ \\
Huber & $13$ & $3$ & $0.49\%$ & $8.15\%$ & $1.93\%$ & $9.54\%$ \\
BlackScholes & $15$ & $2$ & $4.43\%$ & $3.61\%$ & $1.30\%$ & $4.00\%$ \\
Camera & $69$ & $3$ & $2.83\%$ & $0.56\%$ & $0.22\%$ & $1.36\%$ \\
EQuake & $34$ & $2$ & $7.45\%$ & $2.25\%$ & $7.45\%$ & $2.25\%$ \\
Jmeint & $176$ & $18$ & $2.34\%$ & $0.01\%$ & $8.44\%$ & $1.02\%$ \\
\midrule
Geomean & & & $3.33\%$ & $4.81\%$ & $4.33\%$ & $5.43\%$ \\
  \bottomrule
  \end{tabular}%
}
\end{table*}
\begin{table*}

\caption{Benchmark statistics.}
\label{tab:benchmark-stats}
\begin{tabular}{ccc|ccc} \toprule
  \multirow{2}{1.6cm}{\centering \bf Benchmark} & \multirow{2}{1.5cm}{\centering\textbf{Path}} & \multirow{2}{2cm}{\centering\textbf{Complexity}} & \multirow{2}{2cm}{\centering\bf Frequency \\ Distribution} & \multirow{2}{2cm}{\centering\bf Uniform \\ Distribution} & \multirow{2}{2cm}{\centering\bf Complexity \\ Distribution} \\&&&&\\ \midrule
\multirow{3}{1.6cm}{\bf\centering Luminance}
& \texttt{ll} & 0.01 & $50.00\%$ & $33.33\%$ & $36.94\%$  \\
& \texttt{rl} & 1.21 & $10.00\%$ & $33.33\%$ & $13.98\%$  \\
& \texttt{rr} & 9.00 & $40.00\%$ & $33.33\%$ & $49.07\%$  \\
\midrule
\multirow{3}{1.6cm}{\bf\centering Huber}
& \texttt{ll} & 9.00 & $50.00\%$ & $33.33\%$ & $44.25\%$  \\
& \texttt{lr} & 9.00 & $25.00\%$ & $33.33\%$ & $27.88\%$  \\
& \texttt{r} & 9.00 & $25.00\%$ & $33.33\%$ & $27.88\%$  \\
\midrule
\multirow{2}{1.6cm}{\bf\centering BlackScholes}
& \texttt{l} & 165.72 & $75.00\%$ & $50.00\%$ & $59.34\%$  \\
& \texttt{r} & 485.23 & $25.00\%$ & $50.00\%$ & $40.66\%$  \\
\midrule
\multirow{3}{1.5cm}{\bf\centering Camera}
& \texttt{ll} & 0.86 & $44.54\%$ & $33.33\%$ & $36.96\%$  \\
& \texttt{lrl} & 0.81 & $35.48\%$ & $33.33\%$ & $31.63\%$  \\
& \texttt{rrr} & 9.53 & $19.98\%$ & $33.33\%$ & $31.41\%$  \\
\midrule
\multirow{2}{1.6cm}{\bf\centering EQuake}
& \texttt{l} & 56.29 & $50.00\%$ & $50.00\%$ & $26.99\%$  \\
& \texttt{r} & 1169.50 & $50.00\%$ & $50.00\%$ & $73.01\%$  \\
\midrule
\multirow{18}{1.6cm}{\bf\centering Jmeint}
& \texttt{lllrrrll} & 7236100.00 & $18.74\%$ & $5.56\%$ & $13.25\%$  \\
& \texttt{lllrrrlrl} & 7236100.00 & $5.31\%$ & $5.56\%$ & $5.72\%$  \\
& \texttt{lllrrrlrrl} & 7236100.00 & $5.31\%$ & $5.56\%$ & $5.71\%$  \\
& \texttt{lllrrrrll} & 7236100.00 & $5.30\%$ & $5.56\%$ & $5.71\%$  \\
& \texttt{lllrrrrlrl} & 7236100.00 & $2.52\%$ & $5.56\%$ & $3.48\%$  \\
& \texttt{lllrrrrlrrl} & 7236100.00 & $2.51\%$ & $5.56\%$ & $3.47\%$  \\
& \texttt{lllrrrrrll} & 7236100.00 & $5.30\%$ & $5.56\%$ & $5.71\%$  \\
& \texttt{lllrrrrrlrl} & 7236100.00 & $2.52\%$ & $5.56\%$ & $3.48\%$  \\
& \texttt{lllrrrrrlrrl} & 7236100.00 & $2.52\%$ & $5.56\%$ & $3.48\%$  \\
& \texttt{rrrrrrll} & 7236100.00 & $18.67\%$ & $5.56\%$ & $13.22\%$  \\
& \texttt{rrrrrrlrl} & 7236100.00 & $5.29\%$ & $5.56\%$ & $5.71\%$  \\
& \texttt{rrrrrrlrrl} & 7236100.00 & $5.29\%$ & $5.56\%$ & $5.70\%$  \\
& \texttt{rrrrrrrll} & 7236100.00 & $5.34\%$ & $5.56\%$ & $5.74\%$  \\
& \texttt{rrrrrrrlrl} & 7236100.00 & $2.52\%$ & $5.56\%$ & $3.48\%$  \\
& \texttt{rrrrrrrlrrl} & 7236100.00 & $2.51\%$ & $5.56\%$ & $3.47\%$  \\
& \texttt{rrrrrrrrll} & 7236100.00 & $5.29\%$ & $5.56\%$ & $5.70\%$  \\
& \texttt{rrrrrrrrlrl} & 7236100.00 & $2.53\%$ & $5.56\%$ & $3.49\%$  \\
& \texttt{rrrrrrrrlrrl} & 7236100.00 & $2.52\%$ & $5.56\%$ & $3.48\%$  \\
  \bottomrule
\end{tabular}

\end{table*}

\paragraph{Methodology.}
Following the input scale assumptions from \citet{agarwala_monolithic_2021}, we sample each input variable uniformly between $[-1, 1]$ or $[0, 1]$ as appropriate for the program.
We insert scale factors as appropriate given the expected data distribution of the original program.
We then uniformly sample inputs from these ranges.
This induces both a data distribution over inputs and a path frequency distribution.

For all benchmarks other than the Jmeint benchmark, we evaluate using a training data budget using 10 points logarithmically spaced between 10 and 1000.
For the Jmeint benchmark, which is more data intensive, we evaluate using a training data budget using 10 points logarithmically spaced between 1000 and 10000.
When computing the complexity-guided sampling distribution, we use $\delta=0.1$.

For each path in each benchmark, we train a 1-hidden-layer MLP with 1024 hidden units with a ReLU activation, using 10,000 steps of Adam with learning rate $0.0005$ and~batch~size~128.
We run the training for 5 trials.

We report both the predicted error (\Cref{eq:predicted-error}) improvement and the empirical improvement, the geometric mean improvement in error across trials.
As in \Cref{sec:renderer}, improvement is defined as the mean percentage error between the predicted error for complexity-guided sampling and the baseline sampling method.

\subsection{Results}
\Cref{tab:benchmark-results} presents the results of the evaluation across 6 benchmark programs: Luminance, Huber, BlackScholes, Camera, EQuake, and Jmeint.
\Cref{tab:benchmark-stats} presents path and distribution statistics for each benchmark program.

We find that across this selection of programs, from predicted error improvements of $3.33\%$ against the frequency sampling baseline, complexity-guided sampling results in an empirical improvement of $4.81\%$;
from predicted error improvements of $4.33\%$ against the uniform sampling baseline, complexity-guided sampling results in an empirical improvement of $5.43\%$.
Such a magnitude of decrease in error can significantly affect a system end-to-end.
For example, consider \Cref{tab:speedups}.
This table shows the results of a hyperparameter search to choose the fastest-to-execute neural network that meets an error threshold of $10\%$\footnote{
This methodology is also used by \citet{esmaeilzadeh_neural_2012}.
}
(this table is a replication of Table 4 in Appendix A of \citet{renda_programming_2021}, with an added column of ``Error - $5\%$'').
\citeauthor{renda_programming_2021} chose the network with size 64, which has a $1.57\times$ speedup; however, a decrease in error of $5\%$ would result instead in choosing the network with size 32, which has a $2.01\times$ speedup, a $28\%$ improvement in application performance.

\begin{table}
  \caption{
The validation error and speedup of BERT models over a range of candidate embedding widths.
This table is a replication of Table 4 in Appendix A of \citet{renda_programming_2021}, with an added column of ``Error - $5\%$''.
}
  \label{tab:speedups}

  \begin{center}
    \begin{tabular}{c|ccc}
      \toprule
      \textbf{Embedding Width} & \textbf{Error} & \textbf{Error - $5\%$} & \textbf{Speedup over W=128} \\
      \midrule
      128 & $8.9\%$ & $8.5\%$ & $1\times$ \\
      64 & $9.5\%$ & $9.0\%$ & $1.57\times$ \\
      32 & $10.1\%$ & $9.6\%$ & $2.01\times$ \\
      16 & $10.8\%$ & $10.3\%$ & $2.22\times$ \\
      \bottomrule
    \end{tabular}%
  \end{center}
\end{table}

\begin{figure}
  \begin{subfigure}[t]{0.45\textwidth}
    \input{code/huber.tex}
    \caption{Huber benchmark, which calculates a variant of the Huber loss for $\texttt{x} \in [-1, 1]$ and $\delta \in [0, 1]$.}
    \label{fig:huber}
  \end{subfigure}
  \hspace*{1em}
  \begin{subfigure}[t]{0.49\textwidth}
    \input{code/blackscholes.tex}
    \caption{BlackScholes benchmark, which performs a part of the Black Scholes option pricing model (with \lstinline|otype| positive for puts and negative for calls.)}
    \label{fig:blackscholes}
  \end{subfigure}
  \caption{Huber and BlackScholes benchmarks.}
\end{figure}

\paragraph{Luminance.}
The luminance benchmark is that of \Cref{sec:example}, and is presented in \Cref{fig:program-case-study}.
This benchmark has 3 paths: when $\texttt{sunPosition} < 0$ (path \texttt{ll} with complexity $0.01$), when $0 < \texttt{sunPosition} < 0.1$ (path \texttt{rl} with complexity $1.2$), and when $\texttt{sunPosition} > 0.1$ (path \texttt{rr} with complexity $9$).

Against the frequency baseline, compared to a predicted error improvement of $2.58\%$, complexity-guided sampling results in an improvement of $15.01\%$.
Against the uniform baseline, compared to a predicted error improvement of $6.97\%$, complexity-guided sampling results in an improvement of $15.17\%$.

\paragraph{Huber.}
\Cref{fig:huber} presents the Huber benchmark, which calculates the Huber loss for $\texttt{x} \in [-1, 1]$ and $\texttt{delta} \in [0, 1]$.
This benchmark has 3 paths: when $-\texttt{delta} < \texttt{x} < \texttt{delta}$ (path \texttt{ll} with complexity $9$), when $\texttt{x} < -\texttt{delta}$ (path \texttt{lr} with complexity $9$), and when $\texttt{x} > \texttt{delta}$ (path \texttt{r} with complexity $9$).

Against the frequency baseline, compared to a predicted error improvement of $0.49\%$, complexity-guided sampling results in an improvement of $8.15\%$.
Against the uniform baseline, compared to a predicted error improvement of $1.93\%$, complexity-guided sampling results in an improvement of $9.54\%$.

\paragraph{BlackScholes.}
\Cref{fig:blackscholes} presents the BlackScholes benchmark, which performs a part of the Black Scholes option pricing model (with \lstinline|otype| positive for puts and negative for calls), for inputs uniform in $[0, 1]$ (other than \texttt{otype}, which is uniform in $[-1, 1]$).
This benchmark is a fragment of the Black-Scholes benchmark in the AxBench benchmark suite~\citep{yazdanbakhsh_axbench_2017}.
This benchmark has 2 paths: when $\texttt{otype} > 0$ (path \texttt{l} with complexity $165.72$; for puts) and when $\texttt{otype} < 0$ (path \texttt{r} with complexity $485.23$; for calls).

Against the frequency baseline, compared to a predicted error improvement of $4.43\%$, complexity-guided sampling results in an improvement of $3.61\%$.
Against the uniform baseline, compared to a predicted error improvement of $1.30\%$, complexity-guided sampling results in an improvement of $4.00\%$.

\paragraph{Camera.}
\Cref{fig:camera} in \Cref{app:evaluation} presents the Camera benchmark, which performs a part of the conversion from blackbody radiator color temperature to the CIE 1931 x,y chromaticity approximation function, for inputs $\texttt{x} \in [-1, 1]$, $\texttt{y} \in [-1, 1]$, $\texttt{invKiloK} \in [0, 1]$, and $\texttt{T} \in [0.1, 0.5]$ (note that \texttt{T} is used exclusively to determine the path).
This benchmark is included in the Frankencamera platform~\citep{frankencamera}, and is based off of an implementation by \citet{Kang2002a}.
This benchmark has three paths: when $\texttt{T} < 0.2222$ (path \texttt{ll} with complexity $0.86$), when $0.2222 < \texttt{T} < 0.4$ (path \texttt{lrl} with complexity $0.81$), and when $0.4 < \texttt{T}$ (path \texttt{rrr} with complexity $9.53$).

Against the frequency baseline, compared to a predicted error improvement of $2.83\%$, complexity-guided sampling results in an improvement of $0.56\%$.
Against the uniform baseline, compared to a predicted error improvement of $0.22\%$, complexity-guided sampling results in an improvement of $1.36\%$.

\paragraph{EQuake.}
\Cref{fig:quake} in \Cref{app:evaluation} presents the EQuake benchmark, which computes the displacement of an object after one timestep in an earthquake simulation.
This benchmark is a fragment of the 183.equake benchmark in the SPECfp2000 benchmark suite~\citep{specfp}.
This benchmark has 2 paths: when $\texttt{t} > 0.5$ (path \texttt{l} with complexity $56.29$) and when $\texttt{t} < 0.5$ (path \texttt{r} with complexity $1169.50$).

Against both the frequency and uniform baselines, compared to a predicted error improvement of $7.45\%$, complexity-guided sampling results in an improvement of $2.25\%$.

\paragraph{Jmeint.}
\Cref{fig:jmeint} in \Cref{app:evaluation} presents the Jmeint benchmark, which calculates whether two 3D triangles intersect, and several auxiliary variables related to their intersection.
All inputs are sampled from $[-1, 1]$.
This benchmark is a fragment of the Jmeint benchmark in the AxBench benchmark suite~\citep{yazdanbakhsh_axbench_2017}.
This benchmark has 18 paths; each path has the same complexity of $72361000$, but with different frequencies.

Against the frequency baseline, compared to a predicted error improvement of $2.34\%$, complexity-guided sampling results in an improvement of $0.01\%$, a negligible change in error.
Against the uniform baseline, compared to a predicted error improvement of $8.44\%$, complexity-guided sampling results in an improvement of $1.02\%$.

We note that this benchmark has the highest complexity of any evaluated program (requiring more samples and still resulting in higher overall errors), and also that empirically some paths do appear to be significantly easier to learn despite the matching complexities.

\subsection{Analysis: Complexity-Guided Sampling Successes}
\label{sec:analysis-complexity-success}

In this section, we demonstrate examples of where the complexity-guided sampling technique results in significantly better error than~baselines.

\paragraph{Complex paths.}

The first case is when some paths are significantly more complex than others: neither the frequency nor the uniform baseline take into account path complexity, so we expect both baselines to undersample the complex path.

\Cref{fig:complexity-path-good} presents an example of such a case.
In this example, the complexity of the $\texttt{x} < 0.5$ path (\texttt{l}) is $137677$, while the complexity of the $\texttt{x} < 0.5$ path (\texttt{r}) is $57$.
The frequency of both the \texttt{l} and the \texttt{r} paths are $50\%$.
The complexity-guided sampling approach samples the \texttt{l} path with probability $93\%$ and the \texttt{r} path with probability $7\%$.

Against both the frequency and uniform baselines, compared to a predicted error improvement of $22.72\%$, complexity-guided sampling results in an improvement of $10.9\%$.

\paragraph{Skewed frequency distribution.}

The second case is when some paths are significantly more frequent than others.
This confers advantages over the uniform baseline, which does not take into account path frequency, and also over the frequency baseline, which does not take into account the functional form of the learning bound in \Cref{eq:sample-complexity} (i.e., that error decreases proportionally to the square root of the number of samples).

\Cref{fig:skewed-frequency} presents an example of such a case.
In this example, all paths have a complexity of $14$, while the frequencies are either $10\%$ (for paths \texttt{l}, \texttt{rl}, and \texttt{rrl}) or $70\%$ (for path \texttt{rrr}).
The complexity-guided sampling approach samples the $10\%$-frequency paths with probability $15\%$, and the $70\%$-frequency path with probability $55\%$.

Against the frequency baseline, compared to a predicted error improvement of $3.75\%$, complexity-guided sampling results in an improvement of $28.38\%$.
Against the uniform baseline, compared to a predicted error improvement of $14.08\%$, complexity-guided sampling results in an improvement of $20.76\%$.

Note that this type of path distribution (with rare paths below $1\%$ of the input data distribution) matches the distribution of paths in the renderer evaluation in \Cref{sec:renderer}, and results in a similar significant overperformance of the predicted error improvement.

\subsection{Analysis: Complexity-Guided Sampling Failures}

In this section, we demonstrate core examples of where the complexity-guided sampling technique results in significantly worse error than baselines.

\begin{figure}
  \begin{subfigure}[t]{0.45\textwidth}
    \input{code/skewed_complexity.tex}
    \caption{Success: synthetic example with skewed path complexities ($137678$ v.s. $57$) where complexity-guided sampling significantly improves error.}
    \label{fig:complexity-path-good}
  \end{subfigure}
  \hspace*{1em}
  \begin{subfigure}[t]{0.49\textwidth}
  \input{code/skewed_frequency.tex}
  \caption{Success: synthetic example with skewed path frequencies ($10\%$ v.s. $70\%$) where complexity-guided sampling significantly improves error.}
  \label{fig:skewed-frequency}
\end{subfigure}

\vspace*{1em}

  \begin{subfigure}[t]{0.3\textwidth}
    \input{code/complexity_imprecise.tex}
    \caption{Failure: synthetic example with a function on which the complexity bound is imprecise.}
    \label{fig:complexity-imprecise}
  \end{subfigure}
  \hspace*{0.75em}
  \begin{subfigure}[t]{0.31\textwidth}
    \input{code/complexity_nonuniformity.tex}
    \caption{Failure: synthetic example with a function that has different effective complexities across scales.}
    \label{fig:complexity-nonuniformity}
  \end{subfigure}
  \hspace*{0.75em}
  \begin{subfigure}[t]{0.31\textwidth}
  \input{code/analysis_bad.tex}
  \caption{Failure: synthetic example with a function on which the \lang{} analysis is imprecise.}
  \label{fig:analysis-bad}
\end{subfigure}

\caption{Examples of complexity-guided sampling successes and failures.}
\vspace*{-0.5em}
\end{figure}

\paragraph{Complexity imprecision.}

The first case is when the complexity results in too loose of an upper bound on the resulting error of a surrogate of that function.
In this case, the complexity-guided sampling approach can oversample from the corresponding path.

\Cref{fig:complexity-imprecise} presents an example of such a case.
In this example, the complexity of the \texttt{l} path is $18638$ and the complexity of the \texttt{r} path is $16$.
Though we are not aware of any tighter bounds on the complexity of learning $\sin\mleft(4x\mright)$ for $x \in \mleft[0.5, 1\mright]$, in practice we find that neural networks are able to learn this function to low error with relatively few samples.

The frequency of each path is $50\%$.
The complexity-guided sampling approach samples the \texttt{l} path with probability $90.9\%$ and the \texttt{r} path with probability $9.1\%$.
Against both the frequency and uniform baselines, compared to a predicted error improvement of $20.88\%$, complexity-guided sampling results in a change of $-92.59\%$, a significant increase in error.

\paragraph{Nonuniform complexity.}

The second case is when the complexity of a learning function varies significantly across different scales.
In this case, the complexity-guided sampling approach can oversample from the corresponding path.

\Cref{fig:complexity-nonuniformity} presents an example of such a case.
In this example, the complexity of the \texttt{l} path is $12129$ and the complexity of the \texttt{r} path is $2.38$.
This causes an issue with the complexity-guided sampling because with a large target error (e.g., $\epsilon > 0.01$), the \texttt{l} path is essentially zero (and therefore should have low complexity).
However, with a small target error (e.g., $\epsilon < 0.00001$), the \texttt{l} path is very complex.
Because the sample complexity bounds themselves are scale-independent upper bounds, they do not by default incorporate this knowledge.

The frequency of each path in this example is $50\%$.
The complexity-guided sampling approach samples the \texttt{l} path with probability $92.9\%$ and the \texttt{r} path with probability $7.1\%$.
Against both the frequency and the uniform baselines, compared to a predicted error improvement of $22.7\%$, complexity-guided sampling results in a change of $-351\%$, a $3.5\times$ increase~in~error.

\paragraph{Analysis imprecision.}

The third case is when \lang{}'s analysis of the complexity is imprecise: \Cref{thm:analysis-upper-bound} proves that \lang{}'s complexity analysis computes an upper bound on the tilde, but this upper bound may also be loose (as discussed in \Cref{sec:precision}).
In this case, the complexity-guided sampling approach can oversample from the corresponding path.

\Cref{fig:analysis-bad} presents an example of such a case.
In this example, the calculated complexity of the \texttt{l} path is $161604$ and the calculated complexity of the \texttt{r} path is $56.6$.
If we were to perform algebraic simplification (which \lang{} does not), we would find the complexity of the \texttt{l} path to instead be $4$.

The frequency of each path is $50\%$.
With \lang{}'s computed complexities, the complexity-guided sampling approach samples the \texttt{l} path with probability $93.3\%$ and the \texttt{r} path with probability $7.7\%$.
Against both the frequency and the uniform baselines, compared to a predicted error improvement of $23.03\%$, complexity-guided sampling results in a change of $-491.22\%$, a $5\times$ increase~in~error.

}

{
\let\section\subsection
\let\subsection\subsubsection
\let\subsubsection\paragraph

\newcommand\picwidth{0.49\textwidth}
\begin{figure*}
  \begin{center}
    \begin{subfigure}[b]{\picwidth}
      \centering
      \includegraphics[width=\textwidth]{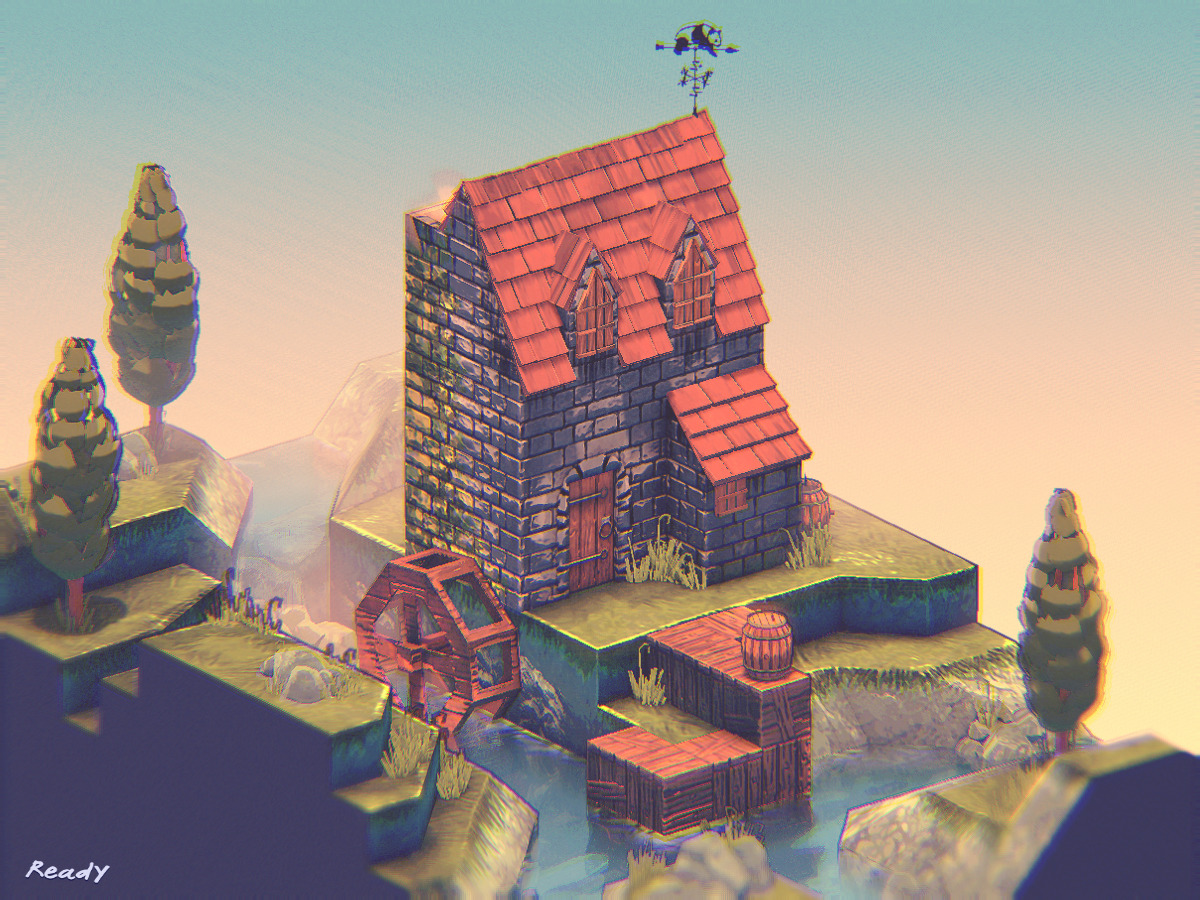}
      \caption[base-day]{Ground-truth front-day scene.}
      \label{fig:base-day}
    \end{subfigure}
    \begin{subfigure}[b]{\picwidth}
      \centering
      \includegraphics[width=\textwidth]{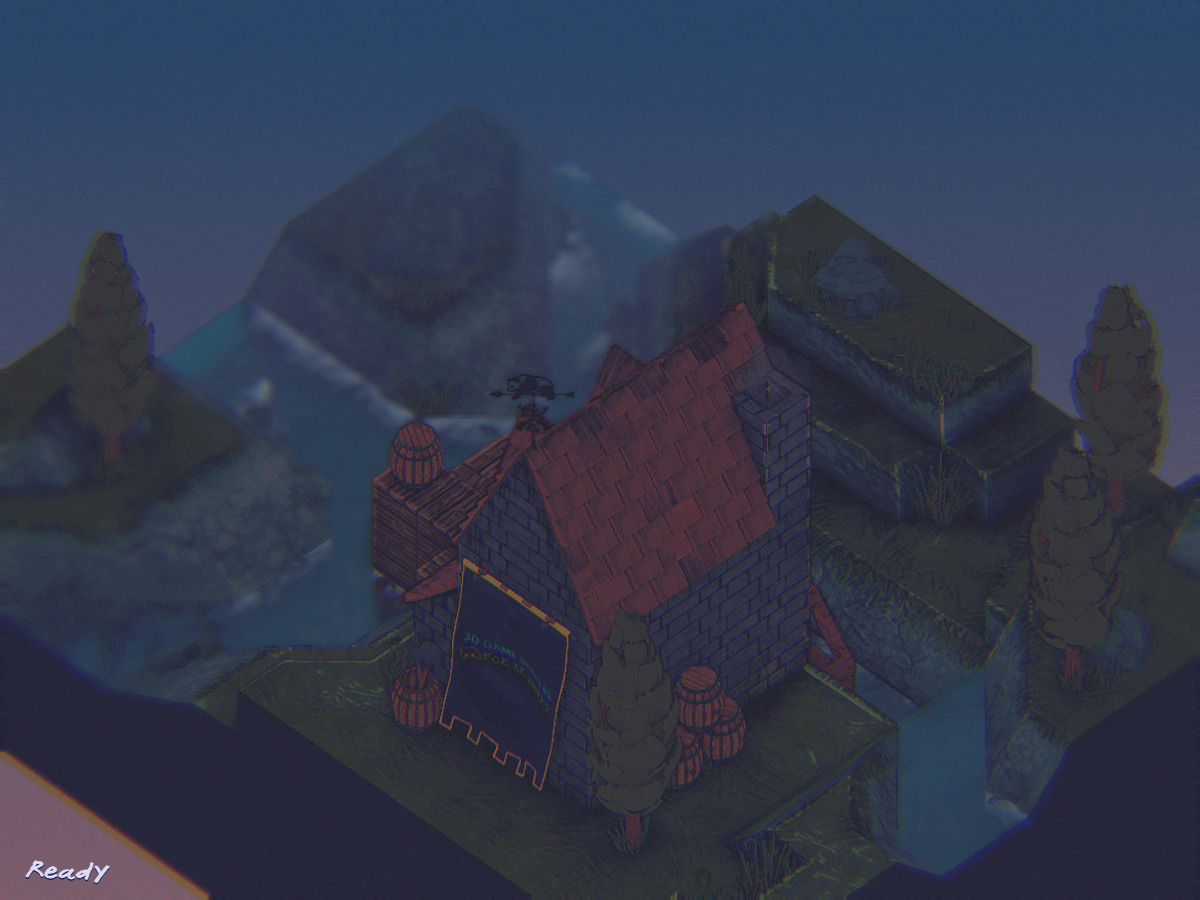}
      \caption[top-night]{Ground-truth top-night scene.}
      \label{fig:top-night}
    \end{subfigure}
  \end{center}
  \begin{center}
  \begin{subfigure}{0.323\textwidth}
    \includegraphics[width=\textwidth]{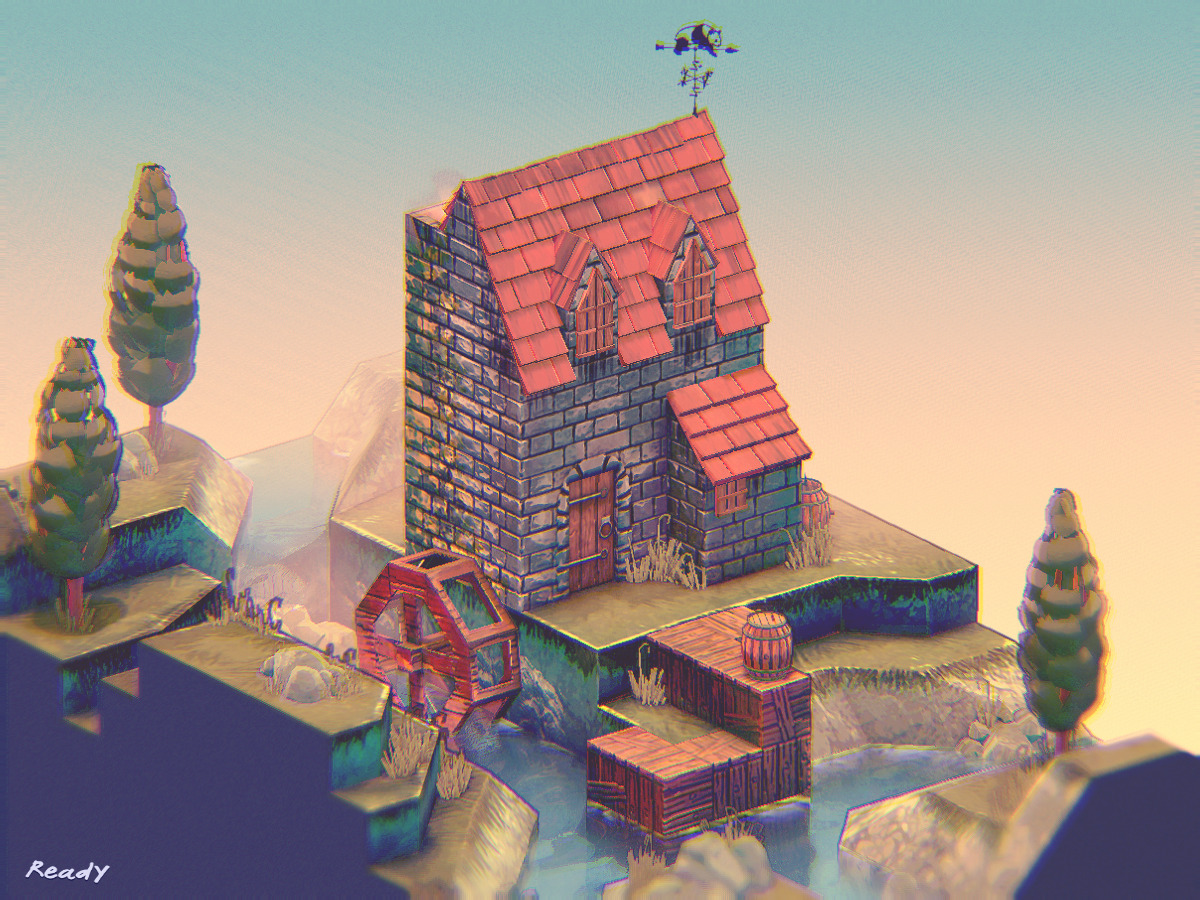}
  \end{subfigure}
  \begin{subfigure}{0.323\textwidth}
    \includegraphics[width=\textwidth]{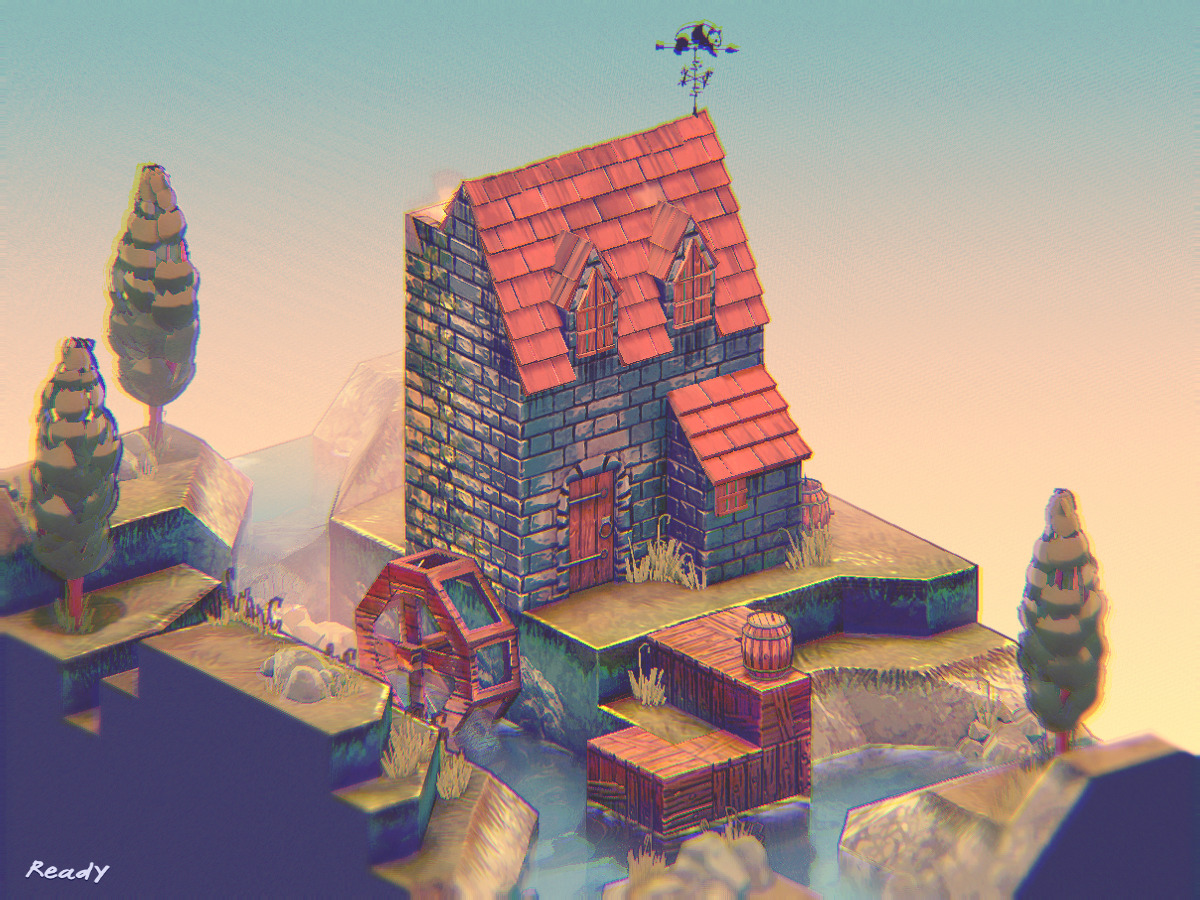}
  \end{subfigure}
  \begin{subfigure}{0.323\textwidth}
    \includegraphics[width=\textwidth]{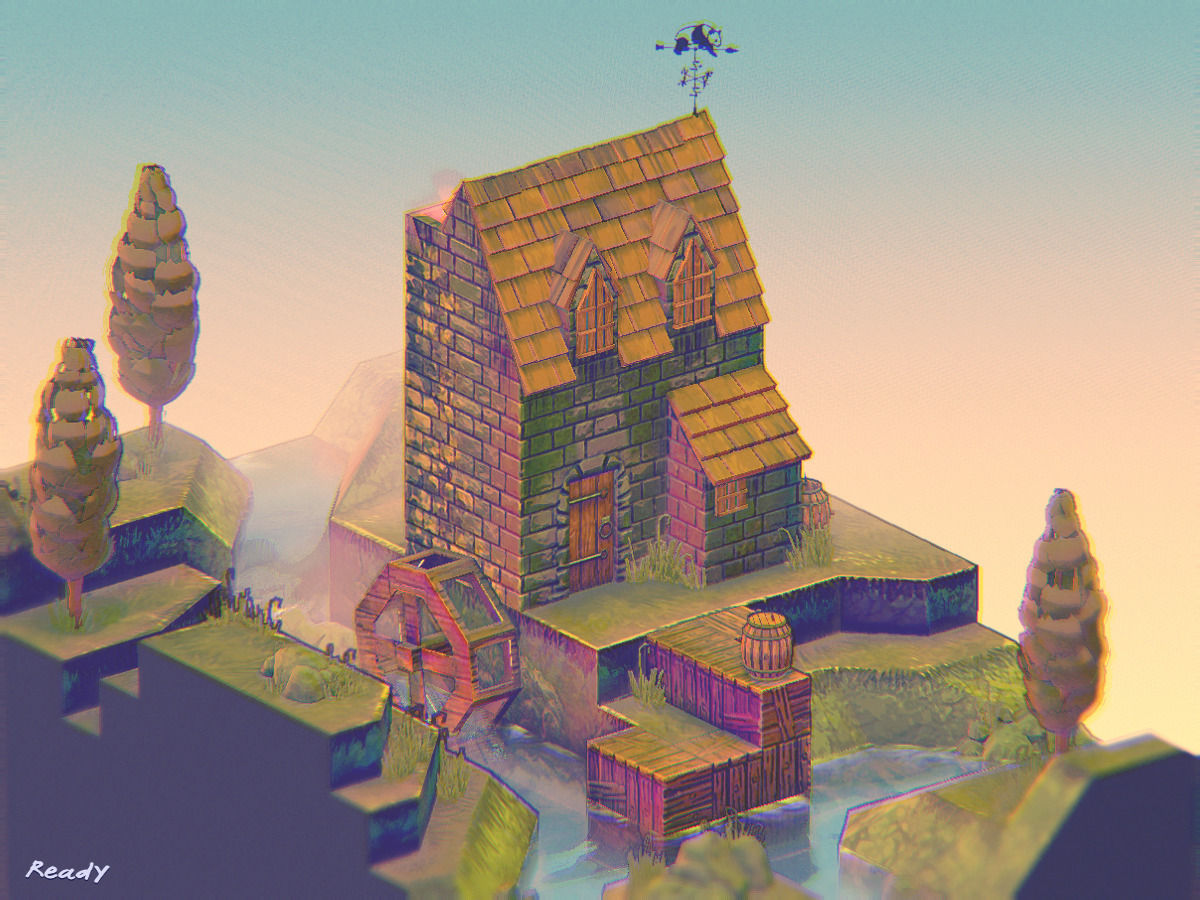}
  \end{subfigure}
\\
  \begin{subfigure}{0.323\textwidth}
    \includegraphics[width=\textwidth]{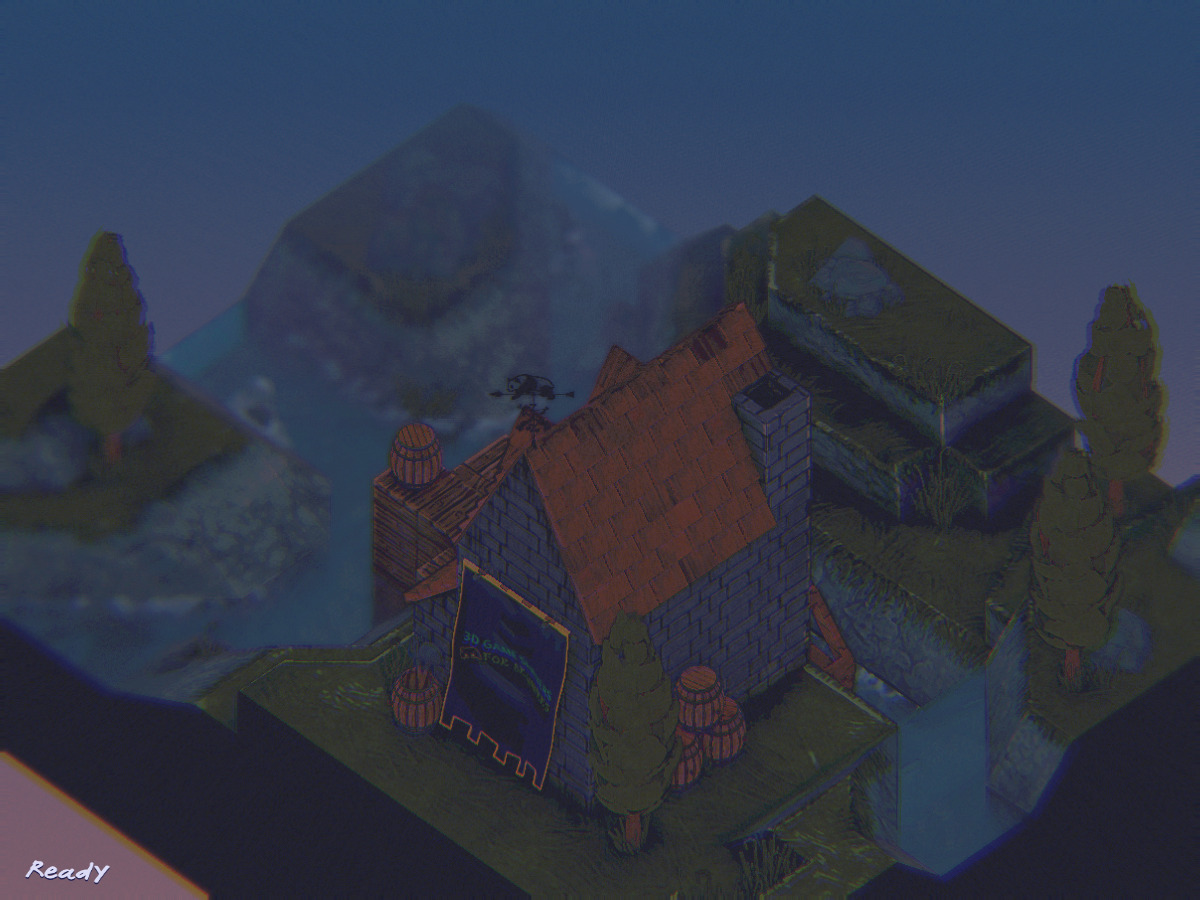}
    \caption{\mOptimal{} surrogate.}
    \label{fig:optimal-surrogate-renderings}
  \end{subfigure}
  \begin{subfigure}{0.323\textwidth}
    \includegraphics[width=\textwidth]{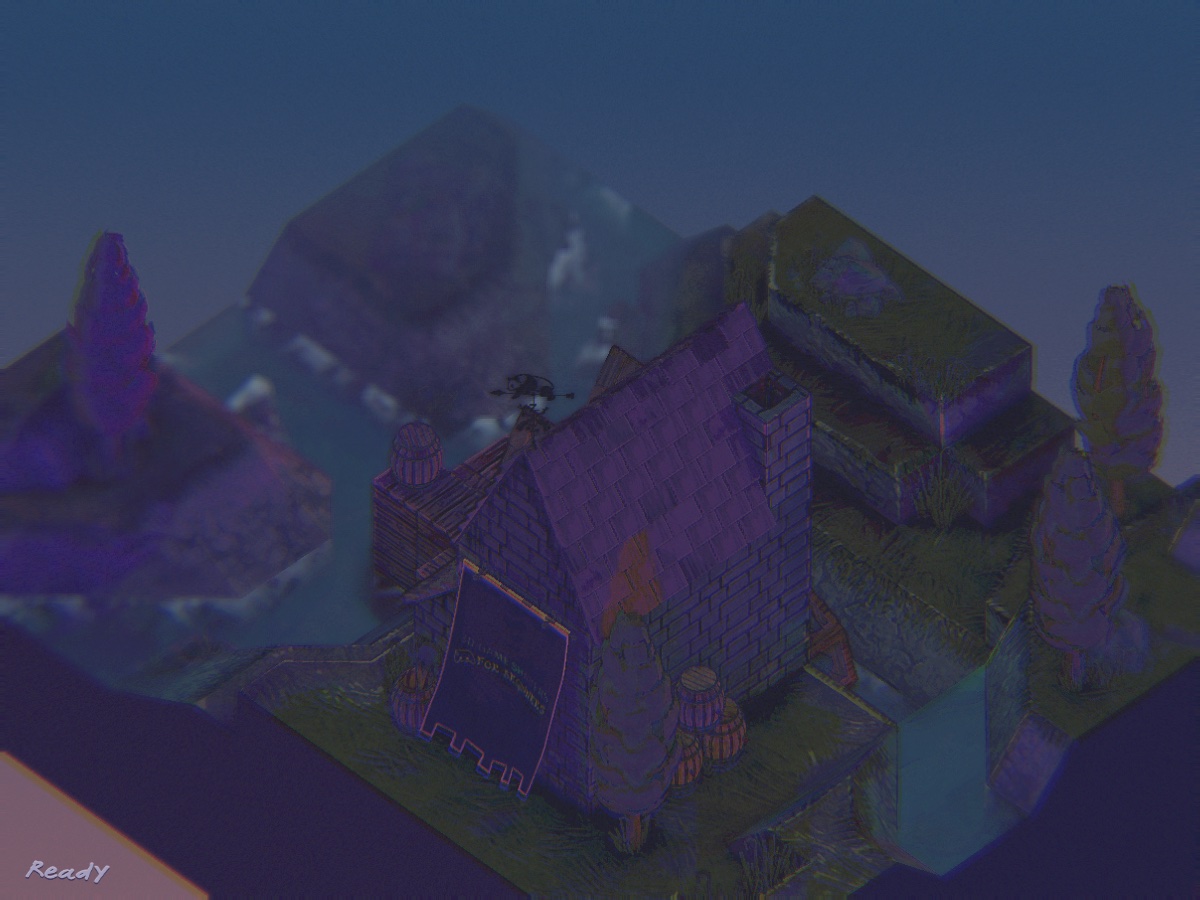}
    \caption{Frequency-based surrogate.}
    \label{fig:frequency-surrogate-renderings}
  \end{subfigure}
  \begin{subfigure}{0.323\textwidth}
    \includegraphics[width=\textwidth]{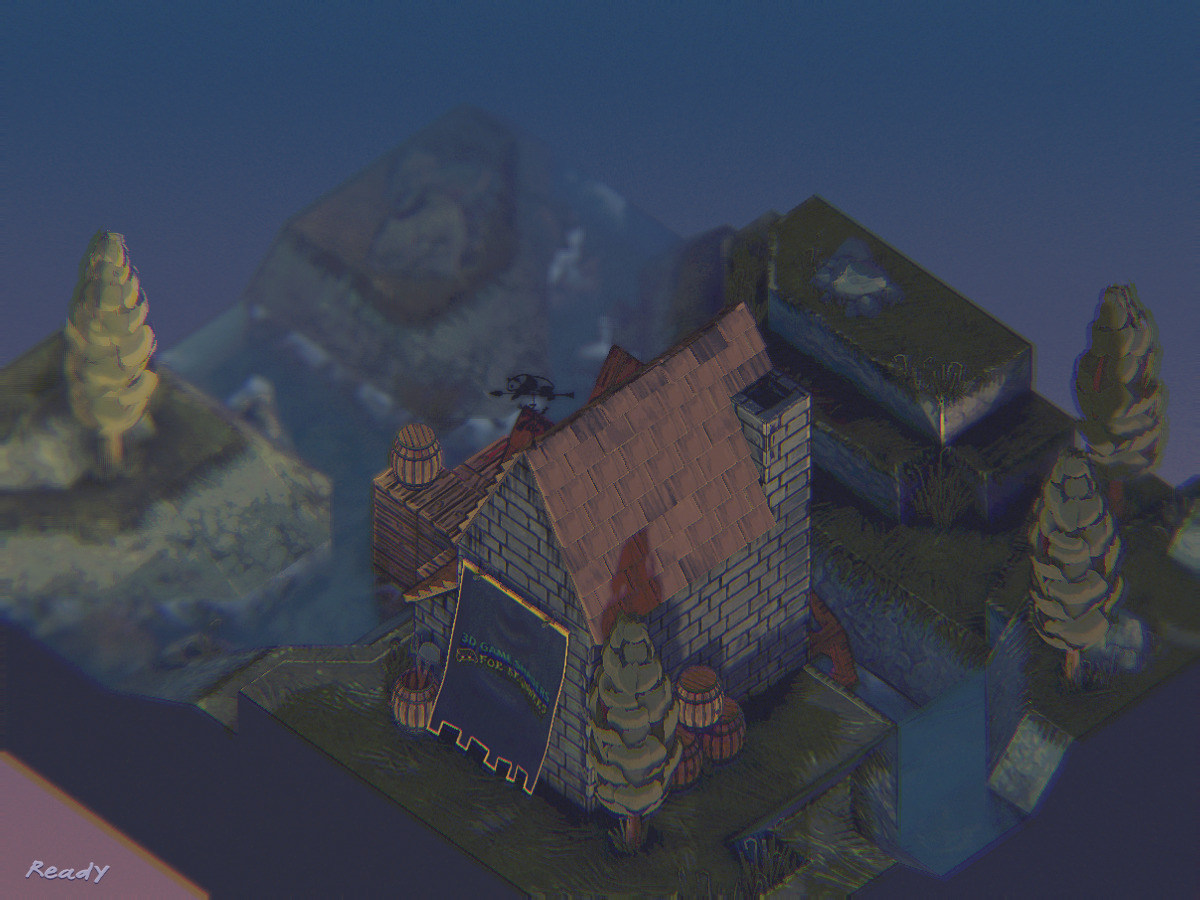}
    \caption{Uniform surrogate.}
    \label{fig:uniform-surrogate-renderings}
  \end{subfigure}
\end{center}
\caption{Ground-truth (top) and surrogate renderings (bottom) of scenes generated by the renderer.}
\label{fig:surrogate-renderings}
\end{figure*}

\section{Renderer Demonstration}
\label{sec:renderer}

In this section we present a case study of our
\moptimal{} sampling results and
complexity analysis.
The program under study is a demonstration 3D renderer~\citep{renderer}, \NA{such as forms the core} of a graphics rendering pipeline for a movie or 3D game engine~\citep{renderman,real_time_3d_rendering}.
\Cref{fig:base-day,fig:top-night} show scenes that the renderer generates.
We demonstrate that the sampling and analysis techniques in \Cref{sec:formalism,sec:language} consistently result in more accurate surrogates than those trained using baseline distributions (the frequency distribution of paths and~the~uniform~distribution).

Compared to training surrogates on the frequency distribution of paths, \moptimal{} sampling decreases error by $17\%$.
Compared to training on the uniform distribution of paths, \moptimal{} sampling decreases error by $44\%$.
These improvements in error correspond to perceptual improvements in the generated images, as shown in \Cref{fig:optimal-surrogate-renderings,fig:frequency-surrogate-renderings,fig:uniform-surrogate-renderings}.

\begin{table*}
  \caption{Top: the identifier, lines of code, complexity, and description of each path present in our datasets.
    Bottom: the distribution (abbreviated distr.) of each path across each dataset: the frequency (Freq.) of each observed path, and the \moptimal{} sampling rate (Com.) of that path.}
  \label{tab:renderer-paths-app}

  \begin{center}
    \resizebox{\textwidth}{!}{
    \begin{tabular}{cc|ccccccccc}
      \toprule

      \multicolumn{2}{c|}{\textbf{Path}}
& \textbf{lrrllr} & \textbf{lrrlrl} & \textbf{lrrlrr} & \textbf{lrrrlr} & \textbf{lrrrrl} & \textbf{lrrrrr} & \textbf{rrrllr} & \textbf{rrrlrl} & \textbf{rrrlrr} \\
      \midrule
      \multicolumn{2}{c|}{\textbf{Lines of Code}} & 17 & 17 & 17 & 18 & 18 & 18 & 17 & 17 & 17 \\
      \multicolumn{2}{c|}{\textbf{Complexity}} &
                                                 6115 & 5806 & 6272 & 6401 & 6084 & 6562 & 8804 & 8433 & 8993 \\
      \multicolumn{2}{c|}{\multirow{2}{*}{\textbf{Description}}}
        & \multirow{2}{1.2cm}{\centering Twilight \\ Water} & \multirow{2}{1.2cm}{\centering Twilight \\ Smoke} & \multirow{2}{1.2cm}{\centering Twilight \\ Solids}
        & \multirow{2}{1.2cm}{\centering Nighttime \\ Water} & \multirow{2}{1.2cm}{\centering Nighttime \\ Smoke} & \multirow{2}{1.2cm}{\centering Nighttime \\ Solids}
        & \multirow{2}{1.2cm}{\centering Daytime \\ Water} & \multirow{2}{1.2cm}{\centering Daytime \\ Smoke} & \multirow{2}{1.2cm}{\centering Daytime \\ Solids}
      \\
      \multicolumn{2}{c|}{}
      \\

      \toprule
      \textbf{Dataset} & \multicolumn{1}{c|}{\textbf{Distr.}}
                         & \textbf{lrrllr} & \textbf{lrrlrl} & \textbf{lrrlrr} & \textbf{lrrrlr} & \textbf{lrrrrl} & \textbf{lrrrrr} & \textbf{rrrllr} & \textbf{rrrlrl} & \textbf{rrrlrr}
      \\
      \midrule

Front Day & Freq. &  &  &  &  &  &  & 5.0\% & 7.9\% & 87.1\%\\
& Com. &  &  &  &  &  &  & 11.0\% & 14.7\% & 74.3\%\\
\midrule
Front Night & Freq. &  &  &  & 5.0\% & 7.9\% & 51.4\% &  &  & 35.6\%\\
& Com. &  &  &  & 8.9\% & 11.9\% & 42.4\% &  &  & 36.9\%\\
\midrule
Top Day & Freq. &  &  &  &  &  &  & 6.7\% & 13.1\% & 80.1\%\\
& Com. &  &  &  &  &  &  & 12.9\% & 19.8\% & 67.4\%\\
\midrule
Top Night & Freq. & 0.16\% & 0.06\% & 1.2\% & 0.3\% & 0.1\% & 2.4\% & 6.3\% & 12.9\% & 76.5\%\\
& Com. & 0.87\% & 0.45\% & 3.3\% & 1.4\% & 0.7\% & 5.3\% & 11.1\% & 17.7\% & 59.2\%\\
\midrule
Front & Freq. &  &  &  & 2.5\% & 4.0\% & 25.7\% & 2.5\% & 4.0\% & 61.4\%\\
& Com. &  &  &  & 5.2\% & 7.0\% & 24.9\% & 5.8\% & 7.8\% & 49.4\%\\
\midrule
Top & Freq. & 0.08\% & 0.03\% & 0.6\% & 0.2\% & 0.1\% & 1.2\% & 6.5\% & 13.0\% & 78.3\%\\
& Com. & 0.56\% & 0.29\% & 2.1\% & 0.9\% & 0.5\% & 3.5\% & 11.7\% & 18.4\% & 62.1\%\\
\midrule
Day & Freq. &  &  &  &  &  &  & 5.9\% & 10.5\% & 83.6\%\\
& Com. &  &  &  &  &  &  & 11.9\% & 17.4\% & 70.7\%\\
\midrule
Night & Freq. & 0.08\% & 0.03\% & 0.6\% & 2.7\% & 4.0\% & 26.9\% & 3.1\% & 6.5\% & 56.1\%\\
& Com. & 0.50\% & 0.26\% & 1.9\% & 5.2\% & 6.7\% & 24.4\% & 6.4\% & 10.3\% & 44.3\%\\
\midrule
All & Freq. & 0.04\% & 0.02\% & 0.3\% & 1.3\% & 2.0\% & 13.5\% & 4.5\% & 8.5\% & 69.8\%\\
        & Com. & 0.33\% & 0.17\% & 1.2\% & 3.4\% & 4.4\% & 16.0\% & 8.5\% & 12.8\% & 53.2\%\\

      \bottomrule
    \end{tabular}
    }
  \end{center}
\end{table*}

\begin{table*}
  \captionof{table}{Average decrease in error across all budgets from using \moptimal{} sampling compared to baselines on each dataset (higher values means \moptimal{} sampling has lower error).}
  \label{fig:renderer-surrogate-errors-table}
  \begin{center}
      \begin{tabular}{c|ccccccccc|c}
      \toprule
\textbf{Baseline} & \textbf{\makecell{Front\\Day}} & \textbf{\makecell{Front\\Night}} & \textbf{\makecell{Top\\Day}} & \textbf{\makecell{Top\\Night}} & \textbf{Front} & \textbf{Top} & \textbf{Day} & \textbf{Night} & \textbf{All} & \textbf{Mean} \\\midrule
        Frequency &
                    5\% & -3\% & -1\% & 48\% & 3\% & 31\% & 2\% & 21\% & 27\% & 17\% \\
        Uniform &
        39\% & 31\% & 36\% & 40\% & 42\% & 61\% & 34\% & 52\% & 52\% & 44\% \\
      \end{tabular}
  \end{center}
\end{table*}

\subsection{Program Under Study}

The full renderer program is a 2750 lines-of-code C++ program, which invokes 38 different GLSL shader programs totaling 2446 lines of code.
We learn a surrogate of a section of one core shader, totaling 60 lines of code.%
\footnote{Lines 278 through 337 of \url{https://github.com/lettier/3d-game-shaders-for-beginners/blob/29700/demonstration/shaders/fragment/base.frag}.}
\Cref{listing:renderer-full-learnguage} in \Cref{app:renderer} presents the code for the renderer case study.

\paragraph{Input-output specification.}
This program is a shader which assigns colors to pixels in the image based on the scene geometry, materials, lights, and other properties.
The program is called for each pixel that is rendered in the image.
Each invocation the of program takes as input a set of 11 fixed-size vectors, totaling 35 inputs.
The program returns as output a set of 4 fixed-size vectors, totaling 8 outputs. %
These outputs represent two RGBA colors, the first representing the base color of the pixel, and the second representing the color and intensity of a specular map~at~that~pixel.

\paragraph{Scenes and datasets.}
We evaluate the renderer on four different scenes, which we combine into nine different datasets.
\Cref{fig:base-day,fig:top-night} present two of the four different scenes under consideration; the four scenes are all combinations of views from the front and top, during the day and night.
We combine these scenes into nine datasets: a dataset with each scene, a dataset combining each scene from each angle (front day and front night, top day and top night), a dataset combining each scene from each time of day, and a dataset combining all scenes.
\Cref{fig:renderer-scenes-app} in \Cref{app:renderer} presents the full set~of~scenes~under~study.

\paragraph{Paths.}
The program is a conjunction of 48 different paths, 9 of which are exercised by the renderer.
The top part of \Cref{tab:renderer-paths-app} presents statistics about the paths under study, showing the identifier (a trace of \texttt{l} and \texttt{r} characters denoting which branch of each if statement the path takes), the lines of code in the corresponding trace, and the complexity of the corresponding trace according to the analysis in \Cref{sec:analysis}.
The paths are broken up into a path for rendering smoke particles from the chimney, water particles in the river, and the solids of the ground and house.
Each set of paths is duplicated for twilight, nighttime, and daytime.
Within each time of day, the smoke paths are the least complex, followed by water then solids.
Across different times, twilight paths are the least complex, followed by nighttime then daytime.

\begin{figure}
  \begin{center}
    \begin{subfigure}[b]{0.32\textwidth}
      \centering
      \includegraphics[width=\textwidth]{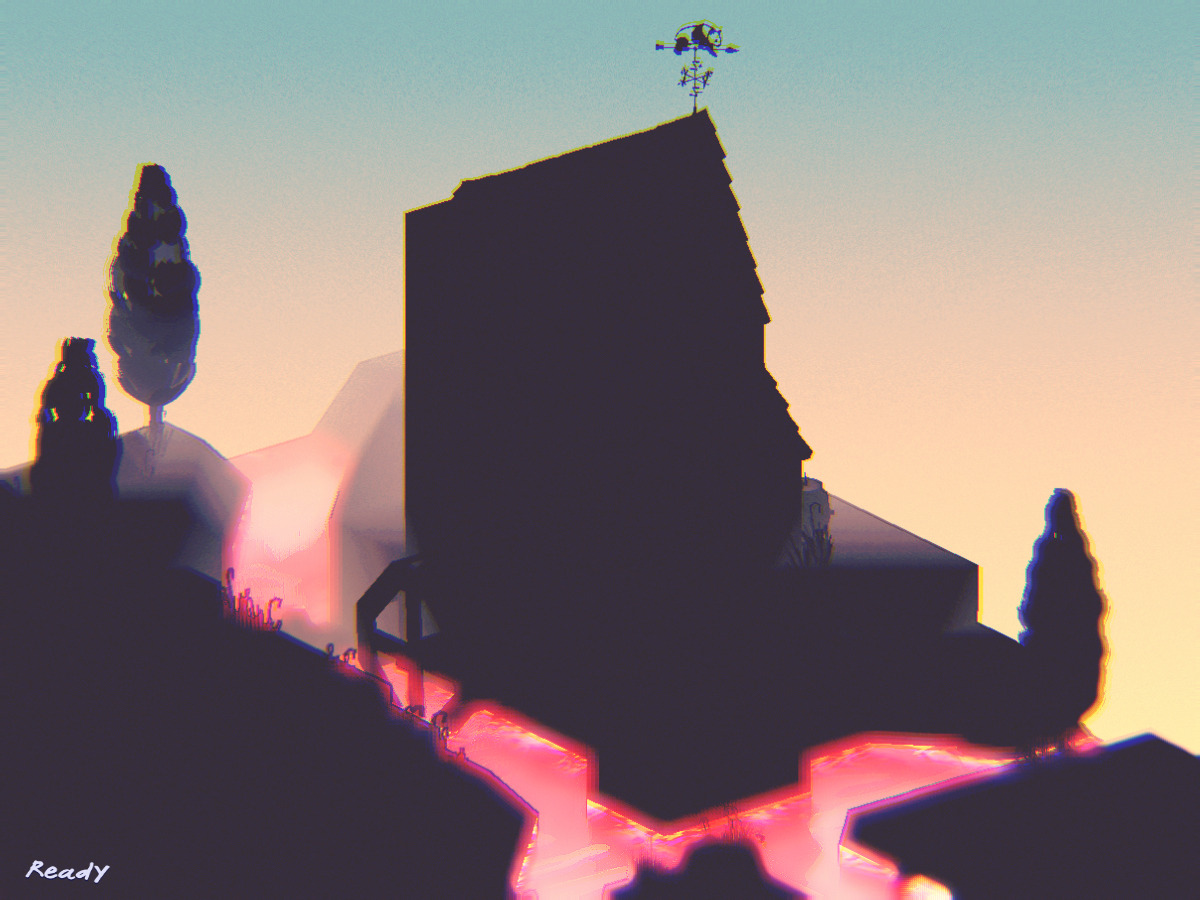}
      \caption[daytime-water]{Water path.}
      \label{fig:daytime-water}
    \end{subfigure}
    \begin{subfigure}[b]{0.32\textwidth}
      \centering
      \includegraphics[width=\textwidth]{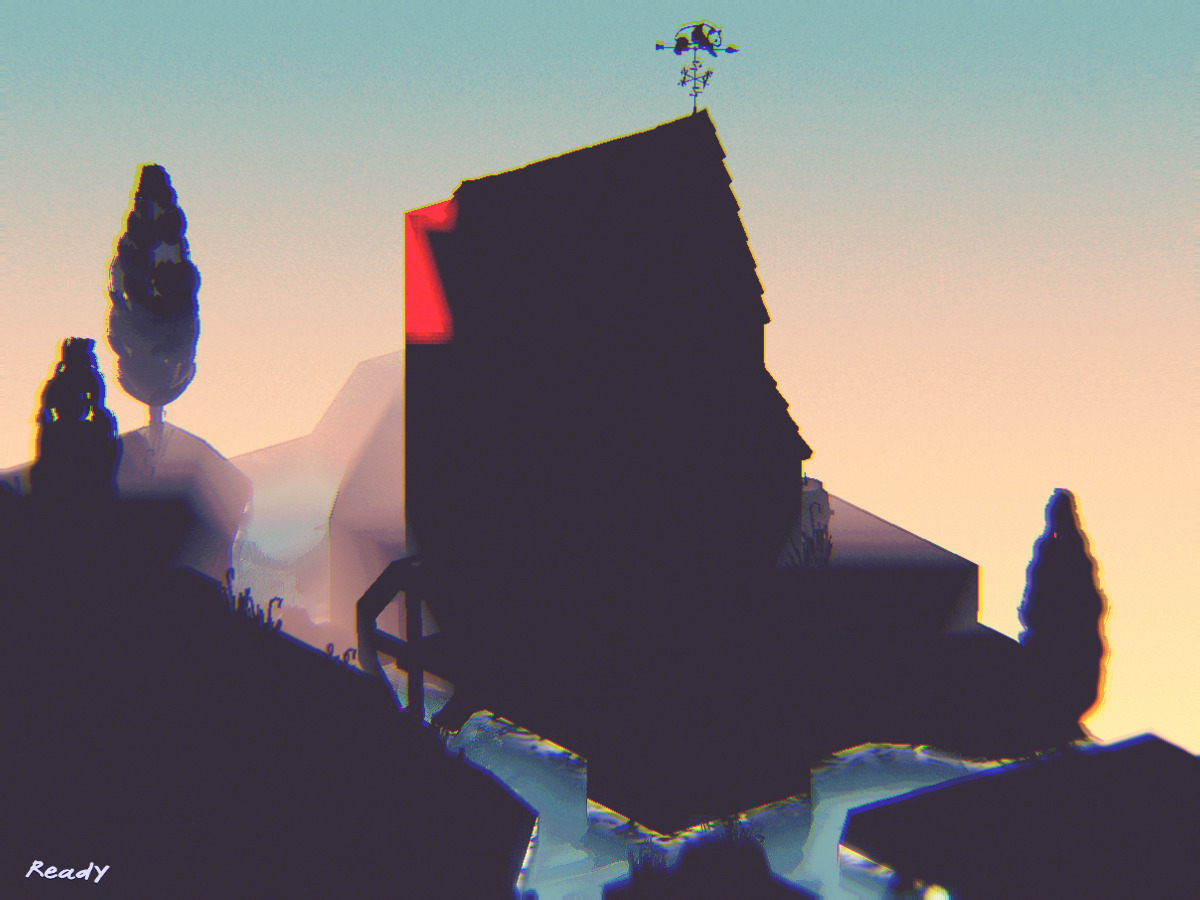}
      \caption[daytime-smoke]{Smoke path.}
      \label{fig:daytime-smoke}
    \end{subfigure}
    \begin{subfigure}[b]{0.32\textwidth}
      \centering
      \includegraphics[width=\textwidth]{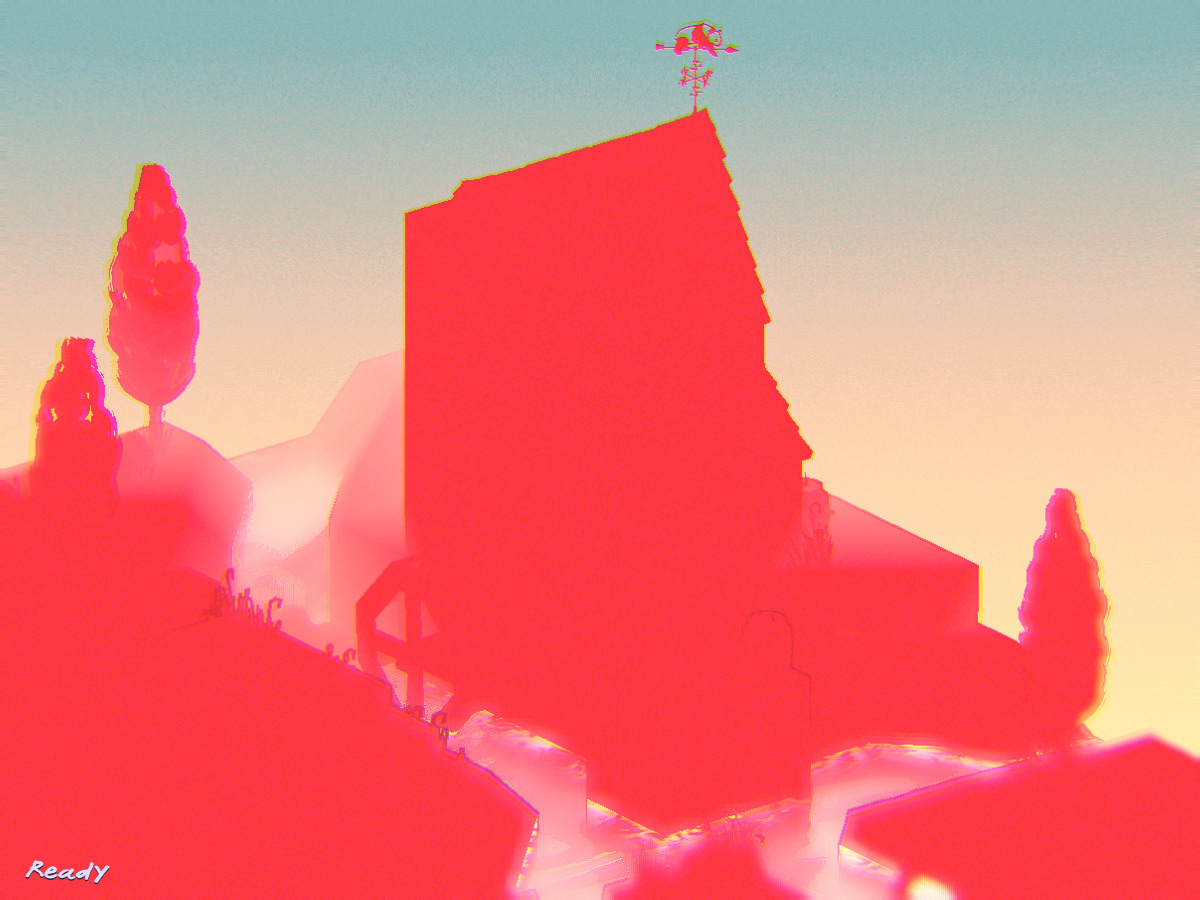}
      \caption[daytime-solids]{Solids path.}
      \label{fig:daytime-solids}
    \end{subfigure}
  \end{center}
  \caption{Daytime scene with each different path highlighted red, and all others black.}
  \label{fig:path-image-comparison-app}
\end{figure}

\Cref{fig:path-image-comparison-app} shows side-by-side comparisons of the three classes of paths: water, smoke, and solids.
In each of these images, one path returns red for all pixels while the other paths return black for all pixels.
The base scene is the front daytime scene in \Cref{fig:base-day}.

\Cref{tab:renderer-paths-app} also presents the observed distribution and the \moptimal{} distribution of paths for each dataset.
In general, the twilight paths are rarer than the nighttime paths, which are rarer than the daytime paths: this is because data collection for the nighttime scenes extends through twilight and into the morning.
For all datasets, the smoke paths are rarer than the water paths, which are in turn rarer than the solids paths; this is purely due to the scene geometry.

\subsection{Surrogate Training and Deployment Methodology}
To create and deploy a surrogate of the renderer, we train a surrogate of each path,
then create a stratified surrogate which branches on the set of path conditions and applies~the~corresponding~surrogate.

Our goal is to compare the errors achieved by training on the \moptimal{} distribution of paths against those of baseline distributions of paths.
We compare the approaches across different training datasets, different total numbers of training data points, and evaluating across different evaluation sets, all with multiple trials.

For the design of each surrogate, we use a simple MLP architecture with a single hidden layer of 512 units and a ReLU activations.
This architecture matches that of \citet{agarwala_monolithic_2021}, except using 512 rather than 1000 hidden units (we found that the accuracy of each path surrogate plateaued by 512 hidden units).
We train the surrogate using the Adam optimizer with a learning rate of $0.0001$ and a batch size of 256 for 50,000 steps.
We run 5 trials of all experiments, and report the error as an arithmetic mean when reported in isolation for a given surrogate (e.g., as in \Cref{fig:renderer-surrogate-errors}) and a geometric mean error when comparing relative error rates across different settings (e.g., as in the headline error improvement numbers in \Cref{fig:renderer-surrogate-errors-table}).

\subsection{Surrogate Errors}
\label{sec:renderer-errors}
\Cref{fig:renderer-surrogate-errors-table} presents the geometric mean decrease in error of using \moptimal{} path sampling compared to each baseline, on each dataset.
Across most datasets, \moptimal{} path sampling results in lower error than both frequency-based path sampling and uniform path sampling.
On datasets with few paths (Front Day) and in which all paths are well represented (minimum 5\% frequency), the gap is minimal, and frequency-based path sampling matches or outperforms \moptimal{} path sampling.
On datasets with more and rarer paths (e.g., Top Night), the gap widens and \moptimal{} path sampling outperforms frequency-based path sampling; we discuss this phenomenon in \Cref{sec:analysis-complexity-success}.
On all datasets, \moptimal{} path sampling outperforms uniform path sampling.

\begin{figure}
  \begin{center}
    \begin{subfigure}[b]{0.49\textwidth}
      \centering
      \includegraphics[width=\textwidth]{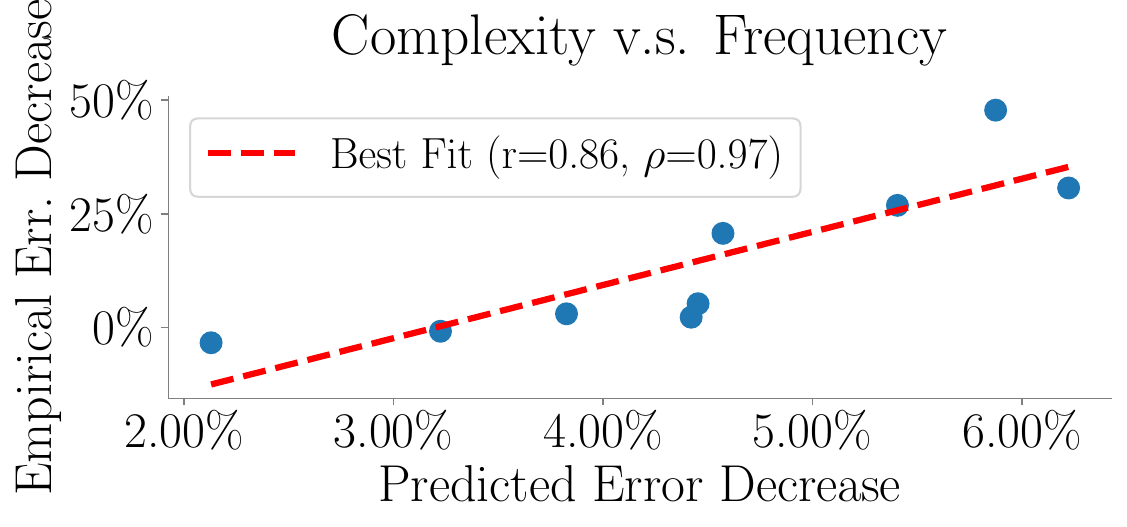}
      \caption{Correlation between predicted and empirical surrogate error decrease between surrogates trained with complexity-guided path sampling compared to frequency-based path sampling.}
    \end{subfigure}
    \hfill
    \begin{subfigure}[b]{0.49\textwidth}
      \centering
      \includegraphics[width=\textwidth]{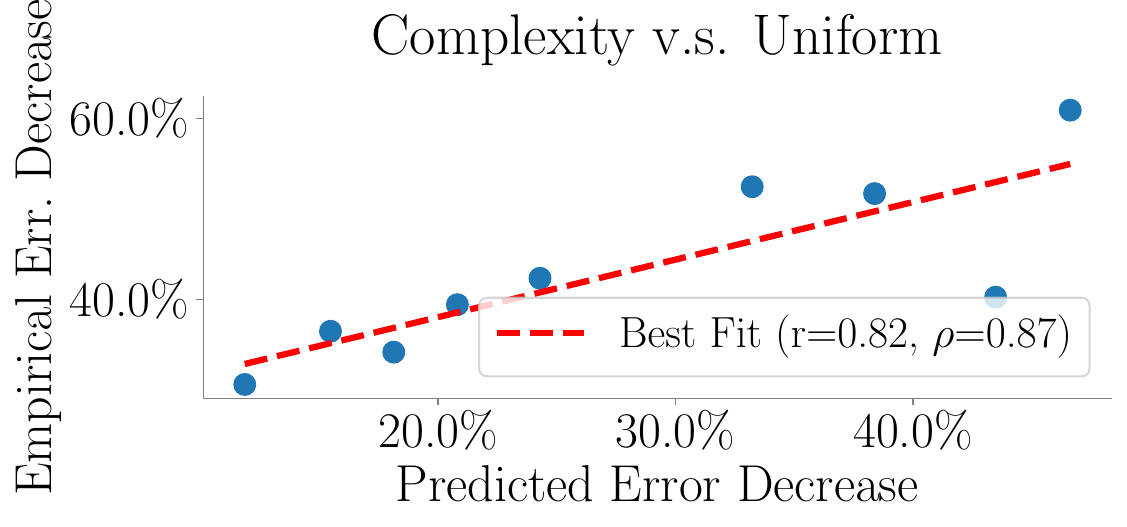}
      \caption{Correlation between predicted and empirical surrogate error decrease between surrogates trained with complexity-guided path sampling compared to uniform path sampling.}
    \end{subfigure}
  \end{center}
  \caption{Correlation between predicted and empirical surrogate error decreases for the renderer case study.}
  \label{fig:renderer-surrogate-correlation}
\end{figure}

\Cref{fig:renderer-surrogate-correlation} presents the correlation between the predicted error (\Cref{eq:predicted-error}) and the observed empirical error for each dataset, showing a strong correlation.
The left plot shows this correlation for surrogates trained with frequency-based path sampling, and the right plot shows this correlation for surrogates trained with uniform path sampling.
The x axis is the decrease in predicted error (specifically, the mean percentage error between the predicted error for complexity-guided sampling and the sampling method in the plot), and the y axis is the decrease in empirical error (the mean percentage error between the error observed from complexity-guided sampling and the sampling method in the plot).
Each point represents a different dataset (e.g., front-day, top-night, etc.).
The red dotted line shows the line of best fit.
For the frequency-based surrogates, the Pearson correlation is $r=0.86$ and the Spearman correlation is $\rho=0.97$.
For the uniform surrogates, the Pearson correlation is $r=0.82$ and the Spearman correlation is $\rho=0.87$.

\begin{figure}
  \begin{center}
    \includegraphics[width=\textwidth]{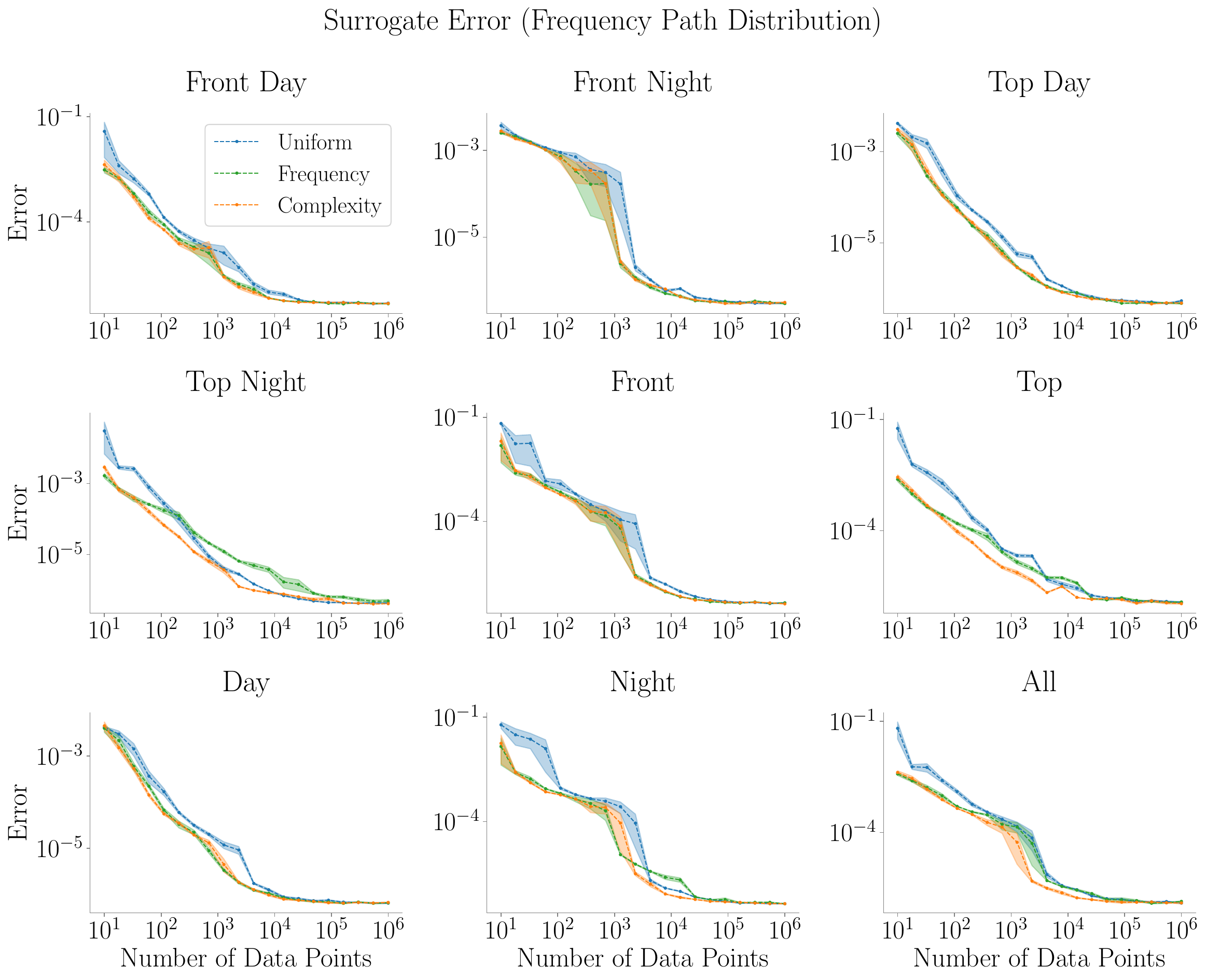}
  \end{center}
  \caption{Errors of stratified surrogates of each dataset.}
  \label{fig:renderer-surrogate-errors}
\end{figure}

\Cref{fig:renderer-surrogate-errors} presents the error of surrogates on each dataset.
Each plot shows the error for a different dataset.
Each plot has three different lines, respectively showing the error of each surrogate training distribution (\moptimal{}, frequency-based, and uniform).
Each x axis is the total training data budget.
Each y axis is the error of the resulting stratified surrogate.

For a given dataset budget, our complexity-guided sampling approach results in lower error than baseline sampling approaches.
Generally, increasing this dataset budget also results in lower error for all approaches.
These two approaches to decreasing surrogate error (better sampling techniques and sampling more data) are not in conflict with each other.

In \Cref{fig:renderer-surrogate-errors}, the sampling approaches converge in error with large dataset budgets.
This convergence is due to our evaluation methodology: following the convention of prior work which established these bounds~\citep{arora_fine-grained_2019,agarwala_monolithic_2021}, we use a fixed width for all neural networks, resulting in neural networks that saturate in error with large datasets.
An alternative methodology would be to grow the width of the neural network along with the size of the dataset, requiring a full hyperparameter search at each network scale.
With such a methodology, the errors would not plateau in the way that they do in \Cref{fig:renderer-surrogate-errors}.

\subsection{Visualization}
\label{sec:renderer-visualization}

\Cref{fig:surrogate-renderings} presents the renderings generated by the surrogates for the Front Day and the Top Night scene. %
These budgets correspond to the smallest budget that lead to a validation error less than $2\%$, which was qualitatively chosen as a threshold around which surrogate renders converge on the ground truth (i.e., the rendered scenes visually approach the quality of the original scene).

The top row shows the Front Day scene using surrogates trained on the Front Day dataset.
In this scene, the \moptimal{} and frequency-based surrogates result in similarly accurate renders, with the primary difference being that the frequency-based surrogate rendering has slightly darker green shadows on the front of the house.
This similarity is expected given the similar errors observed in \Cref{fig:renderer-surrogate-errors-table}.
Uniform sampling results in an~inaccurate~render, as expected~given~its~high~error.

The bottom row show the Top Night scene using surrogates trained on the dataset combining all scenes.
In this scene the \moptimal{} surrogate has the most accurate render, as expected given the errors observed in \Cref{fig:renderer-surrogate-errors-table}.
The frequency-trained surrogate colors everything much darker purple.
The uniform-trained surrogate in contrast colors everything much more tan.

In sum, the error improvements in \Cref{fig:renderer-surrogate-errors-table} correspond with improvements in~the~rendered~images.

}

\section{Related Work}

In this section we survey related work for each contribution.

\paragraph{Optimal stratified sampling.}
Optimal stratified sampling is a classic area in statistics~\citep{thompson_sampling_2012}.
Most work in this domain focuses on optimal parameter estimation, and uses stratified sampling to reduce the variance of estimates by ensuring sufficient independent samples are taken from each subpopulation.
Our approach is novel in the assumptions we make for training stratified surrogates of programs, and in the specific sample complexity bounds we base our results on.

\citet{santner_design_2018} survey sampling techniques for computer experiments.
Chapter 5.2.3 discusses stratified random sampling in particular, showing optimality criteria for sampling for unbiased estimators.
These approaches are generic for minimizing the variance of estimators, and do not consider specifically training a neural network.
These approaches also do not consider the different complexity of different strata.

\citet{cortes_region_2019} present an active learning approach for learning in the regime where the input space is partitioned into separate regions (strata, using our terminology) and a separate hypothesis (surrogate) is trained of each, and derive a similar allocation of data points.
This approach has several differences from our approach.
First, it assumes a different form for sample complexity and derives correspondingly different sampling bounds than ours.
The definition of complexity ($\zeta$ in our formalism) that \citeauthor{cortes_region_2019} use is a function of the number of hypotheses in the hypothesis class, the total number of data points used, and the number of data points for a given stratum that have been queried thus far; it is not a function of any complexity metric of the function being learned.
More concretely, \citeauthor{cortes_region_2019}'s approach assumes a small, finite hypothesis class (the set of possible outputs of the training algorithm) of binary classifiers, and has runtime proportional to the size of the hypothesis class, requires samples proportional to the log of the size of the hypothesis class, and bounds the error of the result as a function of the log of the size of the hypothesis class.
In their evaluation, \citeauthor{cortes_region_2019} use a hypothesis class of a set of \num{3000} random hyperplanes.
However, this approach is not tractable when using a neural network as the hypothesis class: a neural network with \num{43000} 32-bit floating point weights (as in the case study in \Cref{sec:renderer}) induces a hypothesis class of size $10^{\num{414217}}$. This results in intractable runtime and large or meaningless bounds.
Beyond these distinctions, \citeauthor{cortes_region_2019}'s approach is also an active learning approach that determines whether or not to query a label of a given data point for an input stream of data points, whereas our approach operates offline.
\citeauthor{cortes_region_2019}'s approach is thus a better fit when learning stratified functions of unknown complexity (e.g., non-analytic functions) using a finitely sized hypothesis class (not a neural network), and is targeted at the online setting when given a sampler of the overall data distribution but not one for each stratum.
Our approach is a better fit when learning \textit{a priori} known stratified analytic functions with neural networks, and is targeted at the offline setting when given~a~sampler~for~each~stratum.

\paragraph{Sample complexity program analysis.}
Program analysis is a broad set of techniques to determine properties of programs~\citep{nielson_principles_1999,cousot_abstract_1977}.
Our analysis in \Cref{sec:analysis} is a novel nonstandard interpretation calculating the tilde, combined with a standard implementation of forward-mode automatic differentiation~\citep{wengert_simple_1964,griewank_derivatives_2008} and a standard symbolic execution which executes all paths in the program~\citep{king_symbolic_1976,cadar_klee_2008}.

\citet{bao_white_2012} present a program analysis that decomposes programs into continuous regions, with the goal of characterizing the sensitivity of each continuous region to input noise.
This analysis computes a different notion of complexity than ours, and does not represent the sample complexity of learning a surrogate of each region.

\citet{hoffmann_polynomial_2010} present a program analysis that calculates the algorithmic complexity of a program.
This complexity again does not lead to bounds on the sample complexity of learning a surrogate of the program.

\section{Assumptions and Limitations}

Our contributions in \Cref{sec:formalism,sec:language} make assumptions about the programs under study, the functions that those programs denote, and the surrogate training algorithms.
Here we document these assumptions and note possible failure modes for our techniques. %

\paragraph{Assumptions imported from prior work.}
Our sample complexity results are subject to all assumptions from the prior work that gives the sample complexity bounds for neural networks that we use~\citep{agarwala_monolithic_2021}.
These sample complexity bounds only apply to analytic functions.
They further assume that inputs come from the unit sphere; this does not match many practical applications, including those in \Cref{sec:evaluation}.
Finally, these sample complexity bounds assume that the neural network under study is a 2-layer, sufficiently wide neural network trained with SGD with an infinitely small step size, using a 1-Lipschitz loss function.

Despite these assumptions, \citet[Appendix~B.2]{agarwala_monolithic_2021} empirically verify that the sample complexity bounds hold.
We also show in \Cref{sec:example,sec:renderer} that the theoretical sample complexity bounds correlate with empirical sample complexity results.

\paragraph{\mOptimal{} sampling.}
The first assumption is that we know the distribution of inputs ahead of time $D(x)$, both in terms of the distribution of strata $\mdiststratum{i}$ and the distribution of inputs within a given stratum $D(x|s_i)$.
The second assumption is that optimizing the upper bound of the per-stratum loss results in a reasonable optimum for the combined surrogate.
The third is the assumption we make that $\forall i,j.\ \mprob_i=\mprob_j$, which we make to ensure a closed-form solution; this is not guaranteed to be optimal.

\paragraph{Convergence in the limit of strata.}
In the limit of infinite strata ($\lim_{c \to \infty}$), the complexity-guided sampling approach induced by \Cref{th:distributionsampling} converges to sampling each stratum with probability proportional to $\mleft(\mdiststratum{i}\mright)^{\frac{2}{3}}$ (see \Cref{note:convergence}).
In the limit, this distribution does not account for complexity.
However in practice the complexity still guides sampling.
\Cref{sec:renderer} evaluates an example with all complexities $\zeta(f_i) \geq 5899$ and $\delta=0.01$.
For the $\log(\delta_i^{-1})$ term to match the contribution of the complexity term there would need to be $\approx 10^{2600}$ strata; this example only has 9.
Thus while in the infinite limit of strata our approach is complexity-agnostic, in practice it is dominated by the complexity.

We note that this property necessarily occurs with any underlying PAC-style bound with a term that sums complexity and $\log\mleft(\frac{1}{\delta}\mright)$ (e.g., those of \citet{vapnik_convergence_1971} and \citet{valiant_theory_1984}): almost surely, paths are sampled with probability that does not depend on their complexity (see \Cref{note:worst-case-convergence}).

\paragraph{Program analysis.}
The main assumption here is that \citet{agarwala_monolithic_2021}'s algebra on tilde functions results in a sufficiently precise upper bound on the tilde.
This is not always the case, as discussed in \Cref{sec:precision}.
We also note that \lang{}'s program analysis could be make more precise in multiple well-known ways.
For instance, we could perform constant propagation, algebraic simplification, and automatic inference of constraints (which would be useful for \texttt{\color{blue}log} expressions).
We have excluded such extensions for the sake of simplicity of presentation.

\paragraph{Analysis compute cost.}
For a given path, computing the tilde and its derivative has essentially the same cost as executing the path twice.
Thus in the most pessimistic case this would allocate 2 more samples per path to the baseline approaches.
However, this pessimistic case assumes that sampling a program input is free, which it may not be: for example the renderer case study in \Cref{sec:renderer} requires executing the video game engine (including running physics simulations) to get a program input for the shader of which we learn a surrogate.

\paragraph{Stratified surrogates.}
We provide sample complexity bounds for constructing stratified surrogates, assuming that for a given program every path is a different function.
This assumes both that it is tractable to compute which stratum a given input resides in before applying the surrogate.
This evaluation-time stratum check must not preclude the use of the surrogate for its downstream task.
We therefore
adopt
a standard modeling assumption in the approximate computing literature: that precisely determining paths is an acceptable cost during approximate program execution~\citep{sampson_enerj_2011,carbin_rely_2013}.\footnote{
\ ``EnerJ ... prohibit[s] approximate values in conditions that affect control flow.''~\citep{sampson_enerj_2011}.
}\textsuperscript{,}\footnote{
\ ``Rely assumes that ... control flow branch targets \NA{are computed}
}

This also assumes that there are a tractable number of paths, which excludes programs with a large number of if statements or loops.
The assumption that there are a tractable number of paths is a common assumption among techniques like concolic testing~\citep{king_symbolic_1976,cute_sen_2005,cadar_klee_2008}.
Similar to prior work, we find that in practice the programs we evaluate only use a fraction of the syntactically possible paths (e.g., the Jmeint benchmark in \Cref{sec:evaluation} uses 18 out of 1728 possible paths).

\paragraph{Empirical evaluation limitations.}
We note two limitations in our empirical evaluations.
The first is that some evaluations are in the ultra-low-data regime, where rounding to an integer number of data samples affects the accuracy.
\NA{The second is that the $\delta$ parameter is set to an arbitrary value.}

\section{Conclusion}
We present an approach to allocating samples among strata to train stratified neural network surrogates of stratified functions.
We also present a programming language, \lang{}, in which all programs are learnable stratified functions and a program analysis to determine the complexity of learning surrogates of those programs.
Our results take a step towards a cohesive, end-to-end methodology for programming using surrogates of programs.

\begin{acks}                            %
  We would like to thank Alana Marzoev, Cambridge Yang, Charles Yuan, Ellie Cheng, Eric Atkinson, Jesse Michel, and the anonymous reviewers for their helpful comments and suggestions.

  This work was supported by the
  \grantsponsor{GS100000001}{National Science Foundation}{http://dx.doi.org/10.13039/100000001}
  (\grantnum{GS100000001}{CCF-1918839} and \grantnum{GS100000001}{CCF-2217064}),
  and the \grantsponsor{GS100000002}{Defense Advanced Research Projects Agency}{http://dx.doi.org/10.13039/100000185} (\grantnum{GS100000002}{\#HR00112190046}).
  Yi Ding is supported by the
  \grantsponsor{GS100000001}{National Science
    Foundation}{http://dx.doi.org/10.13039/100000001}
  under Grant No.~\grantnum{GS100000001}{2030859}
  to the Computing Research Association for the CIFellows Project.
  Any opinions, findings, and conclusions or recommendations expressed in this material are those of the authors and do not necessarily reflect the~views~of~the~sponsors.
\end{acks}

\bibliography{references}

\clearpage
\appendix

\nocite{kwon_transfer_2020,mendis_vemal_2019,munk_deep_2019}

\section{Complexity-Guided Sampling}
\label{app:formalism}

This appendix provides the proof for \Cref{th:distributionsampling} in \Cref{sec:formalism}. %

For convenience we duplicate the definitions of learnability.
A \emph{learnable function} is one such that:
\begin{equation}
 \forall \mdistribution{}, \mnat{}, \merror, \mprob.\, \exists n.\, \probability_{\msurrogate{} \sim \mlearner{}\mleft(\mfunction{}, \mdistribution{}, \mnat{}, \mloss{}\mright)}\mleft(\expectation_{x \sim D} \mleft[\mloss{}\mleft(\msurrogate{}\mleft(\minput{}\mright), \mfunction{}\mleft(\minput{}\mright)\mright)\mright] \leq \merror{} \mright) \geq 1 - \mprob
 \tag{\ref{eq:learnability}}
\end{equation}
Following \citet{arora_fine-grained_2019} and \citet{agarwala_monolithic_2021} we study functions and learning algorithms for which the relationship between $n$, $\epsilon$, and $\delta$ in \Cref{eq:learnability} is:
\begin{equation}
   \exists K.\, \mnat \leq K \mleft[\frac{\mcomplexity{\mfunction} + \log \mleft(\mprob^{-1}\mright)}{\merror^2}\mright]
 \tag{\ref{eq:sample-complexity}}
\end{equation}

We define a \emph{stratified function} $\mfunction$ as follows: \[
  \mfunction{}\mleft(\minput{}\mright) \triangleq \begin{cases}
    \mfunction{}_1\mleft(\minput{}\mright) & \text{if}\; \minput{} \in \mstratum{}_1 \\
    \vdots \\
    \mfunction{}_\mnstrata{}\mleft(\minput{}\mright) & \text{if}\; \minput{} \in \mstratum{}_\mnstrata{} \\
  \end{cases}
\]
where $\mnstrata{}$ is the number of strata, $\{\mstratum{}_i\}_{i=1}^c$ are strata,  $\forall i \neq j.\; \mstratum{}_i \cap \mstratum{}_j = \varnothing$, and $\cup_i \mstratum{}_i = \mInput{}$.

We define a \emph{stratified surrogate} $\msurrogate{}$ as a stratified function with components $\msurrogate{}_i$.

For a data distribution $\mdistribution{}$, let $\mdistribution{}\mleft(\minput{}\mright)$ be the probability that $\minput{}$ is sampled from $\mdistribution{}$, and $\mdistribution{}\mleft(\mstratum{}_i\mright)$ be the total probability mass of all data points within $\mstratum{}_i$ over $\mdistribution{}$ (i.e., $\mdistribution{}\mleft(\mstratum{}_i\mright)=\int_{\minput{}\in \mstratum{}_i}\mdistribution{}\mleft(\minput{}\mright)$).
Let $\mdistribution{}\mleft(\minput{} \middle| \mstratum{}_i\mright)$ be the probability of a data point $\minput{}$ sampled from $\mdistribution{}$ if $\minput\in \mstratum{}_i$.

Our goal is to learn a stratified surrogate $\msurrogate{}$ of a stratified function $\mfunction{}$, where each component function $\mfunction{}_i$ is learnable.
We are given a data distribution $\mdistribution{}$, a maximum sample budget $n$, a learning algorithm $\mlearner{}$, a loss function $\mloss{}$, and a failure probability $\mprob$.
As described in \Cref{sec:formalism}, we formalize this with the following optimization problem:
\begin{equation}
  \begin{split}
    \argmin_{\vec{n}}
    \expectation_{s_i \sim \mdiststratum{i}} \mleft[
    \sqrt{\frac{\mcomplexity{\mfunction_i} + \log \mleft(\mprob_i^{-1}\mright)}{\vec{\mnat}_i}}
    \mright]
    \text{~~such that~~}
    \sum_i \vec{\mnat{}}_i \leq n
  \end{split}
\end{equation}

\distributionsampling*

\vspace*{-1em}

\begin{proof}

  The task is to find $n_i$ for each surrogate $\msurrogate{}_i$ that find a minimal error $\epsilon$ using (using the upper bound in \Cref{eq:sample-complexity}), while also meeting the total sample size constraint:  $\sum_i n_i \leq n$.

  For convenience, define: \[
    p_i \triangleq \mleft(\mdiststratum{i}\mright)\sqrt{\mleft[\zeta\mleft(f_i\mright) + \log \mleft(\delta_i^{-1}\mright)\mright]}
  \]

  We must show the following KKT conditions:
  \allowdisplaybreaks[3]
  \begin{align*}
    \mathcal{L}\mleft(n_i, \mu_n, \mleft\{\mu_i\mright\}\mright) &= \expectation_{s_i \sim \mleft\{\mleft(\mdiststratum{i}\mright)\mright\}}\mleft[
                                                  \sqrt{\frac{1}{n} \mleft[\mcomplexity{\mfunction_i} + \log \mleft(\mprob_i^{-1}\mright)\mright]}
                                                  \mright]
                                                  + \mu_n \mleft(\sum_i n_i - n\mright)
                                                  + \sum_i \mu_i n_i \\
                                                &= \sum_i p_i n_i^{-\frac{1}{2}} + \mu_n \mleft(\sum_i n_i - n\mright) + \sum_i \mu_i n_i \\
    \textbf{Stationarity:} \hspace*{1em}&
\forall i.\,
     0 = -\frac{1}{2} p_i n_i^{-\frac{3}{2}} + \mu_n + \mu_i \\
    \textbf{Primal feasibility:} \hspace*{1em}& \sum n_i \leq n \text{~~and~~} \forall i.\, 0 \leq n_i \\
    \textbf{Dual feasibility:} \hspace*{1em}& 0 \leq \mu_n \text{~~and~~} \forall i.\, 0 \leq \mu_i \\
    \shortstack{\textbf{Complementary}\\ \textbf{slackness}}\textbf{:} \hspace*{1em}& 0 = \mu_n \mleft(n - \sum_i n_i\mright) \text{~~and~~} \forall i.\, \mu_i n_i = 0
  \end{align*}

  We now show that the following is a solution to the optimization problem:
  \begin{align*}
    n_i &= n\frac{p_i^{\frac{2}{3}}}{\sum_j p_j^{\frac{2}{3}}} \\
    \mu_n &= \mleft(\frac{1}{n} \sum_i \mleft(\frac{1}{2} p_i\mright)^{\frac{2}{3}}\mright)^\frac{3}{2} \\
    \mu_i &= 0
  \end{align*}

  \paragraph{Stationarity.}
  \begin{align*}
    -\frac{1}{2} p_i n_i^{-\frac{3}{2}} + \mu_n + \mu_i &= -\frac{1}{2} p_i \mleft(n\frac{p_i^{\frac{2}{3}}}{\sum_j p_j^{\frac{2}{3}}}\mright)^{-\frac{3}{2}} + \mleft(\frac{1}{n} \sum_j \mleft(\frac{1}{2} p_j\mright)^{\frac{2}{3}}\mright)^\frac{3}{2} + 0 \\
                                                        &= -\frac{1}{2} p_i p_i^{-1} n^{-\frac{3}{2}} \mleft(\sum_j p_j^{\frac{2}{3}}\mright)^{\frac{3}{2}} + \frac{1}{2} n^{-\frac{3}{2}} \mleft(\sum_j p_j^{\frac{2}{3}}\mright)^{\frac{3}{2}} \\
                                                        &= 0
  \end{align*}

  \paragraph{Primal feasibility.}
  \begin{align*}
    n_i &= n\frac{p_i^{\frac{2}{3}}}{\sum_j p_j^{\frac{2}{3}}} \geq 0 \hspace*{1em} \text{(if $n \geq 0$)} \\
    \sum_i n_i &= n\frac{1}{\sum_j p_j^{\frac{2}{3}}} \sum_i p_i^{\frac{2}{3}} = n \leq n
  \end{align*}

  \paragraph{Dual feasibility.}
  \begin{align*}
    \mu_n &= \mleft(\frac{1}{n} \sum_i \mleft(\frac{1}{2} p_i\mright)^{\frac{2}{3}}\mright)^\frac{3}{2} \geq 0 \\
    \mu_i &= 0 \geq 0
  \end{align*}

  \paragraph{Complementary slackness.}
  \begin{align*}
    \mu_n \mleft(n - \sum_i n_i\mright) &= \mu_n \mleft(n - \sum_i n\frac{p_i^{\frac{2}{3}}}{\sum_j p_j^{\frac{2}{3}}}\mright) = \mu_n \mleft(n - n\mright) = 0 \\
    \mu_i n_i &= 0
  \end{align*}

  We note that the objective is convex; therefore this is the globally optimal solution. %

  Expanding $p_i$, we have: \[
    n_i = n\frac{\mleft(\mleft(\mdiststratum{i}\mright) \sqrt{\zeta\mleft(f_i\mright) + \log \mleft(\delta_i^{-1}\mright)}\mright)^{\frac{2}{3}}}{\sum_j \mleft(\mleft(\mdiststratum{j}\mright) \sqrt{\zeta\mleft(f_j\mright) + \log \mleft(\delta_j^{-1}\mright)}\mright)^{\frac{2}{3}}}
    \]
\end{proof}

\begin{note}
  \label{note:convergence}
  Assuming $\exists Z.\, \forall i.\, \zeta\mleft(f_i\mright) \leq Z$, then with $n_i$ as above, $\lim_{c \to \infty} n_i = \frac{\mleft(\mdiststratum{i}\mright) ^\frac{2}{3}}{\sum_j \mleft(\mdiststratum{j}\mright) ^\frac{2}{3}}$.
\end{note}

\begin{note}
  \label{note:worst-case-convergence}
  For a given $\delta = 1-\prod_i\mleft(1-\delta_i\mright)$, all choices of $\delta_i$ result in most strata being dominated by the $\log \delta_i^{-1}$ term in the limit of infinite paths:
  \[
    \forall k.\, \lim_{c \to \infty}
    \mleft[
    \min_{\mleft\{\delta_i \mid i \in [c]\mright\}}
    \begin{matrix}
    \mathbb{E} \mleft[\indicator \mleft[\log \delta_i^{-1} > k\mright] \mright]
    \text{~such that~} 0 < \delta_i \leq 1 \text{~and~} 1 - \prod_i (1 - \delta_i) \geq \delta
    \end{matrix}
    \mright]
    = 1
  \]
\end{note}
\begin{proof}
  The $\delta_i$ constraint is equivalent to: \[
    \log (1 - \delta) \geq \sum_i \log \mleft(1 - \delta_i\mright)
  \]
  To minimize $\mathbb{E} \mleft[\indicator \mleft[\log \delta_i^{-1} > k\mright] \mright]$ subject to $1 - \prod (1 - \delta_i) \geq \delta$, we must set as many $\delta_i$ as small as possible without exceeding $\log \delta_i^{-1} > k$. To do this, we set $\delta_i = e^{-k}$ for as many as possible.
  However, because of the constraint, we can do this for no more than $\frac{\log\mleft(1 - \delta\mright)}{\log \mleft(1 - e^{-k}\mright)}$ strata; the rest must exceed $\log \delta_i^{-1} > k$.
  Thus, $    \forall k.\, \lim_{c \to \infty}
    \min_{\mleft\{\delta_i\mright\}}
    \mathbb{E} \mleft[\indicator \mleft[\log \delta_i^{-1} > k\mright] \mright] = 1
    $.
\end{proof}

\section{Complexity Calculus}
\label{app:complexity-algebra}

This section contains proofs for the sample complexity calculus.
The theorems and proofs are extensions of those of \citet{agarwala_monolithic_2021}.

\complexity*
\begin{proof}
  The proof is similar to that of Theorem 8 in \citet{agarwala_monolithic_2021}, with two deviations.

  First, we note that Equation 15 in \citet{agarwala_monolithic_2021} has a typo, which we correct below:
  \begin{equation*}
    \sqrt{M_g} = \sum_k k\mleft|a_k\mright| \mleft\Vert\mathbf{\beta}_k\mright\Vert^k_2 = \mleft\Vert\mathbf{\beta}_k\mright\Vert_2 \sum_{k=1}^\infty k\mleft|a_k\mright| \mleft\Vert\mathbf{\beta}_k\mright\Vert^{k-1}_2
  \end{equation*}
  Thus, the $\mleft|a_0\mright|$ term is not necessary, meaning that the $\tilde{g}\mleft(0\mright)$ term in Equation 14 in \citet{agarwala_monolithic_2021} is not necessary.

  Second, we note that the proof of Corollary 3 in \citet{agarwala_monolithic_2021} involves appending a 1 to the neural network input; thus, the complexity of learning any function $f\mleft(\vec{x}\mright)$ is the same as learning the complexity of a function $f\mleft(\vec{x}\Vert1\mright)$.

  Other than these two modifications, the proof is identical to Theorem 8 in \citet{agarwala_monolithic_2021}.
\end{proof}

We now prove each component of \Cref{lem:tilde-algebra}:
\tildealgebra*

Aspects of these are shown without proof in \citet{agarwala_monolithic_2021}.
We flesh out the proof (using approaches from \citet{agarwala_monolithic_2021}), and further prove bounds on the tilde derivatives.
\begin{lemma}
  \label{lem:tilde-addition}
  If $f\mleft(\vec{x}\mright) = g\mleft(\vec{x}\mright) + h\mleft(\vec{x}\mright)$, then for $x \geq 0$, $\tilde{f}(x) \leq \tilde{g}(x) + \tilde{h}(x)$ and $\tilde{f}'(x) \leq \tilde{g}'(x) + \tilde{h}'(x)$.
\end{lemma}
\begin{proof}
  Given:
  \begin{align*}
    g\mleft(\vec{x}\mright) &= \sum_k \sum_{v \in V_{g,k}} a_{g,v} \prod_{i=1}^k \beta_{g,v,i} \cdot \vec{x} \\
    h\mleft(\vec{x}\mright) &= \sum_k \sum_{v \in V_{h,k}} a_{h,v} \prod_{i=1}^k \beta_{h,v,i} \cdot \vec{x}
  \end{align*}
  Define:
  \begin{align*}
    V_{f,k} &= V_{g,k} \sqcup V_{h,k} \hspace*{3em} \text{(where $\sqcup$ is the disjoint union)} \\
    a_{k,f,v} &= \begin{cases} a_{g,v} & v \in V_{g,k} \\  a_{h,v} & \text{otherwise} \end{cases} \\
    \beta_{f,v,i} &= \begin{cases} \beta_{g,v,i} & v \in V_{g,k} \\  \beta_{h,v,i} & \text{otherwise} \end{cases}
  \end{align*}
  Then:
  \allowdisplaybreaks[3]
  \begin{align*}
    f\mleft(\vec{x}\mright) &= g\mleft(\vec{x}\mright) + h\mleft(\vec{x}\mright) \\
         &= \sum_k \sum_{v \in V_{g,k}} a_{g,v} \prod_{i=1}^k \beta_{g,v,i} \cdot \vec{x} + \sum_k \sum_{v \in V_{h,k}} a_{h,v} \prod_{i=1}^k \beta_{h,v,i} \cdot \vec{x} \\
         &= \sum_k \sum_{v \in V_{f,k}} a_{f,v} \prod_{i=1}^k \beta_{f,v,i} \cdot \vec{x} \\
    \tilde{f}(x) &= \sum_k \mleft(\sum_{v \in V_{f,k}} |a_{f,v}| \prod_{i=1}^k \mleft\Vert\beta_{f,v,i}\mright\Vert_2 \mright) x^k \\
         &= \sum_k \mleft(\sum_{v \in V_{g,k}} |a_{g,v}| \prod_{i=1}^k \mleft\Vert\beta_{g,v,i}\mright\Vert_2 + \sum_{v \in V_{h,k}} |a_{h,v}| \prod_{i=1}^k \mleft\Vert\beta_{h,v,i}\mright\Vert_2 \mright) x^k \\
         &= \sum_k \mleft(\sum_{v \in V_{g,k}} |a_{g,v}| \prod_{i=1}^k \mleft\Vert\beta_{g,v,i}\mright\Vert_2 \mright) x^k + \sum_k \mleft(\sum_{v \in V_{h,k}} |a_{h,v}| \prod_{i=1}^k \mleft\Vert\beta_{h,v,i}\mright\Vert_2 \mright) x^k \\
         &= \tilde{g}(x) + \tilde{h}(x)\\
           \tilde{f}'(x) &= \sum_k \mleft(\sum_{v \in V_{f,k}} |a_{f,v}| \prod_{i=1}^k \mleft\Vert\beta_{f,v,i}\mright\Vert_2 \mright) k x^{k-1} \\
         &= \sum_k \mleft(\sum_{v \in V_{g,k}} |a_{g,v}| \prod_{i=1}^k \mleft\Vert\beta_{g,v,i}\mright\Vert_2 + \sum_{v \in V_{h,k}} |a_{h,v}| \prod_{i=1}^k \mleft\Vert\beta_{h,v,i}\mright\Vert_2 \mright) k x^{k-1} \\
         &= \sum_k \mleft(\sum_{v \in V_{g,k}} |a_{g,v}| \prod_{i=1}^k \mleft\Vert\beta_{g,v,i}\mright\Vert_2 \mright) k x^{k-1} + \sum_k \mleft(\sum_{v \in V_{h,k}} |a_{h,v}| \prod_{i=1}^k \mleft\Vert\beta_{h,v,i}\mright\Vert_2 \mright) k x^{k-1} \\
          &= \tilde{g}'(x) + \tilde{h}'(x)
  \end{align*}
  Note that other choices of $\beta_f$ are possible and may result in a smaller $\tilde{g}$ (hence the imprecision noted in \Cref{sec:precision}).
\end{proof}

\begin{lemma}
  \label{lem:tilde-multiplication}
  If $f\mleft(\vec{x}\mright) = g\mleft(\vec{x}\mright) \cdot h\mleft(\vec{x}\mright)$, then for $x \geq 0$, $\tilde{f}(x) \leq \tilde{g}(x) \cdot \tilde{h}(x)$ and $\tilde{f}'(x) \leq \tilde{g}'(x) \tilde{h}(x) + \tilde{g}(x) \tilde{h}'(x)$.
\end{lemma}
\begin{proof}
  Given:
  \begin{align*}
    g\mleft(\vec{x}\mright) &= \sum_k \sum_{v \in V_{g,k}} a_{g,v} \prod_{i=1}^k \beta_{g,v,i} \cdot \vec{x} \\
    h\mleft(\vec{x}\mright) &= \sum_k \sum_{v \in V_{h,k}} a_{h,v} \prod_{i=1}^k \beta_{h,v,i} \cdot \vec{x}
  \end{align*}
  Define:
  \begin{align*}
    V_{f,l} &= \mleft\{ \mleft(v_g, v_h\mright) \mid j + k = l \land v_g \in V_{g,j} \land v_h \in V_{h,k} \mright\} \\
    a_{f,\mleft(v_g, v_h\mright)} &= a_{g,v_g} \cdot a_{h,v_h} \\
    \beta_{f,\mleft(v_g, v_h\mright),i} &= \begin{cases} \beta_{g,v_g,i} & i \leq j \\ \beta_{h,v_h,i-j} & i > j \end{cases}
  \end{align*}
  Then:
  \allowdisplaybreaks[3]
  \begin{align*}
    f\mleft(\vec{x}\mright) &= g\mleft(\vec{x}\mright) \cdot h\mleft(\vec{x}\mright) \\
         &= \mleft(\sum_k \sum_{v \in V_{g,k}} a_{g,v} \prod_{i=1}^k \beta_{g,v,i} \cdot \vec{x}\mright) \cdot \mleft(\sum_k \sum_{v \in V_{h,k}} a_{h,v} \prod_{i=1}^k \beta_{h,v,i} \cdot \vec{x}\mright) \\
         &= \sum_j \sum_k \sum_{v_g \in V_{g,j}} \sum_{v_h \in V_{h,k}} a_{g,v_g} a_{h,v_h} \prod_{i=1}^j \mleft(\beta_{g,v_g,i} \cdot \vec{x}\mright) \prod_{i=1}^k \mleft(\beta_{h,v_h,i} \cdot \vec{x}\mright) \\
         &= \sum_l \sum_{j + k = l} \sum_{v_g \in V_{g,j}} \sum_{v_h \in V_{h,k}} a_{g,v_g} a_{h,v_h} \prod_{i=1}^j \mleft(\beta_{g,v_g,i} \cdot \vec{x}\mright) \prod_{i=1}^k \mleft(\beta_{h,v_h,i} \cdot \vec{x}\mright) \\
         &= \sum_l \sum_{v \in V_{f,l}} a_{f,v} \prod_{i=1}^l \beta_{f,v,i} \cdot \vec{x} \\
    \tilde{f}(x) &= \sum_l x^l \sum_{v \in V_{f,l}} |a_{f,v}| \prod_{i=1}^l \mleft\Vert\beta_{f,v,i}\mright\Vert_2 \\
         &= \sum_l \sum_{j + k = l} x^{j+k}  \sum_{v_g \in V_{g,j}} \sum_{v_h \in V_{h,k}} |a_{g,v_g} a_{h,v_h}| \prod_{i=1}^j \mleft\Vert\beta_{g,v_g,i}\mright\Vert_2 \prod_{i=1}^k \mleft\Vert\beta_{h,v_h,i} \mright\Vert_2 \\
         &=\sum_j \sum_k x^{j + k} \sum_{v_g \in V_{g,j}} \sum_{v_h \in V_{h,k}} |a_{g, v_g}| |a_{h, v_h}| \prod_{i=1}^j \mleft\Vert\beta_{g, v_g, i}\mright\Vert_2 \prod_{i=1}^k \mleft\Vert\beta_{h, v_h, i}\mright\Vert_2 \\
         &= \mleft(\sum_j x^j \sum_{v_g \in V_{g,j}} |a_{g, v_g}| \prod_{i=1}^j \mleft\Vert\beta_{g, v_g, i}\mright\Vert_2\mright) \cdot \mleft(\sum_k x^k \sum_{v_h \in V_{h,k}} |a_{h, v_h}| \prod_{i=1}^k \mleft\Vert\beta_{h, v_h, i}\mright\Vert_2\mright) \\
         &= \tilde{g}(x) \cdot \tilde{h}(x) \\
    \tilde{f}'(x) &= \sum_l l x^{l-1} \sum_{v \in V_{f,l}} |a_{f,v}| \prod_{i=1}^l \mleft\Vert\beta_{f,v,i}\mright\Vert_2 \\
         &= \sum_l \sum_{j + k = l} (j + k) x^{j + k - 1} \sum_{v_g \in V_{g,j}} \sum_{v_h \in V_{h,k}} |a_{g,v_g} a_{h,v_h}| \prod_{i=1}^j \mleft\Vert\beta_{g,v_g,i}\mright\Vert_2 \prod_{i=1}^k \mleft\Vert\beta_{h,v_h,i} \mright\Vert_2 \\
         &= \mleft(\sum_j j x^{j-1} x^k \sum_{v_g \in V_{g,j}} \sum_{v_h \in V_{h,k}} |a_{g,v_g} a_{h,v_h}| \prod_{i=1}^j \mleft\Vert\beta_{g,v_g,i}\mright\Vert_2 \prod_{i=1}^k \mleft\Vert\beta_{h,v_h,i} \mright\Vert_2\mright) \\ &\hspace*{4em} + \mleft(\sum_k k x^{k-1} x^j \sum_{v_g \in V_{g,j}} \sum_{v_h \in V_{h,k}} |a_{g,v_g} a_{h,v_h}| \prod_{i=1}^j \mleft\Vert\beta_{g,v_g,i}\mright\Vert_2 \prod_{i=1}^k \mleft\Vert\beta_{h,v_h,i} \mright\Vert_2\mright) \\
         &= \mleft(\mleft(\sum_j j x^{j-1} \sum_{v_g \in V_{g,j}} |a_{g,v_g}| \prod_{i=1}^j \mleft\Vert\beta_{g,v_g,i}\mright\Vert_2\mright) \cdot \mleft(\sum_k x^k \sum_{v_h \in V_{h,k}} |a_{h,v_h}| \prod_{i=1}^k \mleft\Vert\beta_{h,v_h,i} \mright\Vert_2\mright)\mright) \\ &\hspace*{4em} + \mleft(\mleft(\sum_k k x^{k-1} \sum_{v_h \in V_{h,k}} |a_{h,v_h}| \prod_{i=1}^k \mleft\Vert\beta_{h,v_h,i} \mright\Vert_2\mright) \cdot \mleft(\sum_j x^j \sum_{v_g \in V_{g,j}} |a_{g,v_g}| \prod_{i=1}^j \mleft\Vert\beta_{g,v_g,i}\mright\Vert_2\mright)\mright) \\
         &= \tilde{g}'(x) \cdot \tilde{h}(x) + \tilde{g}(x) \cdot \tilde{h}'(x)
  \end{align*}
\end{proof}

\begin{lemma}
  \label{lem:tilde-composition}
  If $f\mleft(\vec{x}\mright) = g\mleft(h\mleft(\vec{x}\mright)\mright)$, then $\tilde{h}(x) \leq \tilde{g}(\tilde{f}(x))$ and $\tilde{h}'(x) \leq \tilde{g}'(\tilde{f}(x)) \cdot \tilde{f}'(x)$.
\end{lemma}
\begin{proof}
  The proof is similar to that of Corollary 2 in \citet{agarwala_monolithic_2021}, following fairly directly from Fact 1 in \citet{agarwala_monolithic_2021}:
  \allowdisplaybreaks[3]
  \begin{align*}
    f(x) &= \sum_k a_{g,k} h\mleft(\vec{x}\mright)^k \\
    \tilde{f}(x) &\leq \sum_k \widetilde{a_{g,k} h\mleft(\vec{x}\mright)^k} & (f = g + h \Rightarrow \tilde{f} \leq \tilde{g} + \tilde{h})\\
    \    &\leq \sum_k |a_{g,k}| \widetilde{h\mleft(\vec{x}\mright)^k} & (f(x) = g(c x) \Rightarrow \tilde{f}(x) = |c|\tilde{g}(x)\\
    \    &\leq \sum_k |a_{g,k}| \mleft(\tilde{h}\mleft(x\mright)\mright)^k & (f = g \cdot h \Rightarrow \tilde{f} \leq \tilde{g} \cdot \tilde{h})\\
         &\leq \tilde{g}\mleft(\tilde{h}\mleft(x\mright)\mright) \\
    \tilde{f}'(x) &= \frac{d}{dx} \reallywidetilde{\sum_k a_{g,k} h\mleft(\vec{x}\mright)^k} \\
         &\leq \sum_k \frac{d}{dx} \widetilde{a_{g,k} h\mleft(\vec{x}\mright)^k} & (f = g + h \Rightarrow \tilde{f}' \leq \tilde{g}' + \tilde{h}')\\
         &= \sum_k |a_{g,k}| \frac{d}{dx}  \widetilde{h\mleft(\vec{x}\mright)^k} \\
         &\leq \sum_k |a_{g,k}| k \tilde{h}(x)^{k-1} \tilde{h}'(x) & (f = g \cdot h \Rightarrow \tilde{f}' \leq \tilde{g}' \cdot \tilde{h} + \tilde{g} \cdot \tilde{h}')\\
         &= \tilde{g}'\mleft(\tilde{h}\mleft(x\mright)\mright) \cdot \tilde{h}'(x)
  \end{align*}
\end{proof}

\begin{lemma}
  \label{lem:tilde-monotonic}
  The tilde turns the function and its derivative into a monotonic function for $x\geq0$:

  \begin{equation*}
    \forall x \geq x' \geq 0. \tilde{f}\mleft(x\mright) \geq \tilde{f}\mleft(x\mright) \land \tilde{f}'\mleft(x\mright) \geq \tilde{f}'\mleft(x\mright)
  \end{equation*}
\end{lemma}
\begin{proof}
  \begin{equation*}
    \tilde{f}\mleft(x\mright) = \sum_{k=0}^\infty \mleft|a_k\mright| x^k
  \end{equation*}
  The derivative of a power series with all nonnegative coefficients is also a power series with all nonnegative coefficients:
  \begin{equation*}
    \tilde{f}'\mleft(x\mright) = \sum_{k=0}^\infty k \mleft|a_k\mright| x^{k-1}
  \end{equation*}
  For $x > 0$, and any power series with all nonnegative coefficients $g$, $g\mleft(x\mright) \geq 0$.
  Thus, the derivative of $\tilde{f}$ and $\tilde{f}'$ are both nonnegative everywhere, thus they are both monotonic.
\end{proof}

\section{\lang{} Statement Big-Step Semantics}
\label{app:language-semantics}

This appendix presents the remaining semantics for \lang{} not presented in \Cref{sec:language}.

\begin{figure}
  \begin{mathpar}
    \centering
    \inferrule{
    }{\mleft\langle \sigma, \texttt{{\color{blue}skip}} \mright\rangle \Downarrow \sigma}

    \inferrule{
      \mleft\langle \sigma, s_1 \mright\rangle \Downarrow \sigma' \\
      \mleft\langle \sigma', s_2 \mright\rangle \Downarrow \sigma''
    }{\mleft\langle \sigma, s_1 \,\texttt{;}\, s_2 \mright\rangle \Downarrow \sigma''}

    \inferrule{
      \mleft\langle\sigma, e\mright\rangle \Downarrow v
    }{\mleft\langle \sigma, x \,\texttt{=}\, e \mright\rangle \Downarrow \sigma\mleft[x \mapsto v\mright]}

    \inferrule{
      \langle \sigma , e\rangle \Downarrow v \\
      v > 0 \\
      \mleft\langle \sigma, s_1 \mright\rangle \Downarrow \sigma_l
    }{\mleft\langle \sigma,
      \texttt{{\color{blue}if} (}e \texttt{ > 0) \{} s_1 \texttt{\} {\color{blue}else} \{} s_2 \texttt{\}}
      \mright\rangle \Downarrow
      \sigma_l
    }

    \inferrule{
      \langle \sigma , e\rangle \Downarrow v \\
      v \leq 0 \\
      \mleft\langle \sigma, s_2 \mright\rangle \Downarrow \sigma_r
    }{\mleft\langle \sigma,
      \texttt{{\color{blue}if} (}e \texttt{ > 0) \{} s_1 \texttt{\} {\color{blue}else} \{} s_2 \texttt{\}}
      \mright\rangle \Downarrow
      \sigma_r
    }
  \end{mathpar}
  \caption{Big-step evaluation relation for statements.}
  \label{fig:evaluation-statement-relational-app}
  \begin{mathpar}
    \centering
    \inferrule{
      \mleft\langle \sigma, s \mright\rangle \Downarrow \sigma'\\
    }{\mleft\langle \sigma, \texttt{{\color{blue}fun} (}x_0, x_1 \dots, x_n\texttt{) \{}s \,\texttt{;}\, \texttt{{\color{blue}return} } x \texttt{\}}\mright\rangle \Downarrow \sigma'\mleft(x\mright)}
  \end{mathpar}
  \caption{Big-step evaluation relation for \lang{}.}
  \label{fig:evaluation-program-relational-app}
\end{figure}

\Cref{fig:evaluation-statement-relational-app} presents the big-step evaluation relation for statements in \lang{}.
The statement relation $\langle \sigma, s \rangle \Downarrow \sigma'$ says that under variable store $\sigma$, the statement $s$ evaluates to a new variable store $\sigma'$.

\Cref{fig:evaluation-program-relational-app} presents the big-step evaluation relation for \lang{} programs.
The program relation $
  \mleft\langle \sigma, \texttt{{\color{blue}fun} (}x_0, x_1 \dots, x_n\texttt{) \{}s \,\texttt{;}\, \texttt{{\color{blue}return} } x \texttt{\}}\mright\rangle \Downarrow v
  $
says that under variable store $\sigma$ (representing the inputs to the program), the program evaluates to value $v$.

\section{\lang{} Analysis Proof of Soundness}
\label{app:language}

This appendix proves that the \lang{} complexity analysis is sound: that it computes an upper bound on the true complexity of learning the program.

We first note that the standard execution semantics big-step relation $\Downarrow$ both for expressions and for traces is a function.
We use $\mleft\llbracket \cdot \mright\rrbracket$ as notation to refer to that function for expressions and $\mleft\llbracket \cdot \mright\rrbracket_x$ to refer to that function for traces followed by taking the value of the variable $x$:
\begin{gather*}
  \mleft\llbracket e \mright\rrbracket\mleft(\sigma\mright) = v \Leftrightarrow \mleft\langle\sigma, e\mright\rangle \Downarrow v \\
  \mleft\llbracket t \mright\rrbracket_x\mleft(\sigma\mright) = v \Leftrightarrow \mleft\langle\sigma, t\mright\rangle \Downarrow \sigma' \land \sigma'\mleft(x\mright) = v
\end{gather*}
We use the notation $\mleft\{f_x\mright\}$ as shorthand for $\mleft\{f_x \,\middle|\, x \in \sigma\mright\}$, a set of functions indexed by $x \in \sigma$.
We use the notation $\tilde{\sigma} \vdash \mleft\{f_x\mright\}$ to denote the predicate that $f_x$ have tildes and tilde derivatives that are bounded by $\tilde{\sigma}$:
\begin{gather*}
  \tilde{\sigma} \vdash \mleft\{f_x\mright\} \Leftrightarrow \forall x \in \tilde{\sigma}.\, \mleft(\tilde{\sigma}\mleft(x\mright) = \mleft(\tilde{v}, \tilde{v}'\mright) \Rightarrow \mleft(
  0 \leq \tilde{f_x}\mleft(1\mright) \leq \tilde{v} \land 0 \leq \tilde{f_x}'\mleft(1\mright) \leq \tilde{v}'\mright)\mright)
\end{gather*}
We use the notation $\circ$ to denote function composition:
\begin{gather*}
  \mleft(\mleft\llbracket \cdot \mright\rrbracket \circ \mleft\{f_x\mright\}\mright)\mleft(\sigma\mright) \triangleq \mleft\llbracket \cdot \mright\rrbracket \mleft(\mleft\{x \mapsto f_x\mleft(\sigma\mright)\mright\}\mright)
\end{gather*}

\tildeexpression*
\begin{proof}
  We prove this by induction using facts from \Cref{app:complexity-algebra}.

  \textbf{Case:} $\inferrule{ }{\langle \tilde{\sigma}, v \rangle \;\tilde\Downarrow\; \mleft(|v|, 0\mright)}$
  \\
  $\mleft\llbracket v \mright\rrbracket\mleft(\sigma\mright) = v$, so $\mleft(\mleft\llbracket v \mright\rrbracket \circ \mleft\{f_x\mright\}\mright)\mleft(\sigma\mright) = v$. Thus, $\widetilde{\mleft\llbracket v \mright\rrbracket \circ \mleft\{f_x\mright\}}\mleft(1\mright) = |v|$ and $\widetilde{\mleft\llbracket v \mright\rrbracket \circ \mleft\{f_x\mright\}}'\mleft(1\mright) = 0$, proving this case.

  \textbf{Case:} $\inferrule{ }{\langle \tilde{\sigma}, x \rangle \;\tilde\Downarrow\; \tilde{\sigma}(x)}$
  \\
  $\mleft\llbracket x \mright\rrbracket\mleft(\sigma\mright) = \sigma\mleft(x\mright)$, so $\mleft(\mleft\llbracket x \mright\rrbracket \circ \mleft\{f_x\mright\}\mright)\mleft(\sigma\mright) = f_x\mleft(\sigma\mright)$.
  We know that for $\tilde{\sigma}\mleft(x\mright) = \mleft(\tilde{v}, \tilde{v}'\mright)$, $\tilde{f_x}\mleft(1\mright) \leq \tilde{v}$ and $\tilde{f_x}'\mleft(1\mright) \leq \tilde{v}'$.
  Thus, $\widetilde{\mleft\llbracket x \mright\rrbracket \circ \mleft\{f_x\mright\}}\mleft(1\mright) = \tilde{f_x}\mleft(1\mright) \leq \tilde{v}$ and $\widetilde{\mleft\llbracket x \mright\rrbracket \circ \mleft\{f_x\mright\}}'\mleft(1\mright) = \tilde{f_x}'\mleft(1\mright) \leq \tilde{v}'$, proving this case.

  \textbf{Case:} $    \inferrule{
      \mleft\langle \tilde{\sigma},  e\mright\rangle \cstepsto \mleft(\tilde{v},\tilde{v}'\mright) \\
    }{\mleft\langle \tilde{\sigma},  {\texttt{-}e}\mright\rangle \cstepsto \mleft(\tilde{v},\tilde{v}'\mright)}
    $
    \\
    $\mleft\llbracket -e \mright\rrbracket\mleft(\sigma\mright) = -\mleft\llbracket e \mright\rrbracket\mleft(\sigma\mright)$, so $\mleft(\mleft\llbracket -e \mright\rrbracket \circ \mleft\{f_x\mright\}\mright)\mleft(\sigma\mright) = -\mleft(\mleft\llbracket e \mright\rrbracket \circ \mleft\{f_x\mright\}\mright)\mleft(\sigma\mright)$.
    By the inductive hypothesis, we know that $\widetilde{\mleft\llbracket e \mright\rrbracket \circ \mleft\{f_x\mright\}}\mleft(1\mright) \leq \tilde{v}$ and $\widetilde{\mleft\llbracket e \mright\rrbracket \circ \mleft\{f_x\mright\}}'\mleft(1\mright) \leq \tilde{v}'$.
    Thus, $\widetilde{\mleft\llbracket -e \mright\rrbracket \circ \mleft\{f_x\mright\}}\mleft(1\mright) = \widetilde{\mleft\llbracket e \mright\rrbracket \circ \mleft\{f_x\mright\}}\mleft(1\mright) \leq \tilde{v}$ and $\widetilde{\mleft\llbracket -e \mright\rrbracket \circ \mleft\{f_x\mright\}}'\mleft(1\mright) = \widetilde{\mleft\llbracket e \mright\rrbracket \circ \mleft\{f_x\mright\}}'\mleft(1\mright) \leq \tilde{v}'$, proving this case.

  \textbf{Case:} $        \inferrule{
      \mleft\langle \tilde{\sigma},  e\mright\rangle \cstepsto \mleft(\tilde{v},\tilde{v}'\mright) \\
    }{\mleft\langle \tilde{\sigma},  \texttt{{\color{blue}sin}(}e\texttt{)}\mright\rangle \cstepsto \mleft(\sinh{\mleft(\tilde{v}\mright)}, \tilde{v}'\cosh\mleft(\tilde{v}\mright)\mright)}
    $
    \\
    $\mleft\llbracket \texttt{{\color{blue}sin}(}e\texttt{)} \mright\rrbracket\mleft(\sigma\mright) = \sin\mleft(\mleft\llbracket e \mright\rrbracket\mleft(\sigma\mright)\mright)$, so $\mleft(\mleft\llbracket \texttt{{\color{blue}sin}(}e\texttt{)} \mright\rrbracket \circ \mleft\{f_x\mright\}\mright)\mleft(\sigma\mright) = \sin\mleft(\mleft\llbracket e \mright\rrbracket \circ \mleft\{f_x\mright\} \mleft(\sigma\mright)\mright)$. By the inductive hypothesis, $\widetilde{\mleft\llbracket e \mright\rrbracket \circ \mleft\{f_x\mright\}}\mleft(1\mright) \leq \tilde{v}$ and $\widetilde{\mleft\llbracket e \mright\rrbracket \circ \mleft\{f_x\mright\}}'\mleft(1\mright) \leq \tilde{v}'$.
    Note that $\widetilde{\sin}\mleft(x\mright) = \sinh\mleft(x\mright)$.
    By \Cref{lem:tilde-composition}, $\widetilde{\mleft\llbracket \texttt{{\color{blue}sin}(}e\texttt{)} \mright\rrbracket \circ \mleft\{f_x\mright\}}\mleft(1\mright) \leq \sinh\mleft(\widetilde{\mleft\llbracket e \mright\rrbracket \circ \mleft\{f_x\mright\}}\mleft(1\mright)\mright) \leq \sinh\mleft(\tilde{v}\mright)$ (\Cref{lem:tilde-monotonic}) and $\widetilde{\mleft\llbracket \texttt{{\color{blue}sin}(}e\texttt{)} \mright\rrbracket \circ \mleft\{f_x\mright\}}'\mleft(1\mright) \leq \widetilde{\mleft\llbracket e \mright\rrbracket \circ \mleft\{f_x\mright\}}'\mleft(1\mright) \cosh\mleft(\widetilde{\mleft\llbracket e \mright\rrbracket \circ \mleft\{f_x\mright\}}\mleft(1\mright)\mright) \leq \tilde{v}' \cosh\mleft(\tilde{v}\mright)$ (\Cref{lem:tilde-monotonic}).

    \textbf{Case:} $        \inferrule{
      \mleft\langle \tilde{\sigma},  e\mright\rangle \cstepsto \mleft(\tilde{v},\tilde{v}'\mright) \\
    }{\mleft\langle \tilde{\sigma},  \texttt{{\color{blue}exp}(}e\texttt{)} \mright\rangle \cstepsto \mleft(\exp{\mleft(\tilde{v}\mright)}, \tilde{v}'\exp\mleft(\tilde{v}\mright)\mright)}
    $
    \\
    $\mleft\llbracket \texttt{{\color{blue}exp}(}e\texttt{)} \mright\rrbracket\mleft(\sigma\mright) = \exp\mleft(\mleft\llbracket e \mright\rrbracket\mleft(\sigma\mright)\mright)$, so $\mleft(\mleft\llbracket \texttt{{\color{blue}exp}(}e\texttt{)} \mright\rrbracket \circ \mleft\{f_x\mright\}\mright)\mleft(\sigma\mright) = \exp\mleft(\mleft\llbracket e \mright\rrbracket \circ \mleft\{f_x\mright\} \mleft(\sigma\mright)\mright)$. By the inductive hypothesis, $\widetilde{\mleft\llbracket e \mright\rrbracket \circ \mleft\{f_x\mright\}}\mleft(1\mright) \leq \tilde{v}$ and $\widetilde{\mleft\llbracket e \mright\rrbracket \circ \mleft\{f_x\mright\}}'\mleft(1\mright) \leq \tilde{v}'$.
    Note that $\widetilde{\exp}\mleft(x\mright) = \exp\mleft(x\mright)$.
    By \Cref{lem:tilde-composition}, $\widetilde{\mleft\llbracket \texttt{{\color{blue}exp}(}e\texttt{)} \mright\rrbracket \circ \mleft\{f_x\mright\}}\mleft(1\mright) \leq \exp\mleft(\widetilde{\mleft\llbracket e \mright\rrbracket \circ \mleft\{f_x\mright\}}\mleft(1\mright)\mright) \leq \exp\mleft(\tilde{v}\mright)$ and $\widetilde{\mleft\llbracket \texttt{{\color{blue}exp}(}e\texttt{)} \mright\rrbracket \circ \mleft\{f_x\mright\}}'\mleft(1\mright) \leq \widetilde{\mleft\llbracket e \mright\rrbracket \circ \mleft\{f_x\mright\}}'\mleft(1\mright) \cdot \widetilde{\exp}\mleft(\widetilde{\mleft\llbracket e \mright\rrbracket \circ \mleft\{f_x\mright\}}\mleft(1\mright)\mright) \leq \tilde{v}' \exp\mleft(\tilde{v}\mright)$ (\Cref{lem:tilde-monotonic}).

    \textbf{Case:} $\inferrule{
      \mleft\langle \tilde{\sigma},  e\mright\rangle \cstepsto \mleft(\tilde{v},\tilde{v}'\mright) \\
      b > \tilde{v}\sqrt{b^2+1}
    }{\mleft\langle \tilde{\sigma},  \texttt{{\color{blue}log}\{}b\texttt{\}(}e\texttt{)}\mright\rangle \cstepsto
      \mleft(
      \mleft|\log\mleft(b\mright)\mright| + \log\mleft(b\mright) - \log\mleft(b - \tilde{v}\sqrt{b^2+1}\mright),
      \frac{\tilde{v}'}{b - \tilde{v}\sqrt{b^2+1}}
      \mright)
    }$
    \\
    $\mleft\llbracket \texttt{{\color{blue}log}\{}b\texttt{\}(}e\texttt{)} \mright\rrbracket\mleft(\sigma\mright) = \log\mleft(\mleft\llbracket e \mright\rrbracket\mright)$, so $\mleft(\mleft\llbracket \texttt{{\color{blue}log}\{}b\texttt{\}(}e\texttt{)} \mright\rrbracket \circ \mleft\{f_x\mright\} \mright)\mleft(\sigma\mright) = \log\mleft(\mleft(\mleft\llbracket e \circ \mleft\{f_x\mright\}\mright\rrbracket\mright)\mleft(\sigma\mright)\mright)$. By the inductive hypothesis we have $\mleft\llbracket e \circ \mleft\{f_x\mright\} \mright\rrbracket\mleft(\sigma\mright) \leq \tilde{v}$ and $\mleft\llbracket e \circ \mleft\{f_x\mright\} \mright\rrbracket'\mleft(\sigma\mright) \leq \tilde{v}'$.
    Expanding around $x=b$, we have $\log\mleft(x\mright) = \log\mleft(b\mright) + \sum_{k=1}^\infty \mleft(-1\mright)^{k+1} \frac{\mleft(x - b\mright)^k}{kb^k}$.
    Note that this is a multivariate analytic power series in the sense of \Cref{eq:tilde-definition-multivar} (with the $1$ appended) with $V_k=\mleft\{v_k\mright\}$, $a_{v_0} = \log\mleft(b\mright)$, $a_{v_{k>0}} = \frac{\mleft(-1\mright)^{k+1}}{kb^k}$, and $\beta_{v_k,i} = \begin{bmatrix}1 \\ -b\end{bmatrix}$.
    Thus, $\widetilde{\log}\mleft(x\mright) = \mleft|\log\mleft(b\mright)\mright| + \sum_{k=1}^\infty \frac{\mleft(1 + b^2\mright)^\frac{k}{2}}{kb^k} x^k = \mleft|\log\mleft(b\mright)\mright| + \log\mleft(b\mright) - \log\mleft(b - x\sqrt{b^2 + 1}\mright)$, which converges for $b > x\sqrt{b^2+1}$.
    By \Cref{lem:tilde-composition}, $\widetilde{\mleft\llbracket \texttt{{\color{blue}log}\{}b\texttt{\}(}e\texttt{)} \mright\rrbracket\mleft(\sigma\mright)} \leq \mleft|\log\mleft(b\mright)\mright| + \log\mleft(b\mright) - \log\mleft(b - \tilde{v}\sqrt{b^2+1}\mright)$ and $\widetilde{\mleft\llbracket \texttt{{\color{blue}log}\{}b\texttt{\}(}e\texttt{)} \mright\rrbracket'\mleft(\sigma\mright)} \leq \frac{\tilde{v}'}{b - \tilde{v}\sqrt{b^2+1}}$ (\Cref{lem:tilde-monotonic}).

    \textbf{Case:} $
    \inferrule{
      \mleft\langle \tilde{\sigma},  e_1\mright\rangle \cstepsto \mleft(\tilde{v},\tilde{v}'\mright) \\
      \mleft\langle \tilde{\sigma},  e_2\mright\rangle \cstepsto \mleft(\tilde{v}_2, \tilde{v}'_2\mright) \\
    }{\mleft\langle \tilde{\sigma},  e_1 \texttt{+} e_2\mright\rangle \cstepsto \mleft(\tilde{v}+ \tilde{v}_2, \tilde{v}' + \tilde{v}'_2\mright)}
    $
    \\
    $\mleft\llbracket e_1 \,\texttt{+}\, e_2 \mright\rrbracket\mleft(\sigma\mright) = \mleft\llbracket e_1\mright\rrbracket\mleft(\sigma\mright) + \mleft\llbracket e_2\mright\rrbracket\mleft(\sigma\mright)$, so $\mleft(\mleft\llbracket e_1 + e_2 \mright\rrbracket \circ \mleft\{f_x\mright\} \mright)\mleft(\sigma\mright) = \mleft(\mleft\llbracket e_1 \mright\rrbracket \circ \mleft\{f_x\mright\} \mright)\mleft(\sigma\mright) + \mleft(\mleft\llbracket e_2 \mright\rrbracket \circ \mleft\{f_x\mright\} \mright)\mleft(\sigma\mright)$.
    By the inductive hypothesis, $\widetilde{\mleft\llbracket e_1 \mright\rrbracket \circ \mleft\{f_x\mright\}}\mleft(\sigma\mright) \leq \tilde{v}$, $\widetilde{\mleft\llbracket e_1 \mright\rrbracket \circ \mleft\{f_x\mright\}}'\mleft(\sigma\mright) \leq \tilde{v}'$, $\widetilde{\mleft\llbracket e_2 \mright\rrbracket \circ \mleft\{f_x\mright\}}\mleft(\sigma\mright) \leq \tilde{v}_2$, and $\widetilde{\mleft\llbracket e_2 \mright\rrbracket \circ \mleft\{f_x\mright\}}'\mleft(\sigma\mright) \leq \tilde{v}_2'$
    Thus by \Cref{lem:tilde-addition}, $\widetilde{\mleft\llbracket e_1 + e_2 \mright\rrbracket \circ \mleft\{f_x\mright\}}\mleft(\sigma\mright) \leq \tilde{v} + \tilde{v}_2$ and $\widetilde{\mleft\llbracket e_1 + e_2 \mright\rrbracket \circ \mleft\{f_x\mright\}}'\mleft(\sigma\mright) \leq \tilde{v}' + \tilde{v}_2'$.

    \textbf{Case:} $
    \inferrule{
      \mleft\langle \tilde{\sigma},  e_1\mright\rangle \cstepsto \mleft(\tilde{v},\tilde{v}'\mright) \\
      \mleft\langle \tilde{\sigma},  e_2\mright\rangle \cstepsto \mleft(\tilde{v}_2, \tilde{v}'_2\mright) \\
    }{\mleft\langle \tilde{\sigma},  e_1\texttt{*}e_2\mright\rangle \cstepsto \mleft(\tilde{v} \cdot \tilde{v}_2, \tilde{v}' \cdot \tilde{v}_2 + \tilde{v} \cdot \tilde{v}'_2\mright)}
    $
    \\
    $\mleft\llbracket e_1 \,\texttt{*}\, e_2 \mright\rrbracket\mleft(\sigma\mright) = \mleft\llbracket e_1\mright\rrbracket\mleft(\sigma\mright) \cdot \mleft\llbracket e_2\mright\rrbracket\mleft(\sigma\mright)$, so $\mleft(\mleft\llbracket e_1 \texttt{*} e_2 \mright\rrbracket \circ \mleft\{f_x\mright\} \mright)\mleft(\sigma\mright) = \mleft(\mleft\llbracket e_1 \mright\rrbracket \circ \mleft\{f_x\mright\} \mright)\mleft(\sigma\mright) \cdot \mleft(\mleft\llbracket e_2 \mright\rrbracket \circ \mleft\{f_x\mright\} \mright)\mleft(\sigma\mright)$.
    By the inductive hypothesis we have $\widetilde{\mleft\llbracket e_1 \mright\rrbracket \circ \mleft\{f_x\mright\}}\mleft(\sigma\mright) \leq \tilde{v}$, $\widetilde{\mleft\llbracket e_1 \mright\rrbracket \circ \mleft\{f_x\mright\}}'\mleft(\sigma\mright) \leq \tilde{v}'$, $\widetilde{\mleft\llbracket e_2 \mright\rrbracket \circ \mleft\{f_x\mright\}}\mleft(\sigma\mright) \leq \tilde{v}_2$, and $\widetilde{\mleft\llbracket e_2 \mright\rrbracket \circ \mleft\{f_x\mright\}}'\mleft(\sigma\mright) \leq \tilde{v}_2'$
    Thus by \Cref{lem:tilde-multiplication}, $\widetilde{\mleft\llbracket e_1 \,\texttt{*}\, e_2 \mright\rrbracket \circ \mleft\{f_x\mright\}}\mleft(\sigma\mright) \leq \tilde{v} \cdot \tilde{v}_2$ and $\widetilde{\mleft\llbracket e_1 \,\texttt{*}\, e_2 \mright\rrbracket \circ \mleft\{f_x\mright\}}'\mleft(\sigma\mright) \leq \tilde{v}' \cdot \tilde{v}_2 + \tilde{v} \cdot \tilde{v}'_2$.

  \end{proof}

\tildetrace*
\begin{proof}
  The proof proceeds by induction on $t$.

  \textbf{Case:} $\inferrule{ }{\mleft\langle \tilde{\sigma}, \texttt{{\color{blue}skip}} \mright\rangle \;\tilde{\Downarrow}\; \tilde{\sigma}}$ \\
  We know that $\mleft\llbracket\texttt{\color{blue}skip}\mright\rrbracket_x\mleft(\sigma\mright) = \sigma\mleft(x\mright)$.
  Thus $\mleft\llbracket\texttt{\color{blue}skip}\mright\rrbracket_y \circ \mleft\{f_x\mright\} = f_y$.
  We know $\tilde\sigma \vdash \mleft\{f_x\mright\}$, that $\tilde{\sigma}' = \tilde{\sigma}$, and that $\mleft\{\mleft\llbracket\texttt{\color{blue}skip}\mright\rrbracket_y \circ \mleft\{f_x\mright\}\mright\} = \mleft\{f_x\mright\}$, so $\tilde{\sigma}' \vdash \mleft\{ \mleft\llbracket t \mright\rrbracket_y \circ \mleft\{f_x\mright\} \mright\}$.

  \textbf{Case:}
    $\inferrule{
      \mleft\langle \tilde{\sigma}, t_1 \mright\rangle \;\tilde{\Downarrow}\; \tilde{\sigma}' \\
      \mleft\langle \tilde{\sigma}' , t_2 \mright\rangle \;\tilde{\Downarrow}\; \tilde{\sigma}''
    }{\mleft\langle \tilde{\sigma}, t_1 \,\texttt{;}\, t_2 \mright\rangle \;\tilde{\Downarrow}\; \tilde{\sigma}''}$ \\
    We know that $\mleft\llbracket t_1 \,\texttt{;}\, t_2\mright\rrbracket_x\mleft(\sigma\mright) = \mleft\llbracket t_2\mright\rrbracket_x\mleft(\mleft\{y \mapsto \mleft\llbracket t_1\mright\rrbracket_y\mleft(\sigma\mright)\mright\}\mright)$.
    By the IH, $\tilde{\sigma}' \vdash \mleft\{\mleft\llbracket t_1\mright\rrbracket_y \circ \mleft\{f_x\mright\}\mright\}$, and $\tilde{\sigma}'' \vdash \mleft\{\mleft\llbracket t_2\mright\rrbracket_z \circ \mleft\{\mleft\llbracket t_1\mright\rrbracket_y \circ \mleft\{f_x\mright\}\mright\} \mright\}$.
    We further know that composition is associative, meaning that we can rewrite this as
    $\tilde{\sigma}'' \vdash \mleft\{\mleft(\mleft\llbracket t_2\mright\rrbracket_z \circ \mleft\{ \mleft\llbracket t_1\mright\rrbracket_y \mright\} \mright\}\mright) \circ \mleft\{f_x\mright\}$.
    Because $\mleft\llbracket t_2\mright\rrbracket_z \circ \mleft\{ \mleft\llbracket t_1\mright\rrbracket_y \mright\} = \mleft\llbracket t_1 \,\texttt{;}\, t_2\mright\rrbracket_z$, this proves this case.

  \textbf{Case:} $
    \inferrule{
      \mleft\langle \tilde{\sigma}, e \mright\rangle \cstepsto \mleft(\tilde{v}, \tilde{v}'\mright)
    }{\mleft\langle \tilde{\sigma}, x \,\texttt{=}\, e \mright\rangle \;\tilde{\Downarrow}\; \tilde{\sigma}\mleft[x \mapsto \mleft(\tilde{v}, \tilde{v}'\mright)\mright]} $ \\
    We know that $\mleft\llbracket x \,\texttt{=}\, e \mright\rrbracket_y\mleft(\sigma\mright)$ is equal to $\mleft\llbracket e \mright\rrbracket\mleft(\sigma\mright)$ for $y=x$ and $\sigma\mleft(y\mright)$ for $y \neq x$.
    By \Cref{lem:tilde-expression} we know $\widetilde{\mleft\llbracket e \mright\rrbracket \circ \mleft\{f_x\mright\}}\mleft(1\mright) \leq \tilde{v}$ and $\widetilde{\mleft\llbracket e \mright\rrbracket \circ \mleft\{f_x\mright\}}'\mleft(1\mright) \leq \tilde{v}'$.
    Thus $\tilde{\sigma}\mleft[x \mapsto \mleft(\tilde{v}, \tilde{v}'\mright)\mright] \vdash \mleft\{\mleft\llbracket e \mright\rrbracket \circ \mleft\{f_x\mright\}\mright\}$.

\end{proof}

\analysisupperbound*
\begin{proof}
  We know that $\mleft\langle \mleft\{x_i \mapsto \mleft(1, 1\mright)\mright\}, t \mright\rangle \;\tilde{\Downarrow}\; \tilde{\sigma}$.
  We note that $\mleft\{x_i \mapsto \mleft(1, 1\mright)\mright\} \vdash \mleft\{f_{x_i} \mleft(x\mright) = x\mright\}$.
  Thus by \Cref{lem:tilde-trace}, $\tilde{\sigma} \vdash \mleft\{\mleft\llbracket t\mright\rrbracket_x\mright\}$, so for $\tilde{\sigma}\mleft(x\mright) = \mleft(\tilde{v}, \tilde{v}'\mright)$ and $z=\tilde{v}'^2$ we know $\widetilde{\mleft\llbracket t\mright\rrbracket_x}' \leq \tilde{v}'$, so $\mleft(\widetilde{\mleft\llbracket t\mright\rrbracket_x}'\mright)^2 \leq z$.
\end{proof}

\section{Log Rule}
\label{app:log-rule}

The $\texttt{{\color{blue}log}\{}b\texttt{\}(}e\texttt{)}$ rule requires the parameter $b$ to expand around since $\log$ is not analytic around 0.
The value set for this parameter must satisfy the $|b - v| < b$ condition for all inputs in the standard interpretation (to ensure that all values are in the radius of convergence) and the $b > \tilde{v}\sqrt{b^2+1}$ condition in the tilde interpretation, but is otherwise free to be set to a value that minimizes the upper bound on the complexity.

As an example of a program that uses value of $b$ value other than $1$, consider the following program:

\begin{lstlisting}[numbers=none,language=learnguage,frame=single,
  breaklines=true,postbreak=\mbox{\textcolor{red}{$\hookrightarrow$}\space},
  escapeinside=``
 ,linewidth=0.99\linewidth,
 ]
fun(x) {
   x = log{3.88}(0.75 * x);
   x = x * x;
   return x;
}
\end{lstlisting}
This program cannot use $b=1$ (since it fails the condition in the tilde interpretation: $\tilde{v}=0.75$ so $\tilde{v}\sqrt{b^2+1} > b$). Instead, this program requires $b>\frac{3}{\sqrt{7}}$, and is minimized around $b=3.88$. It would be interesting to automatically infer a value for this parameter that satisfies the constraints and optimizes the complexity bound, but this is future work requiring separate techniques that are outside of the scope of this paper.

\section{Evaluation Programs}
\label{app:evaluation}

This appendix presents the longer programs in the evaluation that were not presented in \Cref{sec:evaluation}.
\Cref{fig:camera} presents the Camera benchmark.
\Cref{fig:quake} presents the EQuake benchmark.
\Cref{fig:jmeint} presents the Jmeint benchmark.

\begin{figure}[b]
  \input{code/camera.tex}
  \phantomcaption
  \label{fig:camera}
\end{figure} %
\begin{figure}
  \ContinuedFloat
  \input{code/camera-2.tex}
  \caption{Camera benchmark, which performs a part of the conversion from blackbody radiator color temperature to the CIE 1931 x,y chromaticity approximation function.}
\end{figure}

\begin{figure}
  \input{code/equake.tex}
  \caption{EQuake benchmark, which computes the displacement of an object after one timestep in an earthquake simulation.}
  \label{fig:quake}
\end{figure}

\begin{figure}
  \input{code/jmeint.tex}
  \phantomcaption
  \label{fig:jmeint}
\end{figure} %
\begin{figure}
  \ContinuedFloat
  \phantomcaption
  \input{code/jmeint-2.tex}
\end{figure} %
\begin{figure}
  \ContinuedFloat
  \input{code/jmeint-3.tex}
\end{figure} %
\begin{figure}
  \ContinuedFloat
  \input{code/jmeint-4.tex}
  \caption{Jmeint benchmark, which calculates whether two 3D triangles intersect, and several auxiliary variables related to their intersection.}
\end{figure}

\clearpage
\section{Extended Renderer Demonstration}
\label{app:renderer}

This appendix provides further details on the evaluation on the renderer program in \Cref{sec:renderer}.

\Cref{listing:renderer-full-learnguage} presents the full code for the renderer program under study.%
\footnote{The original program is written in GLSL. We present a semantically equivalent translation (preserving all paths) of the program to \lang{} for simplicity of presentation.}

\begin{figure}
  \begin{center}
    \hspace*{2.05em}
    \resizebox{!}{0.45\textheight}{%
    \input{code/renderer_learnguage_code.tex}
  }
  \end{center}
  \caption{Full code for the renderer case study.}
  \label{listing:renderer-full-learnguage}
\end{figure}

\Cref{fig:renderer-scenes-app} presents the full set of scenes used in our evaluation.

\begin{figure}
  \begin{center}
    \begin{subfigure}[b]{0.49\textwidth}
      \centering
      \includegraphics[width=\textwidth]{figures/base_day}
      \caption[base-day]{Front-day scene.}
    \end{subfigure}
    \begin{subfigure}[b]{0.49\textwidth}
      \centering
      \includegraphics[width=\textwidth]{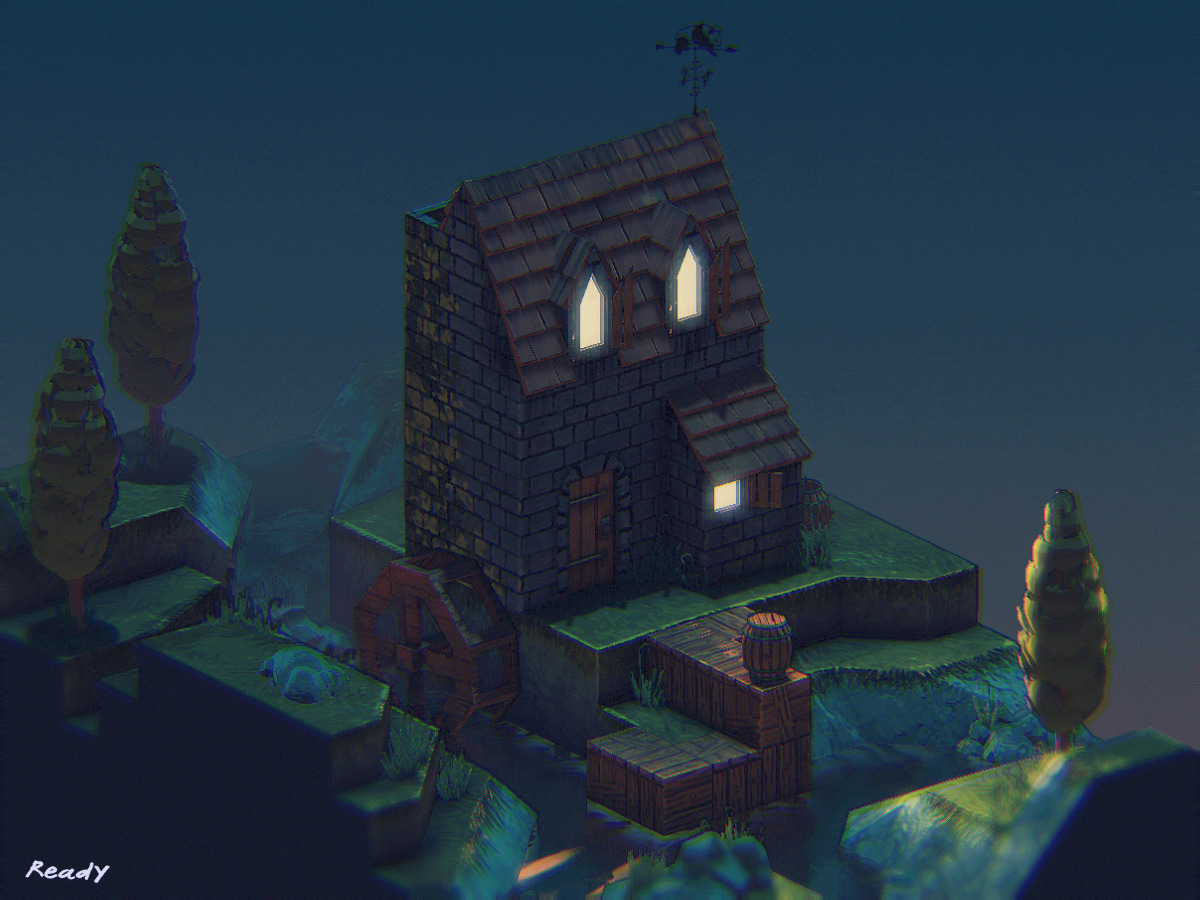}
      \caption[base-night]{Front-night scene.}
      \label{fig:base-night-app}
    \end{subfigure}
    \\
    \begin{subfigure}[b]{0.49\textwidth}
      \centering
      \includegraphics[width=\textwidth]{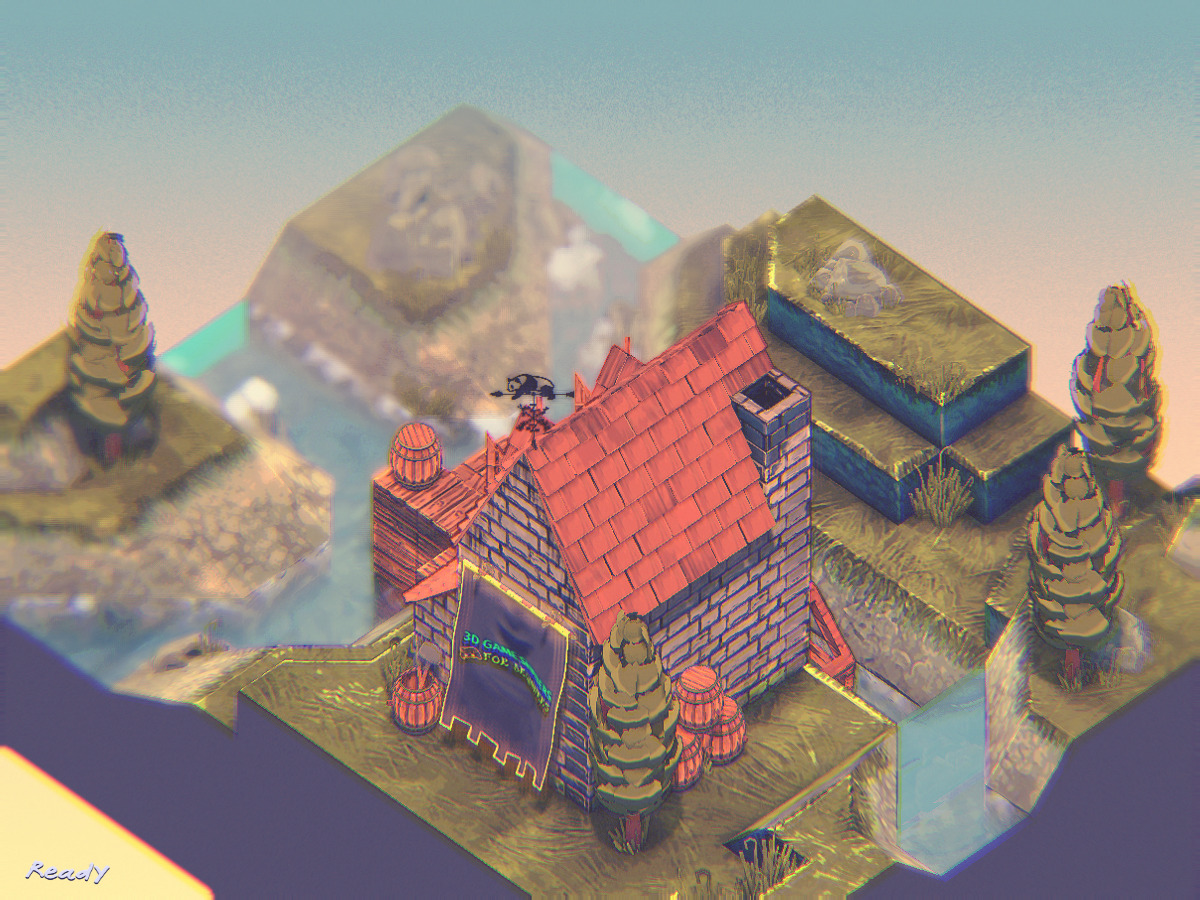}
      \caption[top-day]{Top-day scene.}
      \label{fig:top-day-app}
    \end{subfigure}
    \begin{subfigure}[b]{0.49\textwidth}
      \centering
      \includegraphics[width=\textwidth]{figures/top_night}
      \caption[top-night]{Top-night scene.}
    \end{subfigure}
    \caption[renderer-scenes]{Ground-truth scenes generated by the renderer.}
    \label{fig:renderer-scenes-app}
  \end{center}
\end{figure}

This program is a good candidate to train a surrogate of for several reasons.
First, it is an approximable program: as long as the outputs of a surrogate of the program are sufficiently close to the ground-truth outputs, the generated image will be perceptually indistinguishable.
Second, its paths are all determined by \emph{uniform} input variables, variables in GLSL that are constant across each invocation of the shader.
This means that relative to the cost of executing the program, it is cheap to determine which path a given input induces in the program (and thus which surrogate to apply).
Third, its execution environment is well suited to be replaced with a neural network, since the original program itself performs batch processing on a GPU.

\Cref{fig:base-day-all-app,fig:top-night-all-app} present larger versions of the renders in \Cref{sec:renderer-visualization} for ease of comparison.
The training budgets for these surrogates are 61 samples for the front-day surrogate and 2335 samples for the top-night surrogate.

\begin{figure}
  \begin{center}
    \begin{subfigure}[b]{0.49\textwidth}
      \centering
      \includegraphics[width=\textwidth]{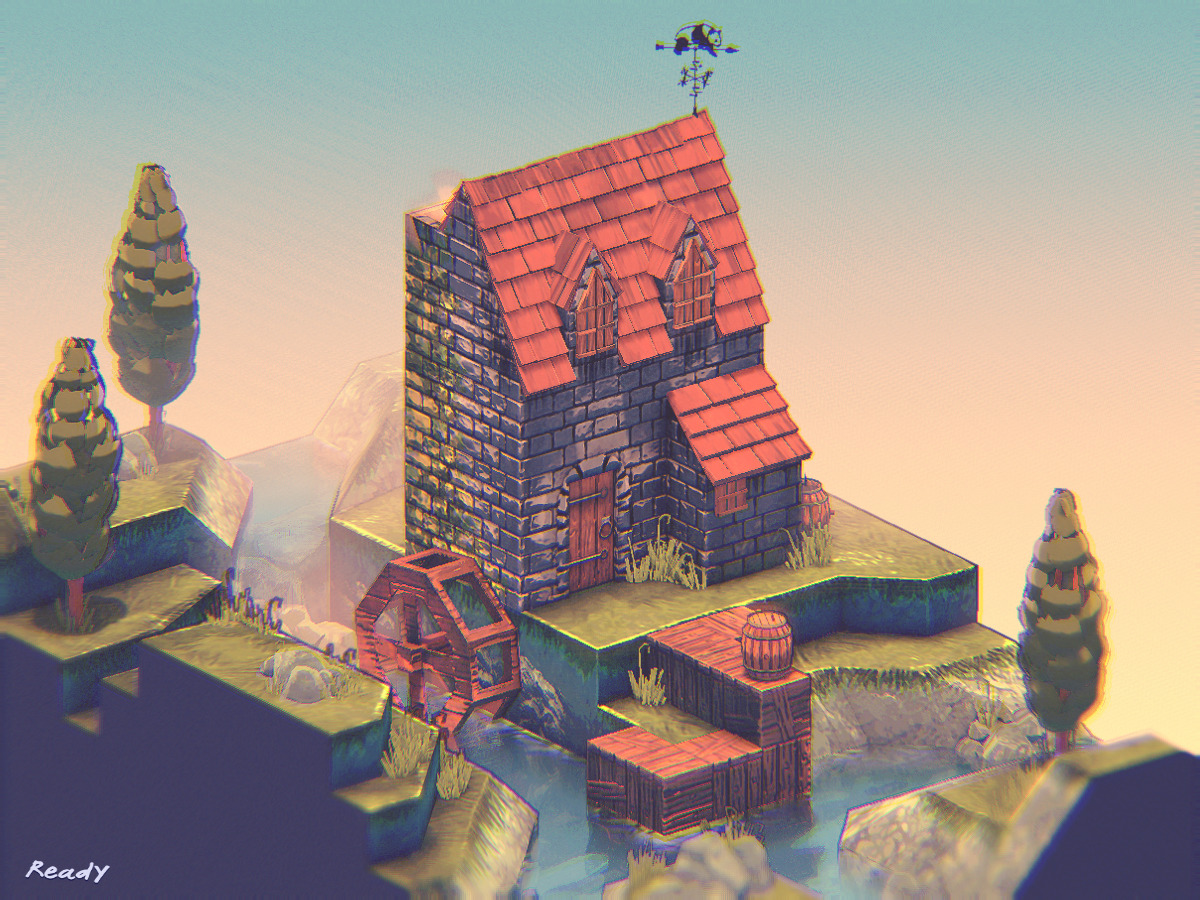}
      \caption[base-day]{Ground truth front-day scene.}
    \end{subfigure}
    \begin{subfigure}[b]{0.49\textwidth}
      \centering
      \includegraphics[width=\textwidth]{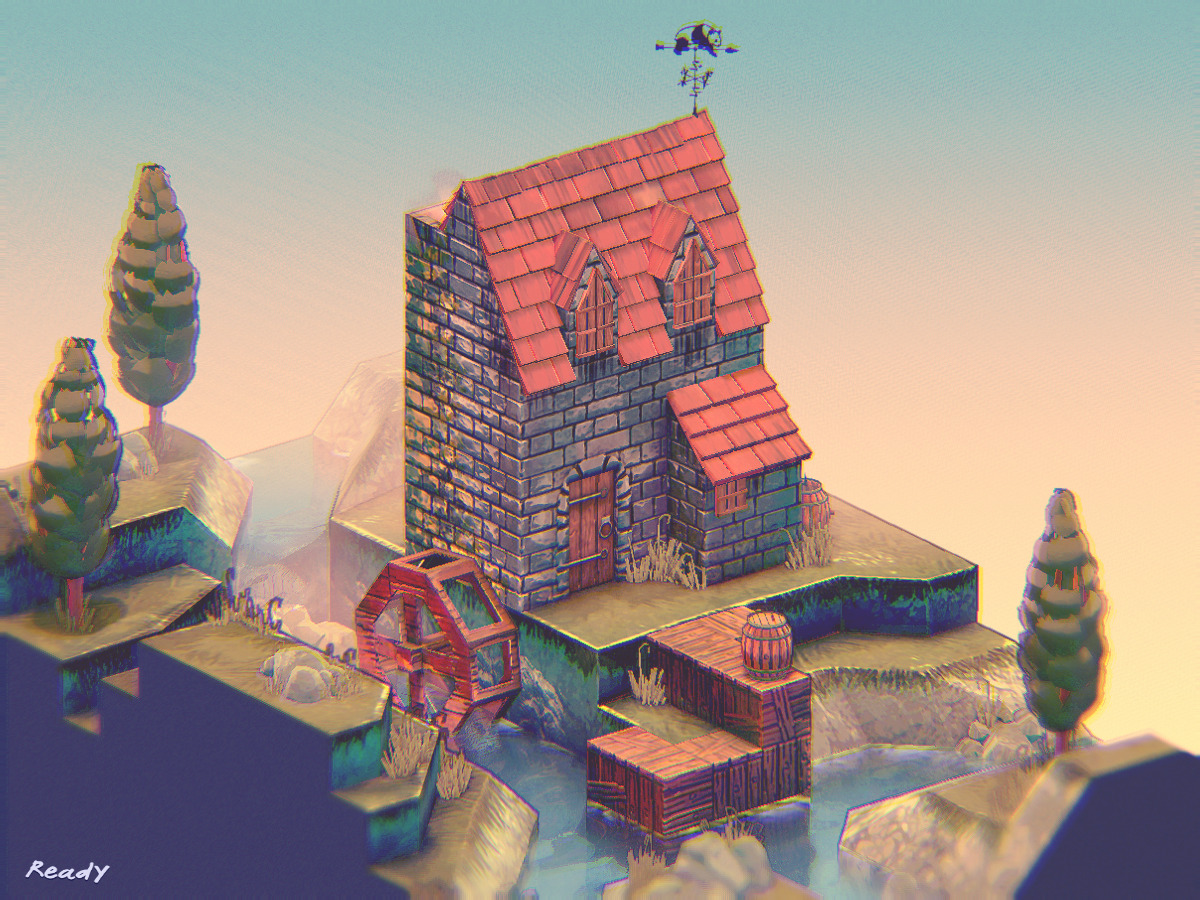}
      \caption[base-night]{\mOptimal{} surrogate front-day scene.}
    \end{subfigure}
    \\
    \begin{subfigure}[b]{0.49\textwidth}
      \centering
      \includegraphics[width=\textwidth]{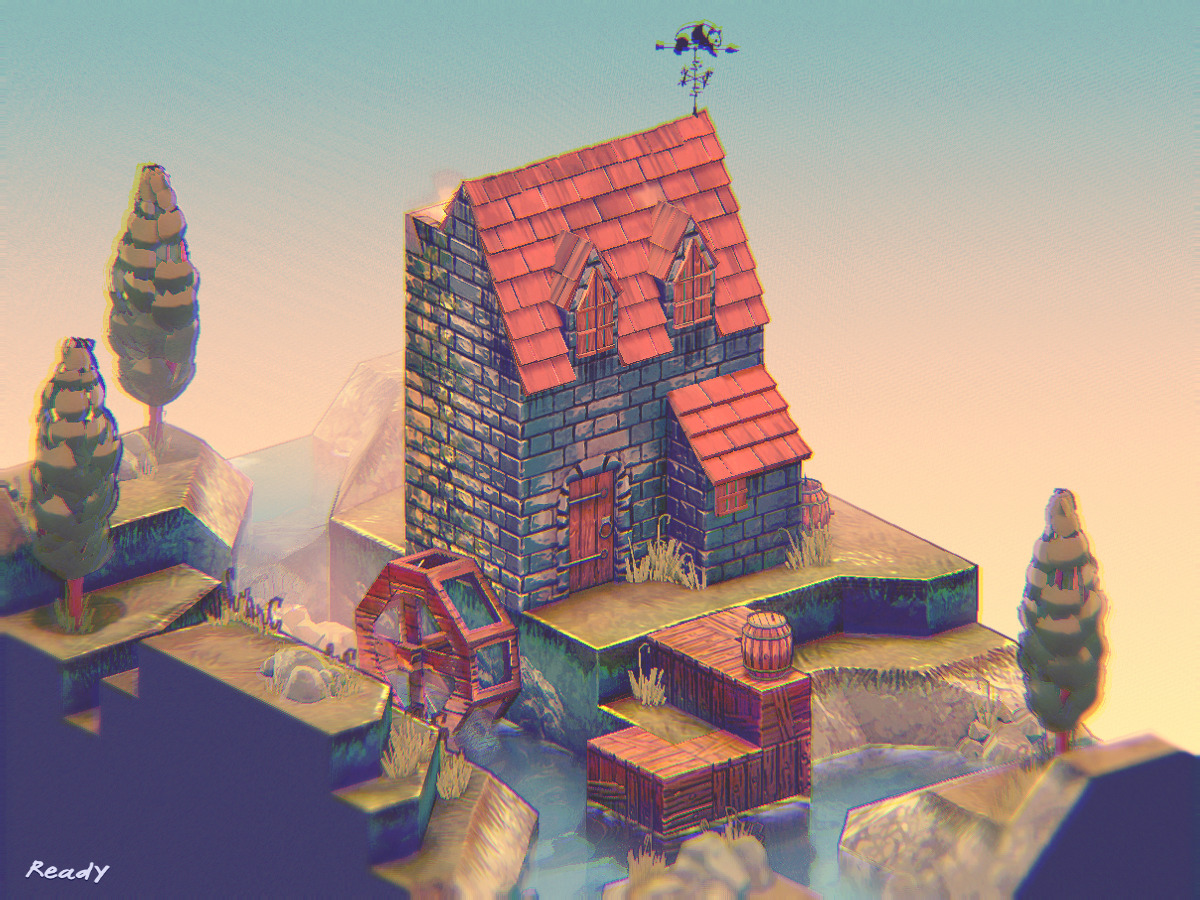}
      \caption[top-day]{Frequency surrogate front-day scene.}
    \end{subfigure}
    \begin{subfigure}[b]{0.49\textwidth}
      \centering
      \includegraphics[width=\textwidth]{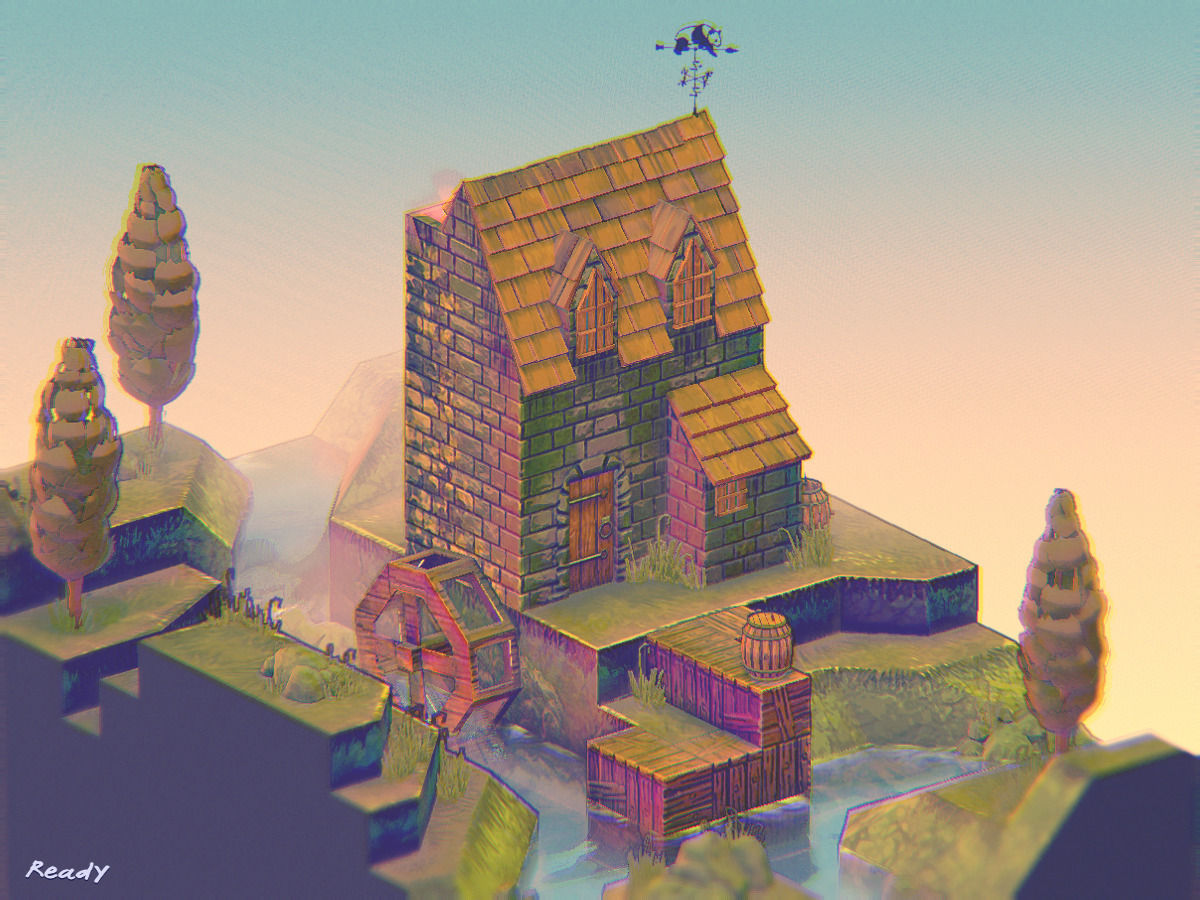}
      \caption[top-night]{Uniform surrogate front-day scene.}
    \end{subfigure}
    \caption[renderer-scenes]{Front-day scene (ground truth in the top left; surrogates in others).}
    \label{fig:base-day-all-app}
  \end{center}
\end{figure}

\begin{figure}
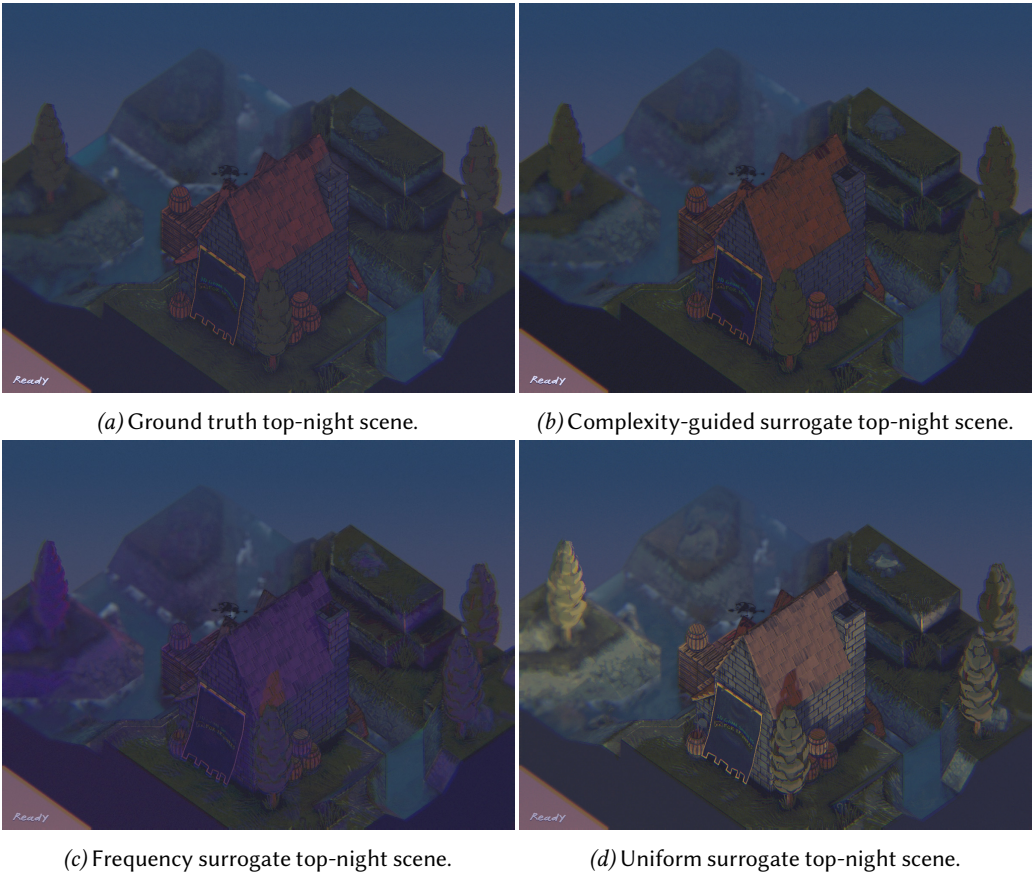

  \begin{center}
    \begin{subfigure}[b]{0.49\textwidth}
      \centering
      \includegraphics[width=\textwidth]{figures/top_night}
      \caption[top-night]{Ground truth top-night scene.}
    \end{subfigure}
    \begin{subfigure}[b]{0.49\textwidth}
      \centering
      \includegraphics[width=\textwidth]{figures/tnightall-optimal-2335-1-s2}
      \caption[top-night]{\mOptimal{} surrogate top-night scene.}
    \end{subfigure}
    \\
    \begin{subfigure}[b]{0.49\textwidth}
      \centering
      \includegraphics[width=\textwidth]{figures/tnightall-test-2335-1-s2}
      \caption[top-night]{Frequency surrogate top-night scene.}
    \end{subfigure}
    \begin{subfigure}[b]{0.49\textwidth}
      \centering
      \includegraphics[width=\textwidth]{figures/tnightall-uniform-2335-3-s2}
      \caption[top-night]{Uniform surrogate top-night scene.}
    \end{subfigure}
    \caption[renderer-scenes]{Top-night scene (ground truth in the top left; surrogates from the All dataset in others).}
    \label{fig:top-night-all-app}
  \end{center}
\end{figure}

\end{document}